\documentclass{ws-ijmpd}
\usepackage{url}
\usepackage[super,compress]{cite}
\usepackage[breaklinks]{hyperref}
\hypersetup{colorlinks,urlcolor=black,citecolor=black,linkcolor=black,filecolor=black}
\usepackage{breakurl}

\begin{document}

\markboth{Alexander Merle}
{keV Neutrino Model Building}

%
\catchline{}{}{}{}{}
%

\title{KEV NEUTRINO MODEL BUILDING}

\author{ALEXANDER MERLE}

\address{Physics and Astronomy, University of Southampton,\\
Highfield, Southampton, SO17 1BJ, United Kingdom\\
A.Merle@soton.ac.uk}

\maketitle

\begin{history}
\received{Day Month Year}
\revised{Day Month Year}
\end{history}

\begin{abstract}
We review the model building aspects for keV sterile neutrinos as Dark Matter candidates. After giving a brief discussion of some cosmological, astrophysical, and experimental aspects, we first discuss the currently known neutrino data and observables. We then explain the purpose and goal of neutrino model building, and review some generic methods used. Afterwards certain aspects specific for keV neutrino model building are discussed, before reviewing the bulk of models in the literature. We try to keep the discussion on a pedagogical level, while nevertheless pointing out some finer details where necessary and useful. Ideally, this review should enable a grad student or an interested colleague from cosmology or astrophysics with some prior experience to start working on the field.
\end{abstract}


\ccode{PACS numbers: 14.60.Pq; 14.60.St; 12.90.+b; 95.35.+d}

\tableofcontents

\section{\label{sec:intro}Introduction}	

The \emph{Dark Matter} problem has been a long standing one in physics: since Oort in 1932 (by studying stellar motions~\cite{FirstOort}) and Zwicky in 1933 (by studying galaxy clusters~\cite{Zwicky:1933gu}) independently concluded its existence, Dark Matter (DM) has puzzled the scientific community. Even though we know that it must exist~\cite{Komatsu:2008hk,Komatsu:2010fb,Hinshaw:2012fq,Ade:2013lta}, we do not know much about its true Nature. While alternative explanations such a Massive Compact Halo Objects (MACHOs)~\cite{Evans:2004gd} or Modified Newtonian Dynamics (MOND)~\cite{Milgrom:1983ca} now seem to be strongly disfavored by the observation of the Bullet Cluster~\cite{Clowe:2006eq}, we tend to believe more and more that DM consists of some new, yet unknown, type of elementary particle~\cite{Bertone}. Depending on its velocity, one classifies particulate DM into \emph{Cold} (CDM), \emph{Warm} (WDM),\footnote{One might also find references to \emph{Cool} (CoolDM) Dark Matter, see \emph{e.g.} Ref.~\refcite{Shi:1998km}.} or \emph{Hot} (HDM) Dark Matter. While HDM is constrained by cosmological structure formation to make up at most about $1\%$ of all DM~\cite{Abazajian:2004zh,dePutter:2012sh}, it is not clear which of the other two possibilities is the correct choice. Even mixed DM scenarios are possible~\cite{Boyanovsky:2007ba,Boyarsky:2008xj}.

An equally puzzling topic in elementary particle physics is the Nature of the neutrino: after its postulate by Pauli~\cite{Neutrino_Wiki} in 1930 and its first detection by Cowan and Reines~\cite{Cowan:1992xc} in 1956, we have gone a long way to measuring all light neutrino mixing angles~\cite{Tortola:2012te,Fogli:2012ua,GonzalezGarcia:2012sz} and narrowing down its mass range to be below the eV-scale~\cite{Komatsu:2008hk,Komatsu:2010fb,Hinshaw:2012fq,Ade:2013lta,Kraus:2004zw,Lobashev:1999tp,KamLANDZen:2012aa,Auger:2012ar,Andreotti:2010vj}. Nevertheless, not only have we not been able to experimentally pin down the actual neutrino mass, its apparent smallness and also the oddly large mixing angles look as if some unknown pattern was behind them. While at the moment we have no fully accepted explanation for this pattern, we have yet been successful in building models that \emph{can} explain the measured patterns and at the same time lead to testable predictions.

The present review deals with a crossover topic that relates both fields described above: \emph{model building for keV sterile neutrino Dark Matter}. In this approach, typically the first-generation sterile neutrino $N_1$ is the lightest and plays the role of the DM particle. The minimal setting yielding such a pattern is the so-called \emph{neutrino-minimal Standard Model} ($\nu$MSM)~\cite{Asaka:2005an}, a framework where the SM is simply extended by three right-handed neutrinos with a very specific mass pattern. This framework can accommodate for a vast variety of phenomena, such as neutrino oscillations~\cite{Boyarsky:2009ix}, the baryon asymmetry of the Universe~\cite{Asaka:2005pn,Shaposhnikov:2008pf}, or the DM problem~\cite{Asaka:2005an}. However, the $\nu$MSM does not contain any \emph{explanation} for the required mass pattern, \emph{i.e.}, it has no mechanism enforcing it. This is the reason why more specific models are required, which contain a suitable mechanism that can to some extent explain the appearance of the required mass pattern. We will review the current models on the market, and explain their respective advantages and disadvantages whenever possible. The key point is that a sterile neutrino with a mass of a few keV could be a very suitable (typically warm) DM candidate. If this is the case, we will have a direct relation between Dark Matter and neutrinos. Connected to that, limits from either sector will constrain the other~\cite{Boyarsky:2006jm}. This raises the question of whether it is actually possible to find complete models that can at the same time lead to a keV sterile neutrino and to successful active neutrino mixing. It will turn out that this is indeed possible. However, due to constraints from the DM sector one has to take into account different and new aspects compared to ordinary model building for light neutrinos. Generically the models reviewed here are quite strongly constrained, since they have to explain data from various sectors. On the other hand, this is exactly what makes keV sterile neutrino DM so interesting: if such a neutrino remains to be a good DM candidate, it will be necessary that astrophysicists, cosmologists, and particle physicists collaborate on the topic in the future, in order to prove or disprove the role of the keV sterile neutrino as DM.

While even the astrophysical and cosmological aspects of keV sterile neutrino DM have hardly been discussed for much more than one decade, the field of model building for keV sterile neutrinos is even younger and less developed, which suggests the question why at all one could need a review on the subject at this stage. However, since the field strongly requires experts from different disciplines to collaborate in order to investigate it thoroughly in all its aspects, it is certainly useful to have a review at hand which explains the goals and methods of one of the related disciplines in a preferably pedagogical manner. While very good overviews of the astrophysical and cosmological aspects of keV sterile neutrinos already exist~\cite{SterShap,Canetti:2012kh}, the present manuscript comprises an attempt to achieve a similar work for the model building aspects of keV neutrinos. Ideally, the text should serve as a travel guide for astrophysicists and cosmologists who have been puzzled for long by the essences of particle physics model building. At the same time the review is supposed to provide the reader with the argumentations that are typically used by model builders when considering certain aspects of a model as ``good'' or ``bad'', or rather ``advantageous'' or ``disadvantageous''. This is not meant in any way as a criticism of certain models, but rather it is supposed to exemplify which arguments are used to get an opinion about a model. It is important to understand that any model will have certain drawbacks, depending on its purpose and perspective. Furthermore, particle physics graduate students should be able to equally profit from the text which gives a pedagogical overview of the matter without going into formal details whenever they \emph{can} be avoided. Finally, even one or the other more experienced researcher in particle theory might be interested in a concise summary of the topic or in a motivation for why to study keV neutrinos. Hopefully these multiple needs are, at least to some extent, covered by the present review.

The text is structured as follows. After shortly discussing some astrophysical and experimental aspects of keV sterile neutrinos in Sec.~\ref{sec:astro}, we review the experimental status and the construction of models for active neutrinos in Sec.~\ref{sec:neutrino}. We then give in Sec.~\ref{sec:keV_general} a brief discussion of some aspects which are specific for keV sterile neutrinos and do not usually play a role for light neutrinos. The central chapter of the review is Sec.~\ref{sec:keV}, where most of the models currently present in the literature are discussed in various detail, depending on feasibility of and suitability for a pedagogical discussion. We finally conclude in Sec.~\ref{sec:conc}. Some technical details on the neutrino mass matrix are discussed in \ref{sec:seesaw}, and a discussion of the calculation of vacuum alignments is given in Sec.~\ref{sec:vacuum}.

Last but not least, it is worth to note that right-handed neutrinos of course do have many other applications, ranging from the use in the seesaw mechanism to leptogenesis, see Ref.~\refcite{Drewes:2013gca} for a recent and concise review of all these multiple aspects. In particular very light sterile neutrinos with a mass around the eV scale~\cite{Abazajian:2012ys,Palazzo:2013me} have attracted considerable attention in the recent years, fueled by several experimental anomalies in accelerator-based experiments (LSND~\cite{Aguilar:2001ty}, KARMEN~\cite{Armbruster:2002mp}, and MiniBooNE~\cite{AguilarArevalo:2012va}, which could not be ruled by a test measurement performed with ICARUS~\cite{Antonello:2012pq}), by re-evaluations of the neutrino flux from nuclear reactors~\cite{Mueller:2011nm,Huber:2011wv} which led to the so-called \emph{reactor anomaly}~\cite{Mention:2011rk}, and by cosmological observations pointing to a non-standard contribution to the content of very light ($\lesssim$eV) particles in the Universe~\cite{Dunkley:2010ge,Keisler:2011aw,Hinshaw:2012fq,Ade:2013lta}. A very detailed analysis of the resulting parameter values from accelerator and reactor experiments has recently been provided in Ref.~\refcite{Kopp:2013vaa}, and first bounds~\cite{Mirizzi:2013kva} on these scenarios using the new data from Planck~\cite{Ade:2013lta} have been derived, too. The big difference of these settings compared to keV neutrinos is that the data, although not entirely conclusive, suggest relatively large [\emph{i.e.}, $\mathcal{O}(0.1)$] mixings between the ordinary active neutrinos and the eV steriles while, as we will see \emph{e.g.} in Secs.~\ref{sec:astro_X-ray} and~\ref{sec:keV_general_X}, settings with keV neutrinos require very small active-sterile mixings, $\mathcal{O}(10^{-5})$. Apart from this difference, several of the models and mechanisms presented in Sec.~\ref{sec:keV} would in principle also be suitable to motivate or even explain sterile neutrinos with eV scale masses. While we do focus on keV neutrinos in this review, we will at some places in the main text nevertheless point out relations or differences to the case of eV sterile neutrinos.

Before starting the discussion, it is worth to point out that the reader should at least be familiar with the basics of the fields. Some knowledge on the Standard Model (SM) of Elementary Particle Physics and on the aspects of quantum field theory required to understand the SM is absolutely necessary to understand the text. Furthermore, some very basic group theory is used without sufficient explanation, while the more complicated aspects are explained in some detail. Finally, since the text focuses on the particle physics aspects of keV sterile neutrinos, an elementary knowledge of cosmology is not strictly necessary but strongly recommended. To finish with a comment on the notation, throughout this review we will make use of natural units: $\hbar = c = 1$.

\section{\label{sec:astro}Astrophysical and Experimental Aspects}	

While the question about the nature of DM is mainly a particle physics problem, many of the properties of DM are known from astrophysics. Although we cannot review all the astrophysical requirements here, we want to at least mention them and give a flavor of their importance. We also give some information on the experimental aspects of keV neutrinos and suggest references for further reading.

\subsection{\label{sec:astro_WCDM}Warm or cold?!?}	

HDM is excluded as the dominant DM component by cosmological structure formation~\cite{Abazajian:2004zh,dePutter:2012sh}, since the corresponding top-down formation (\emph{i.e.}, large structures form first) leads to contradictions with observations. However, the situation of WDM and CDM is not quite as clear. While historically mostly CDM has been considered, due to the variety of possible candidates in theories with, \emph{e.g.}, supersymmetry or extra dimensions~\cite{Bertone}, at the moment we unfortunately do not seem to have any solid direct evidence for such a DM candidate. Even worse, recent experiments such as XENON~\cite{Aprile:2010um,Aprile:2011hi} cut more and more chunks from the allowed parameter space, and some CDM candidates may at some point start being in trouble.

WDM, on the other hand, has also been investigated for some time by the structure formation community: while pioneering studies appeared already more than one decade ago~\cite{Bode:2000gq,Hansen:2001zv}, further collaborations have formed by now which investigate the topic~\cite{Boyarsky:2008xj,Lovell:2011rd}. One particular focus of the simulations is the structure at relatively small scales~\cite{Boyanovsky:2010pw,Boyanovsky:2010sv,VillaescusaNavarro:2010qy,Destri:2012yn,Destri:2013pt}, since this is the main point distinguishing WDM from CDM. Thus smaller galaxies, and in particular the so-called \emph{dwarf satellite galaxies}, are considered to be the key to distinguish the two types of DM. The decisive differences between WDM and CDM arise only at scales of about $0.1$~Mpc, where WDM typically start washing out structures, while CDM only washes out objects of roughly a tenths of that size or below. Since $0.1$~Mpc is just about the size of a typical dwarf galaxy, such objects are in the observational focus, since structure formation on larger scales cannot distinguish between WDM and CDM. Indeed, WDM simulations tend to yield a smaller number of dwarf satellites~\cite{Lovell:2011rd}, which seems to be in coincidence with the fact that we have observed considerably less dwarf satellites in space than predicted by CDM simulation. This is often called the \emph{missing satellite problem}~\cite{Kauffmann:1993gv,Klypin:1999uc,Moore:2005jj} in the literature. However, there could also be astrophysical reasons for this, \emph{e.g.}, it may be that supernova explosions occur in the formation of dwarf satellites at a high enough rate to simply blow away much of the visible material~\cite{Cole:2000ex,Benson:2001au}, such that the dwarf satellite galaxies appear much fainter than larger galaxies. In any case, there is agreement that the core profiles of dwarf satellites are the key to distinguish between CDM and WDM, and hence the dwarf satellites which have already been observed should be investigated in more detail~(see, \emph{e.g.}, Refs.~\refcite{Lovell:2011rd,Destri:2012yn,BoylanKolchin:2011dk,deVega:2009ky,deVega:2010yk}). Note also that, on the more observational side, model-independent surveys~\cite{Papastergis:2011xe} and data analyses~\cite{deVega:2009ku} seem to point towards a keV-mass DM particle.

To give a flavor of the lower bounds one can obtain on the keV neutrino mass, the most model-independent limit one can derive originates from the Tremaine-Gunn bound~\cite{Tremaine:1979we}. The principle idea is that the averaged phase base density of a fermionic DM candidates cannot be smaller that of a degenerate Fermi gas. When applied to the case of keV sterile neutrinos, the typical lower bound one obtains on the DM mass is around $M \gtrsim 1$~keV~\cite{Boyarsky:2008ju,Gorbunov:2008ka}. Stronger limits can be obtained when the production mechanism is taken into account, which results in different lower bounds ranging roughly from $1.6$ to $10$~keV, see \emph{e.g.} Refs.~\refcite{Abazajian:2012ys,Bezrukov:2009th,Boyarsky:2008mt,Kusenko:2006rh}.

Let us end this section by stressing that, although often used more or less equivalently in the literature, \emph{a keV mass of the DM particle does not necessarily mean that it is WDM}. The decisive point is the velocity profile rather than the mere mass (\emph{e.g.}, axions have tiny masses but are typically CDM~\cite{Dine:1982ah,Preskill:1982cy,Abbott:1982af,Stecker:1982ws}). Nevertheless, it is true in many cases that keV mass particles, and in particular keV sterile neutrinos, indeed turn out to be WDM. But the exact velocity profile depends on the production mechanism under consideration, and sometimes on subsequent effects modifying the expansion rate and hence the cooling of (parts of) the Universe.

\subsection{\label{sec:astro_prod}Production mechanisms}	

A necessary condition for any DM candidate particle is to participate in a mechanism which can produce enough of this particle to make up all or at least a significant part of the amount of the DM in the Universe. This is non-trivial, since the requirement of having not too much and not too little DM can often only be fulfilled in a very narrow parameter range. The most generic DM production mechanism is the so-called \emph{thermal freeze-out}, which is discussed in many textbooks (see, \emph{e.g.}, Ref.~\refcite{KolbTurner}). However, for this production mechanism to be effective a particle needs an interaction strength comparable to that of weak interactions, which is generically difficult for a \emph{sterile} neutrino. On the other hand, if the SM gauge group is extended, thermal freeze-out might be revived, but this will lead to other difficulties. An alternative thermal production mechanism is \emph{freeze-in} production~\cite{Hall:2009bx}, which is well-suited for extremely weakly interacting particles. Finally, non-thermal production is also a valid possibility.

We will now briefly discuss the most generic mechanisms used to produce keV sterile neutrino DM.

\subsubsection{\label{sec:astro_prod_DW}The Dodelson-Widrow mechanism}	

The easiest and most natural possibility is the so-called \emph{Dodelson-Widrow} (DW) mechanism~\cite{Dodelson:1993je}. The idea, based on Refs.~\refcite{Dolgov:1980cq,Manohar:1986gj,Barbieri:1989ti,Barbieri:1990vx,Enqvist:1990dq,Enqvist:1990ek,Enqvist:1991qj,Cline:1991zb}, is that although the interaction strength of keV sterile neutrinos is too small for them to be in thermal equilibrium with the plasma, they could nevertheless be produced by the plasma. This is because, as we will see in Sec.~\ref{sec:neutrino_modeling_mass}, sterile neutrinos are typically not absolutely sterile but instead have tiny admixtures to active neutrinos. Although these admixtures are not large enough to keep the keV sterile neutrinos in thermal equilibrium with the plasma, from time to time they are produced in the plasma by processes which mainly generate active neutrinos, but which can also produce steriles by their admixtures. Thus a certain amount of keV sterile neutrinos gradually builds up, and the corresponding interaction strengths and densities are small enough so that the keV neutrinos do not annihilate again. This is very similar to what happens in freeze-in~\cite{Hall:2009bx}, where more general DM particles (or particles which later on decay into DM particles) have so feeble interactions with the SM that they never enter thermal equilibrium, but are only produced from time to time in the early Universe, thereby gradually building up a significant abundance. While the details of the DW mechanism can be quite tricky (\emph{e.g.} due to hadronic contributions to the sterile neutrino production~\cite{Asaka:2006rw}), the general idea is nevertheless very easy to understand.

Despite its simplicity, this mechanism is by now known to be excluded by the strong X-ray bound (cf.\ Secs.~\ref{sec:astro_X-ray} and~\ref{sec:keV_general_X}) in the minimal version, \emph{i.e.}, if no primordial lepton asymmetry, which must arise from further new physics, is present~\cite{Canetti:2012vf,Canetti:2012kh}. Hence, the DW mechanism alone is not enough to produce a sufficient amount of DM. On the other hand, the DW contribution to the DM relic density is \emph{unavoidable} as long as the active-sterile mixing is not completely switched off. Thus, many other production mechanisms should actually be accompanied by an additional amount of keV sterile neutrino production by the DW mechanism.

\subsubsection{\label{sec:astro_prod_SF}The Shi-Fuller mechanism}	

An alternative possibility is the so-called \emph{Shi-Fuller} (SF) mechanism~\cite{Shi:1998km}, which can arise in addition to the DW mechanism. The basic idea is that a significant primordial lepton-antilepton asymmetry in the early Universe can act onto active-sterile neutrino transitions in a way similar to a background made of ordinary matter acts onto ordinary neutrino oscillations. In such an environment it is well known that \emph{resonant} flavor transitions can occur in the presence of a certain matter density, which is known as the \emph{Mikheev-Smirnov-Wolfenstein} (MSW) effect~\cite{Wolfenstein:1977ue,Wolfenstein:1979ni,Mikheev:1986gs,Mikheev:1986wj,Mikheev:1986if}. Such transitions can, in the early Universe, produce a large amount of keV sterile neutrinos at a very specific temperature, in case that a suitable lepton asymmetry is present~\cite{Laine:2008pg}.

While this mechanism cannot stand alone, as it will be necessarily accompanied by a DW contribution, its great benefit lies in the fact that the neutrino spectrum is changed: the non-thermal component produced by the resonant transition leads to an overall cooler DM spectrum than for pure DW production~\cite{Shi:1998km}. This makes the keV sterile neutrinos in the Universe more similar to ordinary CDM, which helps to avoid problems associated with a too warm spectrum, \emph{i.e.}, the washout of structures larger than dwarf satellites, which is not observed.

\subsubsection{\label{sec:astro_prod_Scal}Scalar decays}	

Another option for producing keV sterile neutrinos non-thermally is by the decay of scalar singlets~\cite{SterShap,Kusenko:2006rh,Petraki:2007gq} which could, \emph{e.g.}, freeze-out by Higgs portal interactions and decay at a temperature similar to their mass (or, rather, at their freeze-out temperature which is in turn similar to their mass). One case that has been particularly studied is the one of this scalar particle being the inflaton~\cite{Shaposhnikov:2006xi,Bezrukov:2009yw,Boyanovsky:2008nc}. This possibility is particularly attractive from the point of view of solving the additional problem of inflation. On the other hand, a general singlet scalar field is hardly constrained from the particle physics side in case that it only has feeble interactions with the SM Higgs sector. Moreover, the mechanism of producing keV sterile neutrinos from scalar decays also leads to a somewhat cooler DM spectrum if the scalar is a frozen-out relic~\cite{Kusenko:2006rh}, which can help to evade certain bounds.

In addition it is worth to stress the versatility of scenarios with scalar production: not only do they have an immediate connection to the scalar sector of the SM and theories beyond, potentially leading to interesting collider signatures~\cite{Shoemaker:2010fg}, a non-zero vacuum expectation value of the scalar field can also immediately lead to Majorana neutrino masses~\cite{Bezrukov:2009yw} and to other lepton number violation phenomenology~\cite{Helo:2010cw}. From the particle physics point of view, this production mechanism might be one of the most interesting ones.

\subsubsection{\label{sec:astro_prod_Therm}Thermal overproduction with subsequent entropy dilution}	

Finally, it could also be that the keV sterile neutrinos are only sterile with respect to SM interactions, but are non-trivially charged under the actual (larger) gauge group which breaks down to the SM at a low enough temperature. The prime example of such a situation is Left-Right symmetry~\cite{Beg:1977ti,Pati:1977jh,Mohapatra:1974hk,Mohapatra:1974gc,Senjanovic:1975rk}, $SU(3)_C \times SU(2)_L \times SU(2)_R \times U(1)_{B-L}$, where all right-handed fermions including neutrinos are singlets under $SU(2)_L$, but charged non-trivially under $SU(2)_R$. Hence, at temperatures above the breaking $SU(3)_C \times SU(2)_L \times SU(2)_R \times U(1)_{B-L} \to SU(3)_C \times SU(2)_L \times U(1)_Y$, the keV ``sterile'' neutrino can actually be in thermal equilibrium with the thermal plasma in the early Universe, and at some point it can freeze-out like nearly all the other particles.

This idea has been applied to the case of keV sterile neutrinos in Refs.~\refcite{Bezrukov:2009th,Nemevsek:2012cd,Asaka:2006ek,Asaka:2006nq}. However, it turns out that, just as for the freeze-out of SM neutrinos if they had masses of a few keV, the keV steriles would freeze-out at a temperature much larger than their mass, and hence their abundance would not be suppressed. Accordingly, they would overclose the Universe by comprising far too much DM. While this is seemingly a big problem, it can be compensated to some extent by a subsequent dilution by entropy production~\cite{Scherrer:1984fd} from the decay of a frozen-out non-relativistic species that temporarily dominates the energy density of the Universe. Typically, the second and third generation sterile neutrinos $N_{2,3}$ do this job, as their decays are generically much faster than those of the $N_1$ due to their larger mass, cf.\ Sec.~\ref{sec:astro_X-ray}. Indeed, such a mechanism can lead to consistency with most bounds~\cite{Bezrukov:2009th,Nemevsek:2012cd}, but one has to mention that such settings are nevertheless under a certain amount of tension when all the data is put together. One particularly hard bound is the lower value of the ``reheating'' temperature before the event of Big Bang Nucleosynthesis~\cite{Kawasaki:2000en,Hannestad:2004px}, which can drastically reduce the allowed parameter space even for relatively general settings~\cite{King:2012wg}.

\subsection{\label{sec:astro_X-ray}The X-ray bound}	

Probably the most important bound in practice on models with keV sterile neutrinos is their radiative decay, $N_1 \to \nu \gamma$~\cite{Pal:1981rm,Barger:1995ty}. This decay proceeds via a 1-loop diagram, which looks very similar to typical 1-loop diagrams for the lepton flavor violating decay $\mu \to e \gamma$~\cite{Cheng:1977nv,Petcov:1976ff,Marciano:1977wx,Lee:1977tib}. As generic for such a decay, if the final state lepton is practically massless, the decay rate is proportinal to the fifth power of the initial state mass~\cite{Lavoura:2003xp}. Accordingly, a decay $N_i \to \nu \gamma$ is proportional to $M_i^5$, where $M_i$ is the mass of $N_i$. If out of three sterile neutrinos the first generation fermion $N_1$ is the keV neutrino, while $N_{2,3}$ are considerably heavier, this implies that the two heavy neutrinos decay extremely quickly, whereas the keV neutrino is practically stable even on cosmological scales. Otherwise it could not play the role of DM.

Nevertheless, if enough of the keV neutrinos are present in the Universe, some of them will decay and produce a certain amount of photons. Since the light neutrino mass is practically zero compared to the mass of the keV neutrino, the photon resulting from the decay is practically monoenergetic. Accordingly, the smoking gun signature to search for is a monoenergetic X-ray line which can be searched for by satellite experiments, see \emph{e.g.} Refs.~\refcite{Dolgov:2000ew,Abazajian:2001vt,Boyarsky:2005us,Boyarsky:2006fg,RiemerSorensen:2006fh,Abazajian:2006yn,Watson:2006qb,Boyarsky:2006ag,Abazajian:2006jc,Boyarsky:2007ay,Boyarsky:2007ge,Loewenstein:2008yi,Watson:2011dw,Loewenstein:2012px}.

Turning the argumentation round, the non-observation of this X-ray line produces strong bounds on the active-sterile mixing and hence on the full neutrino mass matrix. However, we will postpone the discussion of this constraint to the dedicated Sec.~\ref{sec:keV_general_X}. The reason is that we will first need to introduce some more aspects of neutrino physics before we can appreciate the effect of this bound on a concrete model, which will be done in Sec.~\ref{sec:neutrino}.

\subsection{\label{sec:astro_exp}Other experimental aspects keV sterile neutrinos}	

Apart from astrophysical signals, one could ask the question whether there are alternative experiments, ideally in a laboratory, which could have the potential to detect signals related to keV sterile neutrinos. In general, one would search for reactions in which typically an ordinary neutrino is produced, but which could at least in principle also produce sterile neutrinos by their admixtures to the active sector. A generic example for such a reaction would be a nuclear $\beta$ decay, \emph{e.g.} $(Z, A) \to (Z-1, A) + e^+ + \nu_e$, where the electron neutrino $\nu_e$ is a superposition of all kinematically accessible mass eigenstates. In case the $Q$-value of the transition (\emph{i.e.}, the energy release) is larger than the mass $M_1$ of the keV neutrino, a certain fraction of such decays will produce the keV mass eigenstate. However, this fraction is proportional to the (tiny) active-sterile mixing angle squared, which is the biggest challenge for the experiments. Nevertheless more detailed studies of such settings~\cite{Liao:2010yx,deVega:2011xh} have revealed that at least in the future a possible detection is thinkable, where in particular the improvements in the technology of ion traps may become interesting~\cite{Bezrukov:2006cy}. Alternative reactions which could be interesting are electron captures, where a keV neutrino could either be (resonantly) absorbed or simply produced in the final state~\cite{Li:2011mw,Li:2010vy}. Interestingly, experimental studies are done for the cases of tritium beta decay (KATRIN~\cite{Osipowicz:2001sq} and Project~8~\cite{Monreal:2009za} collaborations), of rhenium beta decay (MARE collaboration~\cite{Nucciotti:2010tx}), and of electron capture on Holmium (MARE and ECHo collaborations~\cite{ECHo}). A concise summary with references to several useful talks can be found in Ref.~\refcite{deVega:2013ysa}.

Other experimental searches are less promising: for example, one could think of detecting matter effects related to keV sterile neutrinos if neutrinos propagate \emph{e.g.} through the Earth over macroscopic distances. However, these matter interactions can be computed and their effects turn out to be tiny~\cite{Ando:2010ye}. Furthermore, one could think of contributions to neutrino-less double beta decay (to be discussed in more detail in Sec.~\ref{sec:neutrino_mix_masses}), where a nucleus decays into another one via the emission of only two electrons, $(Z,A) \to (Z+2,A) + e^- + e^-$. However, also there the corresponding effects are hardly visible within the experimental accuracy, again due to the strong observational bound on active-sterile mixing~\cite{Bezrukov:2005mx,Asaka:2011pb,keV0nbb}. Finally, one could also think of detecting the heavier brothers $N_{2,3}$ of the keV neutrino in $pp$-collisions and meson decays~\cite{Gorbunov:2007ak}, which is at least not hopeless.

Apart from these searches, however, there is little hope to experimentally detect keV sterile neutrinos, which is one of the reasons why it is so beneficial to have more concrete models, cf.\ Sec.~\ref{sec:keV}, which entangle keV sterile neutrino DM with the light neutrino sector. This can lead to very concrete predictions of the light neutrino mass spectrum and/or mixing pattern. By this the models to be discussed here offer a considerably increased potential for testability by ordinary neutrino oscillation experiments, cf.\ Sec.~\ref{sec:neutrino_mix_osc}. This is why the models presented are very interesting from the experimental point of view, too.

\subsection{\label{sec:astro_other}Other aspects of keV sterile neutrinos}	

There are many further interesting aspects of keV sterile neutrinos, which are too numerous to all be mentioned here. However, to give at least an idea of the many existing interrelations, we mention a few.

For example, there are astrophysical implications different from structure formation or DM production, where keV sterile neutrinos could be interesting: for example, they could help to explain the observed motions of pulsars (``pulsar kicks'')~\cite{Kusenko:1997sp,Fuller:2003gy,Kusenko:2006rh,Kusenko:2008gh}, they could influence star formation~\cite{Stasielak:2006br}, or their decays could lead to the production of X-rays causing a partial reionization of the Universe~\cite{Biermann:2006bu}. Note that at least the latter two processes rely on a significant amount of keV sterile neutrinos being present in the Universe, such that the amount X-rays originating from their decays is at least non-negligible. However, this does not mean that they have to make up all of the DM observed, but only a significant fraction

In general, the topic of keV sterile neutrino Dark Matter exhibits considerable interrelations with other sectors in particle and/or astrophysics, which is why it makes sense to study it in some more detail.

\section{\label{sec:neutrino}Neutrino Model Building}	

In this section, we will review the model building aspects of \emph{light}, \emph{i.e.}, eV-scale \emph{active} neutrinos. Although the main focus of this review is on keV-scale sterile neutrinos, reproducing the correct low energy neutrino data is a necessary requirement for any model to be built. Furthermore, the basic principles of model building can be nicely illustrated for light neutrinos.

\subsection{\label{sec:neutrino_mix}Neutrino masses and mixings}	

After having been theoretically predicted by Pauli~\cite{Neutrino_Wiki} in 1930 to rescue the conservation laws in nuclear beta decay, the neutrino was detected for the first time in the famous Cowan-Reines experiment~\cite{Cowan:1992xc}. At that time, the neutrino was assumed to be essentially massless. On the other hand, a small mass would not have had a big effect in any experiment, and could hence not be excluded. However, as proposed by Pontecorvo first for meson-oscillations~\cite{Pontecorvo:1957cp} and later for neutrino-antineutrino systems~\cite{Pontecorvo:1957qd}, even a tiny neutrino mass can still have an observable effect if the neutrino mass basis is not equal to the flavor basis. The resulting phenomenon is known as \emph{neutrino oscillations}, and it essentially means that a neutrino produced in one definite flavor, say $\nu_e$, can change its identity while traveling from one place to another, to be finally detected as, \emph{e.g.}, $\nu_\mu$: neutrinos are \emph{mixed}. An excellent although slightly outdated review on this subject can be found in Ref.~\refcite{Barger:2003qi}. Note that, since neutrino oscillation probabilities only depend on \emph{mass square differences} $\Delta m_{ij}^2 \equiv m_i^2 - m_j^2$, neutrino oscillations cannot give any information about the actual scale of neutrino masses. However, the observation of at least one non-zero $\Delta m_{ij}^2$, \emph{i.e.} of neutrino oscillations, immediately implies that not all neutrinos can be massless.

\subsubsection{\label{sec:neutrino_mix_osc}Neutrino mixing and neutrino oscillations}	

What neutrino mixing means in practice is that, \emph{e.g.}, an electron neutrino $\nu_e$ (\emph{i.e.}, the neutrino which couples to the electron via $W$-bosons) does not have a definite ``electron neutrino mass'', but it is rather a superposition of (usually three) mass eigenstates $\nu_i$, each of which has a well-defined and definite mass $m_i$. This implies that a neutrino produced with a definite flavor will change its identity when propagating, since the different masses of the different mass eigenstate components will lead to different quantum mechanical phases. In the two-flavor approximation, \emph{i.e.}, when having two weak interaction eigenstates $(\nu_e, \nu_\mu)$ and two mass eigenstates $(\nu_1, \nu_2)$, the corresponding oscillation formula for a neutrino produced as a pure $\nu_e$ is given by~\cite{Barger:2003qi}:
\begin{equation}
 P(\nu_e \to \nu_\mu, L) = \sin^2 (2 \theta) \sin^2 \left( \frac{L\ \Delta m^2}{4 E} \right),
 \label{eq:2-flav-osc}
\end{equation}
where $E$ is the neutrino energy, $L$ is the distance between production and detection of the neutrino, and $\Delta m^2 = m_2^2 - m_1^2$. Indeed, if $m_1 = m_2$ (and hence in particular if $m_{1,2}=0$), the transition probability from Eq.~\eqref{eq:2-flav-osc} is zero. Furthermore, if the \emph{mixing angle} $\theta$ which describes the mismatch between the two bases is zero, then we also have $P\equiv 0$. In that case, the flavor would completely determine the mass and any quantum mechanical uncertainty would vanish, thereby killing the interference-like oscillation. More information on the correct treatment of neutrino oscillations in the context of quantum mechanics and quantum field theory can, e.g., be found in Refs.~\refcite{Giunti:1991ca,Giunti:1993se,Giunti:1997wq,Beuthe:2001rc,Kienert:2008nz,Akhmedov:2009rb,Merle:2009re,Merle:2010qq,Akhmedov:2010ms,Akhmedov:2010ua,Kayser:2010bj,Akhmedov:2012uu}.

Formally, 3-flavor neutrino mixing can be described by a rotation of the basis, and the corresponding mixing matrix~\cite{Maki:1962mu} is known as the Pontecorvo-Maki-Nakagawa-Sakata (PMNS) matrix $U\equiv U_{\rm PMNS}$:
\begin{equation}
 \begin{pmatrix}
 \nu_e\\
 \nu_\mu\\
 \nu_\tau
 \end{pmatrix} = U \begin{pmatrix}
 \nu_1\\
 \nu_2\\
 \nu_3
 \end{pmatrix},
 \label{eq:PMNS_def}
\end{equation}
or $\nu_\alpha = U_{\alpha i} \nu_i$ for short, in a basis where the charged lepton mass matrix is diagonal. Explicitly, the PMNS matrix can be written in terms of three mixing angles $(\theta_{12}, \theta_{13}, \theta_{23})$, one \emph{Dirac} $CP$-violating phase $\delta$, and two \emph{Majorana} $CP$-violating phases $(\alpha, \beta)$, the latter two having trivial values $0$ or $\pi$ in the case of Dirac neutrinos:
\begin{equation}
  U = 
  \begin{pmatrix} c_{12} c_{13} & s_{12} c_{13} & s_{13} \, e^{-i \delta}  \\
  -s_{12} c_{23} - c_{12} s_{23} s_{13} e^{i \delta} & c_{12} c_{23} - s_{12} s_{23} s_{13} e^{i \delta} & s_{23} c_{13}  \\ 
  s_{12} s_{23} - c_{12} c_{23} s_{13} e^{i \delta} & - c_{12} s_{23} - s_{12} c_{23} s_{13} e^{i \delta} & c_{23} c_{13}
  \end{pmatrix}
 {\rm diag}(1, e^{i \alpha}, e^{i (\beta + \delta)}),
 \label{eq:PMNS_expl}
\end{equation}
where $c_{ij}\equiv \cos \theta_{ij}$ and $s_{ij}\equiv \sin \theta_{ij}$. If these mixing angles were small, we could approximate the flavor by the mass eigenstates, similar to the quark sector: there, a weak interaction eigenstate down-quark $d'$ is not very different from the corresponding strong interaction (and mass) eigenstate $d$, a fact that is reflected in the quark mixing angles being relatively small~\cite{Beringer:1900zz}. However, experimentally -- and indeed surprisingly at first sight -- the neutrino mixing angles have been measured to partially be very large~\cite{Tortola:2012te}: for a normal mass ordering, $m_1 < m_2 < m_3$, the mixing angles are given by
\begin{eqnarray}
 \sin^2\theta_{12} &=& 0.320\ \ \text{(0.303--0.336, 0.27--0.37)},\nonumber\\
 \sin^2\theta_{13} &=& 0.0246\ \ \text{(0.0218--0.0275, 0.017--0.033)},\nonumber\\
 \sin^2\theta_{23} &=& 0.613\ \ \text{(0.400--0.461 \& 0.573--0.635, 0.36--0.68)},
 \label{eq:angles_NH}
\end{eqnarray}
while for inverted ordering, $m_3 < m_1 < m_2$, we have
\begin{eqnarray}
 \sin^2\theta_{12} &=& 0.320\ \ \text{(0.303--0.336, 0.27--0.37)},\nonumber\\
 \sin^2\theta_{13} &=& 0.0250\ \ \text{(0.0223--0.0276, 0.017--0.033)},\nonumber\\
 \sin^2\theta_{23} &=& 0.600\ \ \text{(0.569--0.626, 0.37--0.67)},
 \label{eq:angles_IH}
\end{eqnarray}
where we have quoted the best-fit value (1$\sigma$ region, $3\sigma$ region). Note that the global fit data quoted in Eqs.~\eqref{eq:angles_NH} and~\eqref{eq:angles_IH} includes the very recent (2012) measurements of the previously unknown mixing angle $\theta_{13}$ by the Daya Bay~\cite{An:2012eh}, RENO~\cite{Ahn:2012nd}, and Double Chooz~\cite{Abe:2011fz,Abe:2012tg} collaborations. Other mass orderings different from normal and inverted ordering are not possible, since the MSW effect~\cite{Wolfenstein:1977ue,Wolfenstein:1979ni,Mikheev:1986gs,Mikheev:1986wj,Mikheev:1986if} for solar neutrinos~\cite{Cleveland:1998nv,Abdurashitov:2009tn,Hosaka:2005um,Cravens:2008aa,Abe:2010hy,Aharmim:2008kc,Aharmim:2009gd}, supported by reactor neutrino data~\cite{Kaether:2010ag,Arpesella:2008mt}, forces the mass square difference $\Delta m_{21}^2 = m_2^2 - m_1^2$ to be positive, \emph{i.e.}, $\Delta m_{21}^2 = 7.62\ \ \text{(7.43--7.81, 7.12--8.20)}\cdot 10^{-5}~{\rm eV}^2$~\cite{Tortola:2012te}.

Finally, measurements of atmospheric neutrinos~\cite{Wendell:2010md} together with experiments detecting neutrinos that are artificially produced in accelerators~\cite{Abe:2011sj,Adamson:2011qu,Adamson:2011ig,Adamson:2011fa,Adamson:2012rm} can be used to obtain $|\Delta m_{31}^2| = 2.55\ \ \text{(2.46--2.61, 2.31--2.74)}\cdot 10^{-3}~{\rm eV}^2$ (normal ordering) and $|\Delta m_{31}^2| = 2.43\ \ \text{(2.37--2.50, 2.21--2.64)}\cdot 10^{-3}~{\rm eV}^2$ (inverted ordering)~\cite{Tortola:2012te}. At the moment, we do not know the sign of $\Delta m_{31}^2$, \emph{i.e.}, if $m_1 < m_3$ or $m_1 > m_3$, which means that the mass ordering has not yet been determined.

Note that we could, of course, have used any of the other currently available global fits to neutrino mixing data~\cite{Fogli:2012ua,GonzalezGarcia:2012sz}, but for definiteness we had to decide for one and have taken Ref.~\refcite{Tortola:2012te}.

\subsubsection{\label{sec:neutrino_mix_masses}Limits on neutrino masses}	

Even though we have considerable knowledge about the two mass square differences $\Delta m_{21}^2$ and $|\Delta m_{31}^2|$, this is not enough to draw conclusions about the absolute neutrino mass scale. We do, however, have experimental bounds on certain observables. At the moment, we know more or less three realistic probes of the absolute neutrino mass scale, which are its kinematical determination in laboratory experiments on single beta decay, the measurement of the effective mass in neutrino-less double beta decay, and the cosmological determination of the sum of all neutrino masses. Although several present and future neutrino detectors might also be interesting in the context of supernova neutrinos (see, \emph{e.g.}, Refs.~\refcite{Abe:2011ts,Wurm:2011zn}), they necessarily would need a supernova to happen in our Galaxy in the first place, which is not guaranteed. Alternative proposals to get information on the absolute neutrino mass (\emph{e.g.}, Refs.~\refcite{Lindroos:2009mx,Kopp:2009yp,Yoshimura:2006nd,Fukumi:2012rn,Dinh:2012qb}) seem even less realistic.

In single beta decays, one tries to use nuclear reactions such as $(Z,A) \to (Z+1,A) + e^- + \bar{\nu}_e$ to determine the neutrino mass, which is the simplest and most model-independent way to get information on the neutrino mass. The trick is to use only the small fraction of all decays in which the electron carries away nearly all the kinetic energy, \emph{i.e.}, the emitted neutrino is essentially at rest. Since current and near future experiments are not sensitive enough to resolve the three mass eigenvalues separately, the corresponding observable arising from an \emph{incoherent sum} over all possible final states, and it is called the \emph{effective electron neutrino mass} $m_\beta$, defined by~\cite{Aalseth:2004hb}
\begin{equation}
 m_\beta^2 \equiv m_1^2 c_{12}^2 c_{13}^2 + m_2^2 s_{12}^2 c_{13}^2 + m_3^2 s_{13}^2.
 \label{eq:m_beta}
\end{equation}
Note that this mass is really just an effective observable and in no sense related to something like a mass of the $\bar{\nu}_e$, since that cannot even exist as argued above. Although such single beta decay experiments require an enormous amount of precision, there are impressive limits by past experiments using tritium: the MAINZ experiment measured $m_\beta^2 = (-1.2 \pm 2.2_{\rm stat.} \pm 2.1_{\rm sys.})~{\rm eV}^2$~\cite{Kraus:2004zw}, and TROITSK measured $m_\beta^2 = (-1.0 \pm 3.0_{\rm stat.} \pm 2.1_{\rm sys.})~{\rm eV}^2$~\cite{Lobashev:1999tp}, both amounting to an upper limit of about $2$~eV on the absolute neutrino mass scale. The next generation experiment KATRIN~\cite{Osipowicz:2001sq} will have a sensitivity on $m_\beta$ that is better by about one order of magnitude. A new technique for kinematical experiments, measuring the coherent cyclotron radiation emitted by  fast electrons in a magnetic field, is investigated by Project~8~\cite{Monreal:2009za}. Furthermore, an alternative approach will be taken by the MARE experiment~\cite{Nucciotti:2010tx}, using rhenium, which has a lower $Q$-value and hence an improved sensitivity on the neutrino mass, but one pays the price of the transition of interest being a first forbidden decay.

One of the biggest questions in neutrino physics is about their actual nature -- are the neutrinos \emph{Dirac} particles, \emph{i.e.} distinct from their antiparticles, $\nu^c \neq e^{i\phi} \nu$, or are they \emph{Majorana} particles, \emph{i.e.} identical to their charge conjugates up to a phase, $\nu^c = e^{i\phi} \nu$? If the latter was the case, then an antineutrino $\bar{\nu}_e$ emitted in a $\beta^-$ decay $(Z,A) \to (Z+1,A) + e^- + \bar{\nu}_e$ could be absorbed as neutrino in an inverse $\beta$ decay $\nu_e + (Z+1,A) \to (Z+2,A) + e^-$. If these two transitions happen in one single nucleus, with the neutrino being purely virtual, the resulting net reaction is called \emph{neutrino-less double beta decay} ($0\nu\beta\beta$), $(Z,A) \to (Z+2,A) + e^- + e^-$. There are many more interesting connections of $0\nu\beta\beta$ to particle physics, as can be seen from the very detailed and informative review by Rodejohann, Ref.~\refcite{Rodejohann:2011mu}. A review that is dedicated a bit more to the nuclear physics side can be found in Ref.~\refcite{Vergados:2012xy}. Recent overviews of the on-going and future experimental activities are provided by Refs.~\refcite{Barabash:2011fg,Sarazin:2012ct}.

It is not a priori clear if light neutrino exchange is indeed the dominant contribution to $0\nu\beta\beta$ (see, \emph{e.g.}, Refs.~\refcite{Hirsch:1995vr,Hirsch:1996qw,Pas:2000vn,Bhattacharyya:2002vf,Bhattacharyya:2004kb,Deppisch:2006hb,Simkovic:2010ka,Faessler:2011qw,Bergstrom:2011dt,Deppisch:2012nb} for discussions of situations where this is not the case). However, if it does dominate, then the decay rate is proportional to the square of a quantity called the \emph{effective mass} $|m_{ee}|$ given by
\begin{equation}
 |m_{ee}| \equiv |m_1 c_{12}^2 c_{13}^2 + m_2 s_{12}^2 c_{13}^2 e^{i\alpha} + m_3 s_{13}^2 e^{i\beta}|.
 \label{eq:m_ee}
\end{equation}
Note that $\alpha$ and $\beta$ are the Majorana phases from Eq.~\eqref{eq:PMNS_expl}, which did \emph{not} appear in the observable $m_\beta^2$. This is a reflection of lepton number being violated in the reaction $(Z,A) \to (Z+2,A) + e^- + e^-$. Also, contrary to some statements in the literature, the effective mass $|m_{ee}|$ does \emph{not} depend on the Dirac phase $\delta$, since this phase can be absorbed in a redefinition of the neutrino mass eigenstate $\nu_3$, cf.\ Refs.~\refcite{Lindner:2005kr,Merle:2006du}, which can be seen easily by adopting the more suitable symmetric parametrization~\cite{Rodejohann:2011vc} of the PMNS matrix. In  contrast to $m_\beta^2$, $|m_{ee}|$ cannot be split into different parts even with an infinitely accurate experimental setup, since the contributions of the different neutrino mass eigenstates are purely virtual, and hence the total amplitude is given by a \emph{coherent sum} over partial amplitudes.

In order to be experimentally observable, one needs a nucleus in which single beta decay is kinematically forbidden, as otherwise this process would completely dominate the decay. Even then the process of double beta decay with the emission of two neutrinos, $(Z,A) \to (Z+2,A) + e^- + e^- + \bar{\nu}_e + \bar{\nu}_e$, cannot be forbidden, and one has to search for the smoking gun signature of the sum of the two electron energies being equal to the $Q$-value of the decay. Among the isotopes investigated so far and in the future are $^{76}{\rm Ge}$ (GERDA~\cite{Abt:2004yk}, Heidelberg-Moscow~\cite{KlapdorKleingrothaus:2000sn}, IGEX~\cite{Aalseth:2002rf}, Majorana~\cite{Aalseth:2004yt}), $^{130}{\rm Te}$ (CUORE~\cite{Ardito:2005ar}, CUORICINO~\cite{Andreotti:2010vj}), or $^{136}{\rm Xe}$ (EXO~\cite{Auger:2012ar}, KamLAND-Zen~\cite{Gando:2012zm}). Alternative decay modes are mainly investigated in COBRA~\cite{Bloxham:2007aa}. Up to now, $0\nu\beta\beta$ has not been observed\footnote{The claim made in Ref.~\refcite{KlapdorKleingrothaus:2001ke} is under strong pressure, if not even excluded, by Refs.~\refcite{Auger:2012ar,Gando:2012zm}.}, and typical lower limits on the half-lives are of the order of $10^{23}$ to $10^{25}$~years, which translates, modulo nuclear physics complications, into upper limits on $|m_{ee}|$ of a few tenths of an eV, cf.\ Ref.~\refcite{Rodejohann:2011mu}.

Finally, note that a positive signal of $0\nu\beta\beta$ would unambiguously show that lepton number is violated and that neutrinos are hence Majorana particles~\cite{Schechter:1981bd}. However, the corresponding mass contribution generated by the resulting \emph{Butterfly diagram} would be at most about $10^{-24}$~eV or in some cases even vanishing at 4-loop level~\cite{Duerr:2011zd}, and hence by far too tiny to explain the minimum neutrino mass scale enforced by the measured value of $|\Delta m_{31}^2|$.

A third way to get information on the absolute neutrino mass scale is by cosmological observations, where neutrinos also play an important role~\cite{Dolgov:2002wy}. Using basic cosmological arguments, one can determine the ratio of the neutrino temperature in the Universe after $e^+ e^-$ annihilations to the ``reheated'' photon temperature. This allows to relate the sum $\Sigma$ of all light neutrino masses to their relic density $\Omega_\nu h^2$:
\begin{equation}
 \Sigma = 94~{\rm eV}\cdot \Omega_\nu h^2.
 \label{eq:Sigma_Omega}
\end{equation}
if one now uses the observations of the \emph{cosmic microwave background} (CMB) to determine the total matter density $\Omega_{\rm mat}$, as well as the baryon density $\Omega_b$ and the dark matter relic density $\Omega_{\rm DM}$\footnote{To be precise, the standard cosmological analysis actually uses the $\Lambda$CDM model, \emph{i.e.}, the cosmological standard model including Dark Energy and \emph{Cold} Dark Matter.}, one can determine the neutrino relic density as $\Omega_\nu = \Omega_{\rm mat} - \Omega_b - \Omega_{\rm DM}$~\cite{Komatsu:2008hk}. The newest analysis of the Wilkinson Microwave Anisotropy Probe (WMAP) 9-year data then yields a limit of $\Sigma < 1.3$~eV\ @\ 95\% C.L.~\cite{Hinshaw:2012fq}, which can be improved to $0.44$~eV when including data from other CMB-measurements, from baryon acoustic oscillations, and from supernova redshift surveys. The even newer results from the Planck satellite, combined with the polarization measurements from WMAP, the data from terrestrial telescopes, and from baryon acoustic oscillation can even push this number down to $\Sigma < 0.230$~eV\ @\ 95\% C.L.~\cite{Ade:2013lta} However, it should be mentioned that all these limits depend strongly on which data set and which type of statistics are used. Note that there has been a recent hint for a non-zero value $\Sigma = (0.32 \pm 0.11)$~eV by the South Pole Telescope (SPT) collaboration~\cite{Hou:2012xq}. However, the similar Atacama Cosmology Telescope (ACT)~\cite{Sievers:2013wk} currently does not seem to confirm this observation, and also the combined Planck results (which include the SPT data!) seem to indicate that this value is unlikely. Furthermore, recent studies seem to indicate that there might be some inconsistency between the analyses of the SPT and ACT data sets~\cite{DiValentino:2013mt,Archidiacono:2013lva}, which could lead to different results even though the primary data of both telescopes do not look very different, cf.\ Fig.~1 in Ref.~\refcite{Hinshaw:2012fq}. It is probably fair to say that there is currently no fully consistent picture.

The aim would be to combine the cosmological value of $\Sigma$ with laboratory data, which could positively influence each other~\cite{Host:2007wh}. Ideally, we would therefore like $\Sigma$ to be the sum of only the three active neutrino masses,
\begin{equation}
 \Sigma = m_1 + m_2 + m_3.
 \label{eq:Sigma_m}
\end{equation}
However, $\Sigma$ actually denotes the sum over \emph{all} light neutrino species, which could also include sterile neutrinos if they are light enough~\cite{Hamann:2010bk}. In this context, ``light enough'' means that these species are highly relativistic at recombination time, when the CMB is produced. Note that, because of this requirement, keV neutrinos themselves will not give a significant contribution to $\Sigma$, as they are required to be slow enough not to spoil structure formation. Extrapolating the expansion of the Universe from that point back to recombination, it is easy to see that keV neutrinos only give a negligible contribution to $\Sigma$. It should be obvious that it can be a bit subtle to truly understand what contributes to $\Sigma$, and what does not. Even worse, there might be unknown sources of such a contribution, in which case the interpretation of $\Sigma$ as the sum of the three active neutrino masses could fail completely. In such a case, if one erroneously assumed Eq.~\eqref{eq:Sigma_m} but the measured parameter $\Sigma$ is in fact dominated by other contributions (\emph{i.e.}, by systematic errors), a naive analysis could lead to wrong limits, or even to a fake ``measurement'' of the neutrino mass scale~\cite{Maneschg:2008sf}.

Summing up, we have, in particular within the last decades, gained considerable knowledge on the neutrino masses and mixing parameters. Yet we are puzzled by their values, and in particular by the smallness of their masses and by the large mixings compared to the quark sector. To the best of our knowledge, this seems to suggest a structure, or symmetry principle, behind these parameters. It is exactly at this point where neutrino model building comes into the game.

\subsection{\label{sec:neutrino_modeling}Model building for light neutrinos}	

The goal of model building in neutrino physics can be summarized in one sentence: we aim to find a particle physics explanation for the sizes and values of neutrino masses and leptonic mixing angles.

Not every approach is powerful enough to explain both these aspects, so it makes sense to divide the models on the market into \emph{mass models}, which give an explanation for the tininess of the active neutrino masses, and \emph{flavor models}, which can explain the mixing pattern in the leptonic sector.\footnote{Often, the quark sector is disregarded in that context, since quark mixing is not as significant as the leptonic mixing. However, even in the more recent literature one can find examples for, \emph{e.g.}, the successful prediction of the Cabibbo angle based on symmetry principles~\cite{Blum:2007nt,Blum:2009nh,Ishimori:2010xk}.} Typically, mass models do not predict an exact value for the mass, which is intrinsically difficult in a quantum field theory, but rather some hierarchy among certain masses, \emph{e.g.}, why the light neutrinos should be much lighter than the other SM fermions. Flavor models, in turn, often predict very specific values for certain mixing angles, however at the prize of introducing practically unobservable high energy sectors. On top of that, the magnitude of mixing angles is intrinsically fixed -- they are either practically zero or of $\mathcal{O}(1)$, which is a qualitative difference to mass hierarchies.

\subsubsection{\label{sec:neutrino_modeling_mass}Neutrino mass models}	

Starting with mass models, the question to be answered is why the three known active neutrinos have masses of at most $\mathcal{O}(1~{\rm eV})$, while even the electron already has a mass of $511$~keV, and the other SM fermions are by far more massive. In the SM itself, neutrinos are even strictly massless, but this is known to be phenomenologically invalid from the observation of neutrino oscillations, cf.\ Sec.~\ref{sec:neutrino_mix}. However, the reason for neutrinos being massless in the SM is actually a rather artificial one: any fermionic mass term would need to couple left-handed (LH) to right-handed (RH) fermion fields, \emph{e.g.} $\mathcal{L}_\nu = -m_\nu \overline{\nu_L} \nu_R + h.c.$, but while (in the 1-family approximation) left-handed neutrinos $\nu_L$ appear together with the left-handed electron $e_L$ as components of $SU(2)$ doublet fields in the SM, their right-handed counterparts $\nu_R$ are simply not included in the particle content. But even introducing right-handed neutrinos to create an extended version of the SM does not really cure the problem: in that case, we would be forced by the SM-gauge symmetry to write down a SM-like Yukawa coupling term,
\begin{equation}
 \mathcal{L}_Y \supset - \overline{L} \tilde H y_\nu \nu_R + h.c. \to - \overline{L} \langle \tilde H \rangle y_\nu \nu_R + h.c. = - \overline{\nu_L} m_D \nu_R + h.c., 
 \label{eq:nu-SM-mass}
\end{equation}
where $L=(\nu_L, e_L)^T$ is the SM-lepton doublet, $\tilde H = i \sigma^2 H$, $H=(H^+, H^0)$ is the SM-Higgs doublet, and $\langle H \rangle = (0, v)^T$ with $v=174$~GeV is its \emph{vacuum expectation value} (VEV) which is generated by electroweak symmetry breaking. Such a mechanism induces a mass $m_D = y_\nu v$ for the neutrino, which we have silently given a lower index $D$. This stands for \emph{Dirac mass}, and it emphasizes that the neutrino in the SM-extended by RH neutrinos is a \emph{Dirac particle}, \emph{i.e.}, distinct from its antiparticle. The problem with this way of generating the neutrino mass is that the VEV $v$ is fixed to a value of $174$~GeV by the known masses of the $W$- and $Z$-bosons, and unless the \emph{Yukawa coupling} $y_\nu$ is extremely tiny, $y_\nu \lesssim 10^{-11}$, the corresponding neutrino mass will be too large.\footnote{Note that one can also find the alternative convention $\langle H \rangle = (0, v/\sqrt{2})^T$ in the literature, in which case $v=246$~GeV. This does, of course, not change the principal argumentation.} Hence, the SM intrinsically predicts a wrong scale for the neutrino mass, unless we find an explanation for $y_\nu$ to be tiny.

The introduction of a Dirac mass reflects the fact that the SM conserves lepton number at the perturbative level. However, lepton number conservation is only an \emph{accidental} symmetry of the SM: it has not been imposed, and it is by far not sacrosanct. Actually, the SM itself does violate lepton number, but only at the non-perturbative level by so-called \emph{sphaleron} processes~\cite{tHooft:1976up,Klinkhamer:1984di,Kuzmin:1985mm,Fukugita:1986hr,Arnold:1987mh}, or by higher-dimensional operators which are not renormalizable~\cite{Weinberg:1979sa,deGouvea:2007xp}. If we accept that lepton number may be violated, then gauge symmetry immediately allows a new term in the SM extended by RH neutrinos, the so-called \emph{Majorana mass term},
\begin{equation}
 \mathcal{L}_M = -\frac{1}{2} \overline{(\nu_R)^c} M_R \nu_R + h.c.,
 \label{eq:RH-Majorana}
\end{equation}
where $(\nu_R)^c$ is the charge-conjugate of a right-handed field and, as such, left-handed. The factor $\frac{1}{2}$ in Eq.~\eqref{eq:RH-Majorana} arises since the spinors used are Majorana spinors, which are effectively real quantities, just as the mass term of a real scalar has a factor of $\frac{1}{2}$ compared to the mass term of a complex scalar. The term in Eq.~\eqref{eq:RH-Majorana} clearly violates lepton number by two units. Even more importantly, the coefficient $M_R$ of this term has the dimension of a mass which is \emph{not} in any way related to the electroweak scale $v$. Hence, $M_R$ can be \emph{arbitrary}, as long as it is smaller than the Planck mass. In particular, $M_R$ will in general not be zero.

Although there is actually not too much of a reason to suspect $M_R$ to be much larger than $v$, it is often assumed to be. This notion is also the reason for the notation $N_R$ instead of $\nu_R$, which is typically used in the literature and which we will also use from now on. If $M_R$ is indeed relatively large, an interesting thing happens: the full neutrino mass term can be rewritten in terms of a larger matrix,
\begin{equation}
 \mathcal{L}_\nu = -\frac{1}{2} (\overline{\nu_L} , \overline{(N_R)^c})
 \begin{pmatrix}
 0 & m_D\\
 m_D^T & M_R
 \end{pmatrix}
 \begin{pmatrix}
 (\nu_L)^c\\
 N_R
 \end{pmatrix} + h.c.
 \label{eq:seesawI_1}
\end{equation}
As shown explicitly in \ref{sec:seesaw}, this large matrix can be approximately block-diagonalized for $M_R \gg m_D$ and leads to a light neutrino mass matrix given by
\begin{equation}
 m_\nu = - m_D M_R^{-1} m_D^T.
 \label{eq:seesawI_2}
\end{equation}
What have we gained by this? Since $m_D \ll M_R$, a term of $\mathcal{O}\left(\frac{m_D}{M_R}\right)$ is much smaller than one, and by this it leads to a suppression of the mass scale $m_D$. If the scale of $M_R$ is large enough, say $10^{14}$~GeV, then a mass $m_D\sim v = \mathcal{O}(100~{\rm GeV})$, as we would expect it, would lead to an active neutrino mass of
\begin{equation}
 m_\nu \sim \frac{(100~{\rm GeV})^2}{10^{14}~{\rm GeV}} = 10^{-10}~{\rm GeV} = 0.1~{\rm eV},
 \label{eq:seesaw_3}
\end{equation}
which is just below the upper bounds mentioned in Sec.~\ref{sec:neutrino_mix}. This is the famous \emph{seesaw type~I mechanism}~\cite{Minkowski:1977sc,Yanagida:1979as,GellMann:1980vs,Glashow:1979nm,Mohapatra:1979ia}. Note that this mechanism can be cast in a Feynman diagram, see left panel of Fig.~\ref{fig:seesaw}: integrating out the heavy neutrino, we exactly retrieve Eq.~\eqref{eq:seesawI_2}.

\begin{figure}[pb]
\centerline{
\begin{tabular}{lr}
\psfig{file=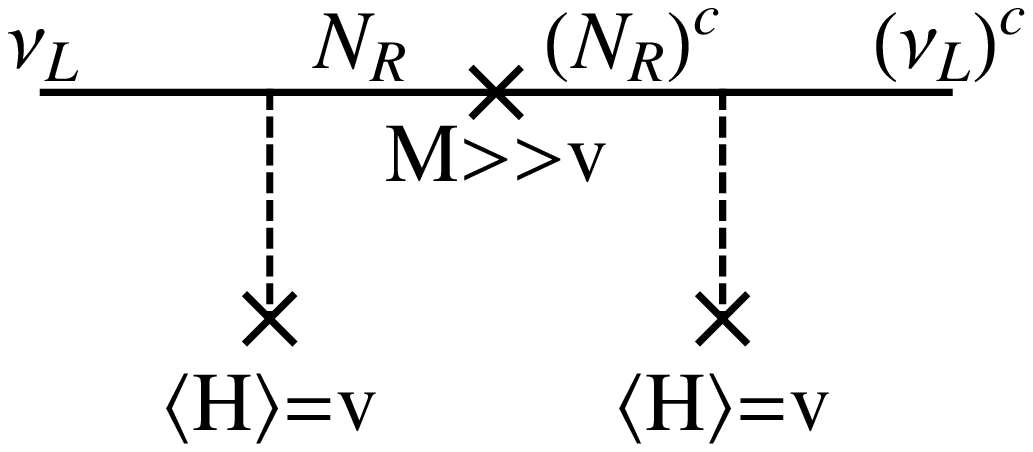,width=5.0cm} & \psfig{file=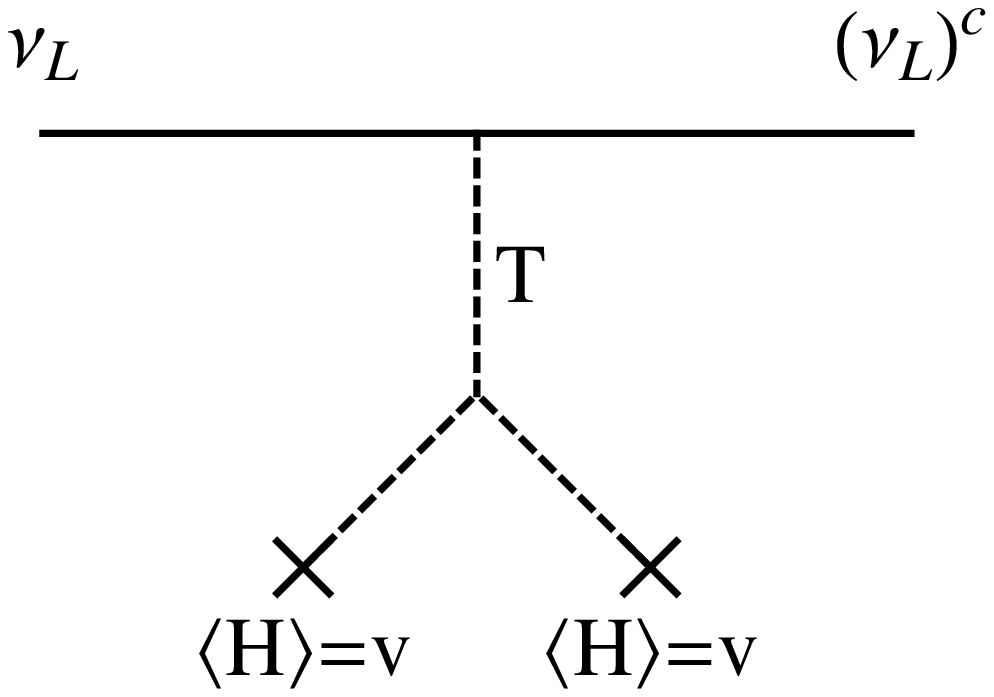,width=5.0cm}
\end{tabular}
}
\vspace*{8pt}
\caption{\label{fig:seesaw}The Feynman diagrams for the seesaw mechanisms of type~I and type~II.}
\end{figure}

This mechanism beautifully exhibits many principles of model building: we try to explain certain facts, in this case the smallness of the neutrino mass, by introducing new ingredients which are less constrained, in this case the Majorana mass term and the heavy fields. At first sight, this might seem to be only a shift of the problem to a different sector, or even sweeping it under the carpet. However, the key point is that a new model should always lead to certain new \emph{predictions}, which make it \emph{testable}. In the seesaw type~I case, these predictions include the mere existence of the right-handed neutrinos, as well as their possible cosmological consequences by \emph{leptogenesis}~\cite{Fukugita:1986hr}.\footnote{Note that, if the seesaw mechanism is associated with an even lower scale, a lot of additional cosmological probes might be possible~\cite{deGouvea:2006gz}.} By this we can get more insights into the possibilities of Nature: many simple-minded models do not work, and our task in model building is to find the simplest and most beautiful models which are not only consistent with current data, but which also lead to testable predictions, ideally in the near future. It is then the task of the experimentalists to measure certain parameters -- either more precisely or for the first time at all -- which we theorists can then use to exclude the various candidate models one by one. In the best case, we would by this procedure end up with the single model that is as close to the truth as possible.

To name a few more possibilities, seesaw-type mechanisms can also be achieved by adding either a Higgs triplet $T$ (\emph{type~II}~\cite{Magg:1980ut,Lazarides:1980nt}) or (typically three) fermion triplets $\Sigma$ (\emph{type~III}~\cite{Foot:1988aq,Ma:2002pf}), or by extending the type~I seesaw by additional singlets which are not right-handed neutrinos, as in the \emph{inverse seesaw}~\cite{Mohapatra:1986bd}. For example, in type~II seesaw this leads, among other terms, to a triplet Yukawa coupling given by
\begin{equation}
 \mathcal{L}_{Y_T} = - \overline{L} i \sigma_2 T y_T L^c + h.c.,
 \label{eq:triplet_Yukawa}
\end{equation}
where the Higgs triplet fields is written in components as
\begin{equation}
 T=\begin{pmatrix}
 T^+/\sqrt{2} & T^{++}\\
 T^0 & -T^-/\sqrt{2}
 \end{pmatrix} \to \langle T \rangle =\begin{pmatrix}
 0 & 0\\
 v_T & 0
 \end{pmatrix}.
 \label{eq:triplet_Higgs}
\end{equation}
The point is that the VEV $v_T$ is actually induced by the ordinary Higgs doublet VEV, due to the term
\begin{equation}
 \mathcal{L}_{\rm scalar} \supset \mu H^T i \sigma_2 T H + h.c. - M_T^2 {\rm Tr} \left( T^\dagger T \right),
 \label{eq:triplet_tadpole}
\end{equation}
which induces a tadpole term for the Higgs triplet, cf.\ right panel of Fig.~\ref{fig:seesaw}. Indeed, this leads to a LH neutrino mass as in Eq.~\eqref{eq:Majorana_app}, where
\begin{equation}
 \mathcal{L}_{Y_T} \to - \overline{\nu_L} v_T y_T (\nu_L)^c + h.c. = - \overline{\nu_L} \mu \frac{v^2}{M_T^2} y_T (\nu_L)^c + h.c. \equiv -\frac{1}{2} \overline{\nu_L} m_L (\nu_L)^c + h.c.
 \label{eq:LH_mass}
\end{equation}
Note that the potentially large mass $M_T$ suppresses the triplet VEV, which would in any case be bound by the deviation of the so-called $\rho$-parameter, $\rho = \frac{M_W^2}{M_Z^2 \cos^2 \theta_W}$ where $M_W$ ($M_Z$) is the mass of the SM $W$- ($Z$-) boson, from its SM value of $1$ to $v_T \lesssim \mathcal{O}(1~{\rm GeV})$~\cite{Kanemura:2012rs}. If in addition the Yukawa coupling matrix $y_T$ has somehow small entries, a small LH neutrino mass $m_L$ is justified.

An interesting alternative for neutrino mass models is to suppress the neutrino mass not by a tree-level diagram but to let it instead vanish exactly, such that it only arises at loop-level. The most simple such extensions of the SM which are not excluded are probably Ma's \emph{scotogenic model}~\cite{Ma:2006km} (1-loop) and the \emph{Zee-Babu model}~\cite{Zee:1985rj,Zee:1985id,Babu:1988ki} (2-loop). Even higher suppressions are present in the \emph{Aoki-Kanemura-Seto model}~\cite{Aoki:2008av,Aoki:2011zg} (3-loop) or in the Butterfly diagram (4-loop) by Schechter and Valle~\cite{Schechter:1981bd,Duerr:2011zd}.

\subsubsection{\label{sec:neutrino_modeling_FN}The Froggatt-Nielsen mechanism}	

What practically all neutrino mass models have in common is that, even though they provide an explanation for the smallness of neutrino masses, they usually do not give immediate predictions for the mixing angles.\footnote{Notable exceptions are, \emph{e.g.}, the \emph{Zee-Wolfenstein model}~\cite{Zee:1980ai,Wolfenstein:1980sy} or the left-right symmetric extension of the scotogenic model~\cite{Adulpravitchai:2009re}. However, the former is excluded by data, and the latter employs an additional symmetry, which actually does not make it a pure mass model.} Indeed, this is a requirement that is non-trivial to achieve.

The state of the art of most of the models is to make use of so-called \emph{flavor symmetries}. The idea behind this approach is to extend the undoubtably successful application of symmetries, which for example predict the structure of the gauge sector of the SM. \emph{E.g.}, the color $SU(3)_C$ symmetry in Quantum Chromodynamics (QCD) predicts the existence of eight different gluons, and the $SU(2)_L \times U(1)_Y$ symmetry of the electroweak sector predicts the $\rho$-parameter to be exactly one at tree-level in the SM~\cite{Beringer:1900zz}. Similar considerations could, in principle, be applied to the flavor sector to predict certain mass and mixing patterns.

The first question to ask is whether the flavor symmetry used should be a \emph{continuous} or a \emph{discrete} symmetry. Having the SM in mind, a continuous symmetry seems to be the natural choice, as such symmetries predict the structures of the SM: \emph{e.g.}, the gauge sector exhibits an $SU(3)_C \times SU(2)_L \times U(1)_Y$ symmetry~\cite{Beringer:1900zz} and also lepton number could be described by one or more $U(1)$ rotations~\cite{Davidson:2006bd}. Unfortunately, a practical problem arises: typically, flavor symmetries are broken for phenomenological reasons. One could in principle assume the symmetry to be explicitly broken, but then the question might arise in how far the symmetry under consideration is actually present in a certain model. Breaking the symmetry spontaneously, however, will by the so-called \emph{Goldstone theorem}~\cite{Goldstone:1961eq,Goldstone:1962es} generically introduce unwanted massless scalar particles, which are normally a phenomenological disaster. One way out would be to gauge the flavor symmetry, as was done, \emph{e.g.}, in Refs.~\refcite{Grinstein:2010ve,Guadagnoli:2011id}. Alternatively, one could simply use a discrete symmetry, which would not suffer from unwanted Goldstones at all as long as no accidental continuous symmetry appears that gets broken at some point.

If not the first actual flavor symmetry, at least the most intuitive to understand was provided by Froggatt and Nielsen in 1979, by now termed the \emph{Froggatt-Nielsen (FN) mechnanism}~\cite{Froggatt:1978nt}. Although originally developed for the quark sector in order to simultaneously explain the mass and mixing pattern as well as the occurrence of $CP$ violation~\cite{CPV}, the FN mechanism is equally suited for the leptonic sector. The basic idea is very simple: since any fermionic mass matrix in the SM arises at the electroweak scale $v$, it is tempting to consider a mass matrix $M= v Y$ to originate from a Yukawa coupling matrix $Y$ of the form
\begin{equation}
 Y_{ij} = Y_{ij}^{\rm nat} \lambda^{a_i + b_j}, 
 \label{eq:FN_1}
\end{equation}
where the ``natural'' entries $Y_{ij}^{\rm nat}$ are all of the same order, and the suppression of certain elements comes from powers of a suppression factor $\lambda$, which depend on certain generation-dependent ``charges'' $a_i$ and $b_j$. The trick is to re-interpret the actual couplings $Y_{ij}^{\rm nat} \lambda^{a_i + b_j}$ as arising from the VEV of a new SM-singlet scalar field $\Theta$, which is integrated out at a high energy scale $\Lambda$. The ``static'' couplings are hence promoted to dynamic quantities arising from a field which would, in a more modern language, be termed a \emph{spurion}. The charges $a_i$ and $b_j$ are interpreted as being the quantum numbers of a new $U(1)_{\rm FN}$ symmetry under which the different fermion generations are charged differently. Furthermore, the new scalar $\Theta$, usually called \emph{flavon},\footnote{Note that this term is generally used for scalar fields that are charged only under the flavor symmetry, but are SM singlets otherwise and obtain a VEV. If such singlets do not obtain a VEV, they are often termed \emph{driving fields} or \emph{waterfall fields}, since they are important for the vacuum structure of the scalar potential.} has a non-trivial charge under $U(1)_{\rm FN}$, and a new set of heavy singlet fermions $S_\alpha$, all suitably charged under the new symmetry, is introduced as well.

What we gain by the introduction of all these new ingredients is that most of the Yukawa couplings, as they appear in the SM, are now \emph{forbidden} by the $U(1)_{\rm FN}$ symmetry. Instead, new couplings made of higher order terms appear, which couple the flavon to either the SM Higgs, one SM fermion, and one heavy fermion, or to two heavy fermions, \emph{e.g.}:
\begin{equation}
 \mathcal{L}_{\Theta, \rm SM} = Y^\Theta_{11} \overline{L_1} H S_1 \frac{\Theta}{\Lambda} + h.c.\ \ \ {\rm or}\ \ \ \mathcal{L}_{\Theta, \rm singlet} = Y^S_{12} \Theta S_1 S_2  + h.c.
 \label{eq:flavon-Yukawas}
\end{equation}
The exact set of terms is determined by the charges chosen. Now, since the flavon $\Theta$ obtains a VEV $\langle \Theta \rangle$, this VEV breaks all symmetries to which the flavon couples, \emph{i.e.}, the $U(1)_{\rm FN}$. Hence, one can use a ``chain'' of couplings as the ones in Eq.~\eqref{eq:flavon-Yukawas} to construct extended seesaw-like diagrams, see Fig.~\ref{fig:FN-diagram} (cf.\ left panel of Fig.~\ref{fig:seesaw}), where each VEV breaks the $U(1)_{\rm FN}$ charge by $k$ units, with $k$ being the charge of the flavon field $\Theta$. Integrating out each of these heavy fermions introduces powers of the small quantity $\lambda \equiv \frac{\langle \Theta \rangle}{\Lambda}$ and, since these suppressions are different for different SM fermion generations, we end up with hierarchical Yukawa couplings at low energies, exactly as in Eq.~\eqref{eq:FN_1}. This allows for a first, though simple, understanding of what ``flavor'' actually could be.

\begin{figure}[pb]
\centerline{
\psfig{file=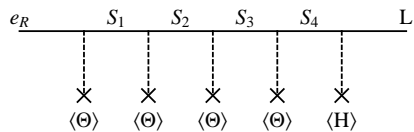,width=6.0cm}
}
\vspace*{8pt}
\caption{\label{fig:FN-diagram}Example Feynman diagram for a mass term generated by the FN mechanism.}
\end{figure}

For example, let us take an easy one-generation model of an electron, contained in a LH doublet $L$ and a RH singlet $e_R$, where the FN charges of ($\Theta$, $L$, $e_R$, $H$) are taken to be $(k_\Theta, k_L, k_e, k_H)=(1, 2, 0, 0)$. Then, a term $\overline{L} H e_R + h.c.$ would have a $U(1)_{\rm FN}$ charge of $-2 + 0 + 0 = -2$ (note the conjugate in the bar!), and hence be forbidden. However, if we introduce a new fermion $S$, we can write down a term $y_L \overline{L} H \frac{\Theta}{\Lambda} S + h.c.$ if we take the $U(1)_{\rm FN}$ charge of $S$ to be $k_S = 1$ (and also account for the correct electrical charge). Equally fine is a term $y_R \overline{S} \Theta e_R + h.c.$, which would have a FN charge of $-1 +1 +0 = 0$. The decisive part of the Lagrangian would then be
\begin{equation}
 \mathcal{L} \supset - y_L \overline{L} H \frac{\Theta}{\Lambda} S - y_L \overline{S} \frac{\Theta^*}{\Lambda} H^* L - y_R \overline{S} \Theta e_R - y_R \overline{e_R} \Theta^* S - M \overline{S} S ,
 \label{eq:FN_ex_1}
\end{equation}
where we have written the $h.c.$-terms explicitly, and we have chosen all couplings to be real by absorbing their phases into the fields. We have furthermore introduced a large mass $M$ of the new fermion $S$. If now $M$ is so large that $S$ is essentially static (\emph{i.e.}, its kinetic term is negligible), then we can integrate out $S$ and $\overline{S}$ by using the Euler-Lagrange equations, \emph{e.g.},
\begin{equation}
 \frac{\delta \mathcal{L}}{\delta \overline{S}} = - y_L \frac{\Theta^*}{\Lambda} H^* L - y_R \Theta e_R - M S = \partial_\mu \frac{\delta \mathcal{L}}{\delta (\partial_\mu \overline{S} )} \simeq 0,
 \label{eq:FN_ex_2}
\end{equation}
and hence
\begin{equation}
 S \simeq \frac{-1}{M} \left( y_L \frac{\Theta^*}{\Lambda} H^* L + y_R \Theta e_R \right).
 \label{eq:FN_ex_3}
\end{equation}
Inserting this value back into Eq.~\eqref{eq:FN_ex_1} and inserting the VEVs  $\langle \Theta^{(*)} \rangle = w$, we obtain:
\begin{equation}
 \mathcal{L} \supset \frac{2 y _L y_R}{M} \overline{L} H \frac{w^2}{\Lambda} e_R + h.c.
 \label{eq:FN_ex_4}
\end{equation}
If the mass $M$ is taken to be identical to the high scale $\Lambda$, which was not specified up to now, then we finally arrive at an electron Yukawa coupling given by
\begin{equation}
 \mathcal{L} \supset 2 y _L y_R \lambda^2  \overline{L} H {\Lambda} e_R + h.c.,
 \label{eq:FN_ex_5}
\end{equation}
where $\lambda\equiv w/\Lambda$. This has precisely the form anticipated in Eq.~\eqref{eq:FN_1}: the Yukawa coupling has a natural size $2 y _L y_R$, which is suppressed by a certain power $\lambda^{k_L + k_e} = \lambda^2$, due to $\lambda$ being smaller than one.

Although the basic principle of the FN mechanism is easy to understand, it already reveals certain problems that are typically associated with flavor symmetries:
\begin{itemize}

\item \emph{Goldstone bosons}:\\
As mentioned earlier, the breaking the $U(1)_{\rm FN}$ may introduce undesirable massless Goldstone bosons~\cite{Goldstone:1961eq,Goldstone:1962es}. To go around that, we either have to gauge the continuous symmetry, or otherwise give mass to the Goldstone modes by some non-perturbative mechanism as in, \emph{e.g.}, Refs.~\refcite{Kallosh:1995hi,Berenstein:2010ta}, or have the symmetry explicitly broken from the very beginning, as is the case for pions~\cite{Bernard:1983ix}. This problem does not only appear when we explicitly impose a continuous symmetry from the very beginning, as in the FN case, but also if there is an \emph{accidental} symmetry in the problem. Even worse, we might not even be aware of such a symmetry, and wonder why massless states arise at all. This problem can be, to some extent, circumvented when using discrete flavor symmetries instead, but a model builder should still be careful to precisely determine the actual symmetry of the problem. Treatments of how to obtain a discrete from the breaking of a continuous symmetry might help in that respect~\cite{Etesi:1997jv,Koca:1997td,Koca:2003jy,Adulpravitchai:2009kd,Berger:2009tt,Luhn:2011ip}, and in particular the use of so-called \emph{invariant polynomials} may be of advantage~\cite{Merle:2011vy}.

\item \emph{higher-order terms}:\\
As in Eq.~\eqref{eq:flavon-Yukawas}, the actual Yukawa couplings at low energies are in fact higher order terms. In other words, we are often working in an \emph{effective field theory} (EFT), and hence our Lagrangians are \emph{non-renormalizable}. While this may not be a big problem in practice, since one can very well work with EFTs (and since the neutrino mass induced by the seesaw diagram itself only exists in the framework of an EFT), one should still keep in mind that in reality we would need to know the fundamental high energy theory in order to truly derive predictions such as the sizes of certain couplings from a model. However, since we treat the couplings as free parameters, and since not even a plausible candidate theory to explain their absolute values is known, we can neglect this problem in most practical considerations. While many models rely on non-renormalizable operators, it should be noted that there exist also models based on flavor symmetries which are fully renormalizable, in particular when being constructed within a Grand Unified Theory context (see, for example, Refs.~\refcite{Babu:2002dz,Frampton:2008bz,Patel:2010hr,Blankenburg:2011vw,Ferreira:2012ri}).

\item \emph{unknown high energy sector}:\\
A bigger problem in practice is the existence of the high energy sector itself, in the FN case the flavon (or several flavons) as well as the heavy singlet fermions. This is, first of all, an aesthetic problem, since a completely unknown sector is introduced but only partially used and often not even discussed in what concerns its possible effects or signatures. However, at least in the early Universe, this scalar sector could potentially be accessible, and by this the associated scalar SM singlet fields could, \emph{e.g.}, be responsible for inflation~\cite{Antusch:2008gw}.

\item \emph{symmetry and charge assignments}:\\
Already in the FN example, it appears that the charges $a_i$ and $b_j$ can essentially be freely chosen.\footnote{Note that, although typically chosen to have integer values in FN-inspired models, $U(1)$ charges can actually be arbitrary real numbers.} Not only that, even the whole symmetry group can be more or less arbitrarily chosen, and the $U(1)_{\rm FN}$ is only an example. However, in the spirit of the unification of three generations of fermions, one might be tempted to focus on symmetry groups that actually do have \emph{three}-dimensional \emph{irreducible representations} (which will be discussed in a second), as this might signal a common origin for, \emph{e.g.}, the three known charged lepton flavors. On the other hand, also two-dimensional and even one-dimensional representations are used, the latter being applied in the $U(1)_{\rm FN}$ example.

\end{itemize}
The Froggatt-Nielsen mechanism already contains practically all the important features and characteristics of modern flavor symmetry models. In fact, it can serve as example throughout the remainder of this review, and actually it is still used by many modern models, which are often only variants (though more complicated ones) of the ``old'' idea behind FN.

\subsubsection{\label{sec:neutrino_modeling_flavor}Neutrino flavor models}	

We will now consider a more modern model, to make clear where we actually need group-theoretical aspects. Although the field of flavor model building is an industry by itself, we will try to give a \emph{flavor} of what is going on, which will turn out to be very useful in the later sections. A very detailed review on the details of discrete flavor symmetries by King and Luhn has recently become available in Ref.~\refcite{King:2013eh}.

The model to be presented here is a slight variant the minimal $A_4$ model from Ref.~\refcite{Chen:2009um}, which is a particularly economic example. Without aiming to present a complete and realistic model, it nicely illustrates how group theory is used to predict mass matrices, and hence mixing patterns.

The basic group theory of $A_4$ is concisely outlined in Ref.~\refcite{Altarelli:2005yx}. The group $A_4$ is an alternating group, to be precise, the group of even permutations of four objects. Accordingly, it has exactly $4!/2=24\div 2 = 12$ elements. All these elements can be built out of combinations of powers of two \emph{generators}, denoted by $S$ and $T$, which fulfill the \emph{presentation rules} of $A_4$: $S^2 = T^3 = (S T)^3 = {\bf 1}$. These rules unambiguously define the group. It turns out that the group $A_4$ has exactly four \emph{irreducible representations} (irreps), cf.\ Tab.~\ref{tab:gens}, out of which three are one-dimensional ($\mathbf{1}$, $\mathbf{1'}$, $\mathbf{1''}$) and one is three-dimensional ($\mathbf{3}$). This essentially means that there are three different ways to denote the group elements by (complex) numbers, and one way to do that using $3\times 3$ matrices.\footnote{However, some of these representations are \emph{non-faithful}: the same number or matrix could be used for more than one group element.} Since we want to explain some structure in the mass matrices, we make use of the $\mathbf{3}$ in order to unify the three generations. In other words, we simply postulate that the three LH lepton doublets of the SM transform under $A_4$ just as the triplet does, while the right-handed charged leptons transform as singlets:
\begin{equation}
 L = \begin{pmatrix}
 L_e\\
 L_\mu\\
 L_\tau
 \end{pmatrix} \sim \mathbf{3},\ \ e_R \sim \mathbf{1},\ \ \mu_R \sim \mathbf{1''},\ \ \tau_R \sim \mathbf{1'}.
 \label{eq:trafos_ex}
\end{equation}
Note that this choice is, to some extent, completely \emph{ad hoc}. An experienced model builder would of course know how to assign transformation behaviors in a successful way, but we could also take the more naive viewpoint of simply choosing some assignments to see where they lead to.

\begin{table}[hp]
\tbl{Generators of $A_4$ (with $\omega \equiv e^{2\pi i/3} $).}
{\begin{tabular}{@{}lcc@{}} \toprule
Irrep & $S$ & $T$ \\
\colrule
$\mathbf{1}$ & $1$ & $1$\\
$\mathbf{1'}$ & $1$ & $\omega^2$\\
$\mathbf{1''}$ & $1$ & $\omega$\\
$\mathbf{3}$ & $\frac{1}{3} \begin{pmatrix} -1 & 2 & 2 \\ 2 & -1 & 2 \\ 2 & 2 & -1 \end{pmatrix}$ & $\begin{pmatrix} 1 & 0 & 0 \\ 0 & \omega^2 & 0 \\ 0 & 0 & \omega \end{pmatrix}$ \\
\botrule
\end{tabular}
\label{tab:gens}}
\end{table}

The group theory now comes in because in a Lagrangian we have to combine certain fields, in order to obtain the total singlet terms we can include in the Lagrangian. The easiest way to arrive at a neutrino mass is to make use of the (non-renormalizable) dimension-5 Weinberg operator~\cite{Weinberg:1979sa},
\begin{equation}
 \mathcal{L}_5 = \frac{-y_{ij}}{\Lambda} (\overline{L_i^c} H) i\sigma_2 (H L_j),
 \label{eq:Wein_op}
\end{equation}
where $\Lambda$ denotes some high energy scale and $\sigma_2$ is the second Pauli matrix. If the SM Higgs field $H$ transforms trivially under $A_4$, \emph{i.e.} $H\sim \mathbf{1}$, then this operator, which leads to a neutrino mass if $H$ obtains a VEV, is in terms of $A_4$ nothing else than a \emph{direct product} of two triplets, $\mathbf{3} \otimes \mathbf{3}$. To know precisely what this product is, we need the help of group theory. In the case of $A_4$, it tells us that two triplets $a=(a_1, a_2, a_3)^T$ and $b=(b_1, b_2, b_3)^T$ are combined as follows,
\begin{equation}
 \mathbf{3} \otimes \mathbf{3} = \mathbf{1} \oplus \mathbf{1'} \oplus \mathbf{1''} \oplus \mathbf{3}_S \oplus \mathbf{3}_A,
 \label{eq:dir_prod}
\end{equation}
where the resulting singlets are given by\cite{Altarelli:2005yx},
\begin{eqnarray}
 \mathbf{1} &=& (a_1 b_1 + a_2 b_3 + a_3 b_2),\nonumber\\
 \mathbf{1'} &=& (a_1 b_2 + a_2 b_1+ a_3 b_3),\nonumber\\
 \mathbf{1''} &=& (a_1 b_3 + a_2 b_2+ a_3 b_1),
 \label{eq:dir_prod_sing}
\end{eqnarray}
and the two resulting triplets are conveniently decomposed into a symmetric and an anti-symmetric expression,
\begin{eqnarray}
 \mathbf{3}_S &=& \frac{1}{3} (2 a_1 b_1 - a_2 b_3 - a_3 b_2, - a_1 b_2 - a_2 b_1 + 2 a_3 b_3, - a_1 b_3 + 2 a_2 b_2  - a_3 b_1)^T,\nonumber\\
 \mathbf{3}_A &=& \frac{1}{2} (a_2 b_3 - a_3 b_2, a_1 b_2 - a_2 b_1, a_1 b_3 - a_3 b_1)^T.
 \label{eq:dir_prod_trip}
\end{eqnarray}
Note that this decomposition can always be done, and the two triplets are perfectly indistinguishable from an $A_4$ point of view, which we can check easily by acting with $S_{\mathbf{3}}$ and $T_{\mathbf{3}}$ onto $a=(a_1, a_2, a_3)^T$ and $b=(b_1, b_2, b_3)^T$.

So far, only the trivial singlet combination, the first term in Eq.~\eqref{eq:dir_prod_sing}, would be allowed in the Lagrangian, which would lead to a relatively boring neutrino mass matrix:
\begin{equation}
 \mathcal{L}_5 \to \frac{-y v^2}{\Lambda} \left( \overline{\nu_{e L}^c} \nu_{e L} + \overline{\nu_{\mu L}^c} \nu_{\tau L} + \overline{\nu_{\tau L}^c} \nu_{\mu L} \right) = \frac{-y v^2}{\Lambda} (\overline{\nu_{e L}^c}, \overline{\nu_{\mu L}^c}, \overline{\nu_{\tau L}^c}) \begin{pmatrix} 1 & 0 & 0\\ 0 & 0 & 1 \\ 0 & 1 & 0 \end{pmatrix} \begin{pmatrix} \nu_{e L}\\ \nu_{\mu L}\\ \nu_{\tau L}\end{pmatrix}.
 \label{eq:mat_simp}
\end{equation}
Even worse, the charged leptons would be completely massless, since any combination of an $A_4$ triplet $a=(a_1, a_2, a_3)$ with a singlet $c$ would generate yet another triplet:
\begin{eqnarray}
 \mathbf{3} \otimes \mathbf{1} &=& (a_1 c, a_2 c, a_3 c),\nonumber\\
 \mathbf{3} \otimes \mathbf{1'} &=& (a_3 c, a_1 c, a_2 c),\nonumber\\
 \mathbf{3} \otimes \mathbf{1''} &=& (a_2 c, a_3 c, a_1 c).
 \label{eq:dir_prod_tripsing}
\end{eqnarray}
All these combinations would transform non-trivially under $A_4$ and hence cannot be allowed terms in a Lagrangian.

This problem can be cured by introducing new fields that are again called \emph{flavons}, which are scalar SM singlets but transform non-trivially under the flavor symmetry. Furthermore, flavons obtain VEVs which can be used to achieve a certain structure in the mass matrices. In the model discussed in Ref.~\refcite{Chen:2009um}, the authors suggest to use one triplet flavon, $\phi_S\sim \mathbf{3}$, as well as one singlet flavon, $u\sim \mathbf{1}$, which are taken to obtain the following VEVs:
\begin{equation}
 \langle \phi_S \rangle = \begin{pmatrix} 1\\ 1\\ 1\end{pmatrix} \alpha_S \Lambda_F,\ \ \ \langle u \rangle = \alpha_0 \Lambda_F,
 \label{eq:ex_VEVs}
\end{equation}
where $\Lambda_F$ is the scale where the discrete flavor symmetry $A_4$ is broken, and the parameters $\alpha_{S,0}$ are coefficients of $\mathcal{O}(1)$. The triplet flavon field can be used to allow for a charged lepton mass term, since $\mathbf{3} \otimes \mathbf{3} = \mathbf{1} \oplus \mathbf{1'} \oplus \mathbf{1''}  \oplus ...$, cf.\ Eq.~\eqref{eq:dir_prod}, and in $A_4$ we furthermore have~\cite{Chen:2009um},
\begin{equation}
 \mathbf{1}^a \otimes \mathbf{1}^b = \mathbf{1}^{(a+b)\ {\rm mod}\ 3},
 \label{eq:further_prod}
\end{equation}
where $a,b=0,1,2$ and $(\mathbf{1}^0, \mathbf{1}^1, \mathbf{1}^2) \equiv (\mathbf{1}, \mathbf{1'}, \mathbf{1''})$. For the charged leptons, we obtain
\begin{eqnarray}
 && \mathcal{L} \supset \overline{L} \frac{\langle \phi_S \rangle}{\Lambda_F} H e_R  \nonumber\\
 && \to \alpha_S \left[ k_e \underbrace{(\overline{L_e} + \overline{L_\mu} + \overline{L_\tau}) }_{{\rm from}\ \mathbf{1}} H e_R + k_\mu \underbrace{ (\overline{L_e} + \overline{L_\mu} + \overline{L_\tau}) }_{{\rm from}\ \mathbf{1'}} H \mu_R + k_\tau \underbrace{ (\overline{L_e} + \overline{L_\mu} + \overline{L_\tau})}_{{\rm from}\ \mathbf{1''}} H \tau_R \right]\nonumber\\
 && \to v \alpha_S (\overline{e_L}, \overline{\mu_L}, \overline{\tau_L})
 \begin{pmatrix}
 k_e & k_\mu & k_\tau\\
 k_e & k_\mu & k_\tau\\
 k_e & k_\mu & k_\tau
 \end{pmatrix}
 \begin{pmatrix}
 e_R\\
 \mu_R\\
 \tau_R
 \end{pmatrix} .
 \label{eq:ex_CL}
\end{eqnarray}
Note that the \emph{VEV alignment} from Eq.~\eqref{eq:ex_VEVs} makes all contributions look very similar; nevertheless, each contribution is a singlet for itself, and hence carries an individual coefficient in the Lagrangian. Although the inclusion of the flavon VEV leads to a non-zero mass matrix, there will be still two massless fermion states, since the rank of the mass matrix in Eq.~\eqref{eq:ex_CL} is one. This suggests that actually one would like to introduce yet another flavon to give masses to the charged leptons, just as done in Ref.~\refcite{Chen:2009um} -- see that reference for more details.

For the neutrinos, in turn, we can make use of both flavons. Since a Majorana mass matrix as the one originating from Eq.~\eqref{eq:Wein_op} must be symmetric, cf.\ \ref{sec:seesaw_matrix}, the only triplet to combine with the flavon $\phi_S$ is the symmetric one. Applying Eq.~\eqref{eq:dir_prod_trip} yields
\begin{eqnarray}
 &&\frac{1}{\Lambda} (\overline{L^c} H) i\sigma_2 (H L) \frac{\langle \phi_S \rangle}{\Lambda_F} \supset \frac{\alpha_S v^2}{3 \Lambda} \left( 2 \overline{\nu_{e L}^c} \nu_{e L} - \overline{\nu_{\mu L}^c} \nu_{\tau L} - \overline{\nu_{\tau L}^c} \nu_{\mu L} - \overline{\nu_{e L}^c} \nu_{\mu L} - \overline{\nu_{\mu L}^c} \nu_{e L} \right.\nonumber\\
 && \left. + 2 \overline{\nu_{\tau L}^c} \nu_{\tau L} - \overline{\nu_{e L}^c} \nu_{\tau L} + 2 \overline{\nu_{\mu L}^c} \nu_{\mu L}  - \overline{\nu_{\tau L}^c} \nu_{e L} \right) =\nonumber\\
 && = \frac{\alpha_S v^2}{3 \Lambda} (\overline{\nu_{e L}^c}, \overline{\nu_{\mu L}^c}, \overline{\nu_{\tau L}^c})
 \begin{pmatrix}
 2 & -1 & -1\\
 -1 & 2 & -1 \\
 -1 & -1 & 2
 \end{pmatrix}
 \begin{pmatrix}
 \nu_{e L}\\
 \nu_{\mu L}\\
 \nu_{\tau L}
 \end{pmatrix}.
 \label{eq:ex_nu_1}
\end{eqnarray}
Similarly, the trivial singlet combination in Eq.~\eqref{eq:Wein_op} can be married with the singlet flavon $u$, yielding
\begin{eqnarray}
 &&\frac{1}{\Lambda} (\overline{L^c} H) (H L) \frac{\langle u \rangle}{\Lambda_F} \supset \frac{\alpha_0 v^2}{\Lambda} \left( \overline{\nu_{e L}^c} \nu_{e L} + \overline{\nu_{\mu L}^c} \nu_{\tau L} + \overline{\nu_{\tau L}^c} \nu_{\mu L} \right) \nonumber\\
 && = \frac{\alpha_0 v^2}{\Lambda} (\overline{\nu_{e L}^c}, \overline{\nu_{\mu L}^c}, \overline{\nu_{\tau L}^c})
 \begin{pmatrix}
 1 & 0 & 0\\
 0 & 0 & 1 \\
 0 & 1 & 0
 \end{pmatrix}
 \begin{pmatrix}
 \nu_{e L}\\
 \nu_{\mu L}\\
 \nu_{\tau L}
 \end{pmatrix}.
 \label{eq:ex_nu_2}
\end{eqnarray}
Making the redefinition $\alpha_S/3 \to \alpha_S$, we finally arrive at
\begin{equation}
 \frac{1}{\Lambda} (\overline{L^c} H) (H L) \frac{\langle u \rangle}{\Lambda_F} \supset \frac{v^2}{\Lambda} (\overline{\nu_{e L}^c}, \overline{\nu_{\mu L}^c}, \overline{\nu_{\tau L}^c})
 \begin{pmatrix}
 \alpha_0+ 2 \alpha_S & -\alpha_S & -\alpha_S\\
 -\alpha_S & 2 \alpha_S & \alpha_0 -\alpha_S \\
 -\alpha_S & \alpha_0 -\alpha_S & 2 \alpha_S
 \end{pmatrix}
 \begin{pmatrix}
 \nu_{e L}\\
 \nu_{\mu L}\\
 \nu_{\tau L}
 \end{pmatrix}.
 \label{eq:ex_nu_3}
\end{equation}
We have now seen how the structure of the mass matrices in Eqs.~\eqref{eq:ex_CL} and~\eqref{eq:ex_nu_3} arose. The next step would be to diagonalize both mass matrices, and to determine the PMNS matrix by the mismatch of the two diagonalizations. A very detailed treatment of this point can be found in Refs.~\refcite{Antusch:2005gp,MPT}, and some more information is also given in \ref{sec:seesaw_diag}.

Let us finish by noting that the neutrino mass matrix in Eq.~\eqref{eq:ex_nu_3} is \emph{form-diagonalizable}, \emph{i.e.}, it can always be diagonalized by the same matrix $U_\nu$, no matter which values the parameters $\alpha_{S,0}$ have. As it turns out in the particular case at hand, the matrix that does this job is \emph{tri-bimaximal}~\cite{Harrison:2002er}, which would lead to mixing angles $\theta^\nu_{12}=\arctan \left( \frac{1}{\sqrt{2}} \right)$, $\theta^\nu_{13}=0$, $\theta^\nu_{23}= \frac{\pi}{4}$.\footnote{Note there here we make use of the notation $\theta^\nu_{ij}$ to denote the mixing stemming from the neutrino sector only, cf.\ \ref{sec:seesaw_diag} for more details.} Since a vanishing physical mixing angle $\theta_{13}$ is, however, excluded, cf.\ Eqs.~\eqref{eq:angles_NH} and~\eqref{eq:angles_IH}, we will need a non-trivial charged lepton mixing, $U_e \neq {\bf 1}$, in order to obtain a phenomenologically viable mixing matrix $U_{\rm PMNS} = U_e^\dagger U_\nu$. This would, again, make it necessary to modify the charged lepton mass matrix, by introducing even more fields that ideally should not destroy the form of the neutrino mass matrix. The reader should at latest now appreciate that finding a working model is indeed a non-trivial task.

Note finally that in our calculation we have relied heavily on the chosen \emph{vacuum alignment}, cf.\ Eq.~\eqref{eq:ex_VEVs}. However, we have not written down a full scalar potential including all allowed terms involving the SM-Higgs as well as both flavon fields. Even if we succeed in writing down such a potential, it is a highly non-trivial task to show that Eq.~\eqref{eq:ex_VEVs} is at all a minimum of this potential (showing the even stronger condition that it is actually the global minimum is a task that is practically impossible to tackle in many realistic models). The typical way to proceed is to do a numerical study in order to show that the alignment chosen is at least fine for a certain choice of parameters. While this task is a highly non-trivial one in general, it is nevertheless doable in the relatively minimal model presented here. For the interested reader, we discuss the vacuum alignment for the example model described here in \ref{sec:vacuum}\\

We will now see how the concepts we have discussed come into play when trying to find working models that include keV sterile neutrinos. Not only will we have to convincingly explain all neutrino data, but we will also have to be careful to be in agreement with certain bounds that are specific for keV neutrinos, such as the non-observation of the astrophysical X-ray line (cf.\ Secs.~\ref{sec:astro_X-ray} and~\ref{sec:keV_general_X}), and in particular the generation of the correct DM abundance.

\section{\label{sec:keV_general}General features of settings with keV Sterile Neutrinos}	

In this chapter, we will summarize a few theorems, conditions, and other features of working models for keV sterile neutrinos. This collection will help us to appreciate many of the steps which are implicitly undertaken in concrete models. To mention one example, a typical argument used against keV sterile neutrino models is that one should be skeptical if the seesaw mechanism (cf.\ \ref{sec:seesaw}) actually works, since after all one would have to divide by a small mass of $\mathcal{O}({\rm keV})$ when applying the seesaw formula, Eq.~\eqref{eq:seesawI_2}. While this is certainly a dangerous enterprise at first sight, one can actually show that in any working model of keV sterile neutrinos which respects the X-ray bound, this is no problem. A couple of such properties, to be respected by a large class of models, will be presented in the following.

Note that some of the proofs presented in this section are a bit formal and can be skipped without loss of information. However, if the reader is inclined to reproduce them, it may be helpful to consult \ref{sec:seesaw} for some technical details.

\subsection{\label{sec:keV_general_X}The X-ray bound}	

Although already mentioned in Sec.~\ref{sec:astro_X-ray}, we will now stress the importance of the X-ray bound, arising from the decay $N_i \to \sum_\alpha \nu_\alpha \gamma$ and the non-observation of the corresponding astrophysical line. Note that we have now generalized the keV sterile neutrino $N_i$ to originate from any generation $i$, and that we have furthermore indicated that any active neutrino flavor $\alpha = e, \mu, \tau$ could be produced. However, since we can never observe this flavor we have to sum over all possibilities.

In order to get a better understanding of the constraint arising from this decay, we have depicted the two lowest order 1-loop Feynman diagrams for the reaction $N_i \to \nu_\alpha \gamma$ in Fig.~\ref{fig:Nu_nu_gamma}. As already mentioned in Sec.~\ref{sec:astro_X-ray}, the diagrams look very similar to a $\mu \to e \gamma$ transition, which is even discussed in some textbooks (see, \emph{e.g.}, Ref.~\refcite{ChengLi}). More formally, the process is a special case of the general radiative fermionic decay $f_1 \to f_2 \gamma$~\cite{Lavoura:2003xp}. Since the charged lepton and the $W$-boson inside the loop both are electrically charged, the external photon can couple to either of them, which is why we have actually two diagrams for each active neutrino flavor $\nu_\alpha$. The decisive point in the diagrams is the left vertex: the mass eigenstate sterile neutrino $N_i$ is not purely right-handed, but it has a small left-handed active admixture, which is parametrized by the (small) active-sterile mixing angle $\theta_{\alpha i}$, cf.\ \ref{sec:seesaw_diag}. Hence, there is a possibility for this mass eigenstate to couple to a charged lepton $e_\alpha$ and a $W$-boson, but this coupling is suppressed by a factor $\theta_{\alpha i}$ and the amplitude $\mathcal{M}_{\alpha i}$ for $N_i \to \nu_\alpha \gamma$ is proportional to exactly the same factor, as indicated in the diagrams.

\begin{figure}[pb]
\centerline{
\begin{tabular}{lr}
\psfig{file=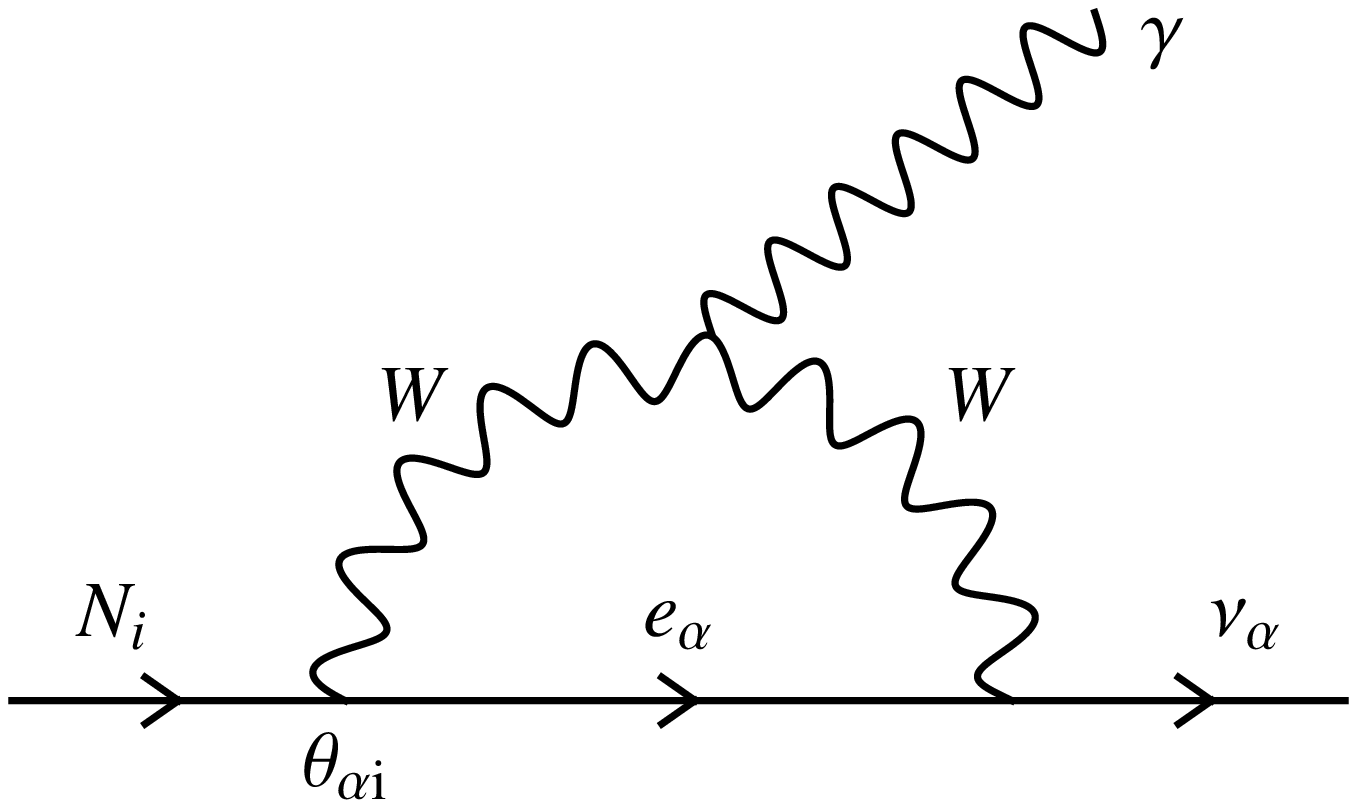,width=6.0cm} & \psfig{file=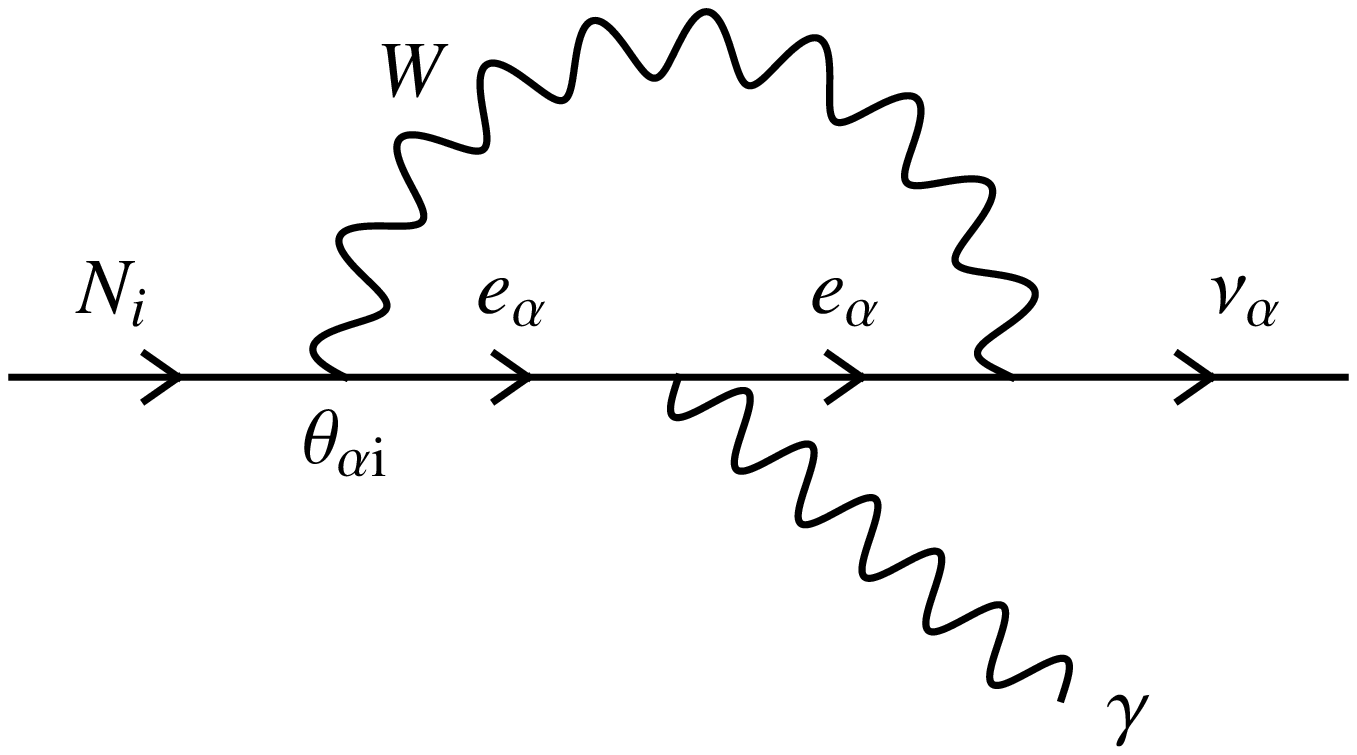,width=6.0cm}
\end{tabular}
}
\vspace*{8pt}
\caption{\label{fig:Nu_nu_gamma}The Feynman diagrams for the radiative decay of the sterile neutrino, $N_i \to \nu_\alpha \gamma$.}
\end{figure}

Note that we can produce any active neutrino flavor $\alpha$ with a probability proportional to $|\mathcal{M}_{\alpha i}|^2$. However, since we have no way to measure the final state flavor $\alpha$, we have to sum over all possibilities, $\alpha = e, \mu, \tau$. Note that, although we cannot determine the flavor of the final state active neutrino, we nevertheless know by laboratory experiments that the different flavors $\nu_{e,\mu,\tau}$ are fundamentally distinct. This is why we have to sum over probabilities (\emph{incoherent summation}) rather than over amplitudes (\emph{coherent summation}), differently from what we did for, \emph{e.g.}, neutrino-less double beta decay, cf.\ Sec.~\ref{sec:neutrino_mix_masses}. Thus, the decay rate $N_i \to \nu \gamma$ into all possible active neutrino flavors must be proportional to the following quantity,
\begin{equation}
 \sum_\alpha |\theta_{\alpha i}|^2,
 \label{eq:as-constraints}
\end{equation}
which is called the \emph{$i$-th active-sterile mixing angle square} and usually denoted as $\theta_i^2$. This is the quantity we can put an upper bound on by a non-observation of the monoenergetic photon $\gamma$.

The precise bound originating from different satellite experiments can be found in Refs.~\refcite{Canetti:2012kh,Canetti:2012vf} (based on Refs.~\refcite{Dolgov:2000ew,Abazajian:2001vt,Boyarsky:2005us,Boyarsky:2006fg,RiemerSorensen:2006fh,Abazajian:2006yn,Watson:2006qb,Boyarsky:2006ag,Abazajian:2006jc,Boyarsky:2007ay,Boyarsky:2007ge,Loewenstein:2008yi}). For our purpose, a simplified version of this bound as used in Ref.~\refcite{Boyarsky:2009ix} is perfectly sufficient:
\begin{equation}
 \theta_i^2 \lesssim 1.8\cdot 10^{-5} \left( \frac{1~{\rm keV}}{M_i} \right)^5.
 \label{eq:X_simp}
\end{equation}
Note, however, that more recent non-observations of the X-ray line for certain galaxies yield even stronger bounds, cf.\ Refs.~\refcite{Watson:2011dw,Loewenstein:2012px}. Corresponding updates of the simplified bound in Eq.~\eqref{eq:X_simp} are available~\cite{keV0nbb}.

Let us now get a more precise understanding of the connection between the active-sterile mixing and the entries in the full neutrino mass matrix. As we have just seen, the definition of the $i$-th active-sterile mixing angle $\theta_i$ is
\begin{equation}
 \theta_i^2 \equiv \sum_\alpha |\theta_{\alpha i}|^2,\ \ \ {\rm where}\ \ \ \theta_{\alpha i} \equiv \overline{U}_{\alpha,3+i} = \left[ m_D^* {M_R^{-1}}^* V_R \right]_{\alpha i}.
 \label{eq:theta1_def}
\end{equation}
Note that we have expressed the generation-dependent active-sterile mixing $\theta_{\alpha i}$ in terms of the full neutrino mixing matrix $\overline{U}$ as defined in \ref{sec:seesaw_diag}. In the basis where the RH neutrino mass matrix is diagonal (and real), $M_R={\rm diag}(M_1, M_2, M_3)$, we have $V_R = \mathbf{1}$, and the above formula simplifies to
\begin{equation}
 \theta_{\alpha i} = \sum _k {m_D^*}_{\alpha k} M_k^{-1} \delta_{k i} = \frac{(m_D^*)_{\alpha i}}{M_i}.
 \label{eq:theta1_simp}
\end{equation}
This can be simplified further,
\begin{equation}
 \theta_i^2 = \sum_\alpha \frac{(m_D)_{\alpha i} (m_D^*)_{\alpha i}}{M_i^2} = \frac{1}{M_i^2} \sum_\alpha (m_D^\dagger)_{i \alpha} (m_D)_{\alpha i} = \frac{(m_D^\dagger m_D)_{i i}}{M_i^2}.
 \label{eq:theta1_simp_further}
\end{equation}
We can now see how to understand the active-sterile mixing $\theta_i^2$: it is simply the proportionality factor between the sterile neutrino mass $M_i$ of $N_i$ and the corresponding active-neutrino mass $m_i$ of $\nu_i$ in a seesaw type~I framework, cf.\ \ref{sec:seesaw_diag}. Hence, it should also be intuitively clear that $\theta_i^2$ indeed parametrizes the small ``active fraction'' of the otherwise sterile mass eigenstate $N_i$, and it can thus be used to quantify how often a neutrino which is mainly sterile can decay via a coupling to active partners, such as the charged leptons and the $W$-bosons in Fig.~\ref{fig:Nu_nu_gamma}.

\subsection{\label{sec:keV_general_seesaw}Seesaw theorems}	

Proceeding with some more technicalities, we will now shortly discuss two very useful theorems which hold for seesaw type~I situations.

The first theorem~\cite{Merle:2012xq} deals with the potential issue of applying the seesaw formula, Eq.~\eqref{eq:seesawI_2}, to models involving keV sterile neutrinos:

\begin{theorem}
\emph{keV seesaw practicality theorem}~\cite{Merle:2012xq}\\
For any working keV neutrino model (\emph{i.e.}, it is consistent with the X-ray bound), the seesaw formula works. This conclusion can only be altered if the decay mode $N\to \nu \gamma$ does not exist (or is not effective) for some reason.
\end{theorem}

\begin{proof}
The active-sterile mixing angle $\theta_{\alpha i}$ is given in Eq.~\eqref{eq:theta1_def}. In a basis where $M_R$ is diagonal we have $\theta_{\alpha i} = \frac{(m_D^*)_{\alpha i}}{M_i}$, cf.\ Eq.~\eqref{eq:theta1_simp}. If $M_1 = \mathcal{O}({\rm keV}) \ll M_{2,3}$, then the only critical property to be fulfilled for the seesaw to work is $|\theta_{\alpha 1}|  \ll 1$. From the rough X-ray bound, Eq.~\eqref{eq:X_simp}, in combination with the rough Lyman-$\alpha$ bound, $M_1 \gtrsim 1$~keV (cf.\ Sec.~\ref{sec:astro_WCDM}), it follows immediately: $(m_D^\dagger m_D)_{11} \lesssim 1.8\cdot 10^{-5}\ {\rm keV}^2 \left( \frac{1\ {\rm keV}}{M_1}\right)^3 \lesssim 1.8\cdot 10^{-5}\ {\rm keV}^2 \ll M_1^2$, which implies $|{m_D}_{e 1}|^2 + |{m_D}_{\mu 1}|^2 + |{m_D}_{\tau 1}|^2 \ll M_1^2$, and hence, due to the sum over absolute values, $|{m_D}_{\alpha 1}| \ll M_1$ for $\alpha = e, \mu, \tau$.
\end{proof}

Thus, as long as a model for keV sterile neutrinos respects the X-ray bound, one can always apply the seesaw formula. If a certain model violates the X-ray bound it is either excluded, in which case we do not need to talk about it, or the decay mode $N_1 \to \nu \gamma$ is somehow forbidden or strongly suppressed, \emph{e.g.} by a symmetry stabilizing the keV neutrino $N_1$, cf.\ Ref.~\refcite{Allison:2012qn}. In the latter case one might have to check the validity of the seesaw formula, if applicable, but in most models one can blindly use it if the X-ray bound is known to make no problems. Note that the above argumentation is not changed if there are additional type~II contributions to the light neutrino mass matrix, since any terms in $m_L$ do not contribute to the active-sterile mixing, cf.\ \ref{sec:seesaw_diag}.

The second theorem was first mentioned in Ref.~\refcite{Schechter:1980gr} and later on brought into a more modern form in Ref.~\refcite{Xing:2007uq}. It is very useful for checking the consistency of diagonalizations of complicated neutrino mass matrices:

\begin{theorem}
\emph{Seesaw fairplay rule}~\cite{Schechter:1980gr,Xing:2007uq}\\
In a seesaw type~I setting with $p$ LH doublets and $q<p$ RH neutrinos, one obtains at most $q$ light neutrinos with a non-zero mass, while at least $(p-q)$ neutrinos remain massless. In particular, in a setting with 3 LH doublets and only 2 RH neutrinos, one light neutrinos is exactly massless.
\end{theorem}

\begin{proof}
The proof of this theorem is trivial if one visualizes the structure of the corresponding matrix. In Eq.~\eqref{eq:seesawI_1}, the upper left zero block is a $p\times p$ matrix. In order for the first $p$ column vectors of the full mass matrix to be linearly independent, one would need $q\geq p$ components to form enough unit vectors. Due to $q<p$, this is never possible and at most $q$ of the first $p$ column vectors can be linearly independent. Even if the remaining $q$ column vectors consisting of the right $m_D$ and $M_R$ blocks are all linearly independent, the maximum rank of the full mass matrix is $2q$. Hence, the matrix has at most $2q$ non-zero mass eigenvalues, and at least $(p-q)$ of the light neutrinos remain massless.
\end{proof}

This theorem is exactly what is at work in, \emph{e.g.}, the model to be presented in Sec.~\ref{sec:keV_FN_mixed}: since the keV sterile neutrino is decoupled from the mass matrix, the lightest neutrino remains exactly massless in that approximation, cf.\ Eq.~\eqref{eq:BRZ_eignvals}. Note that, however, a small but non-zero mass is nevertheless generated at 2-loop level by diagrams involving only SM-fields~\cite{Davidson:2006tg}.

\subsection{\label{sec:keV_general_cancellation}Cancellations in neutrino-less double beta decay}	

We also want to present a theorem discovered in Refs.~\refcite{Blennow:2010th,Barry:2011fp}, which reveals an interesting effect of light sterile neutrinos on the effective mass $|m_{ee}|$ in neutrino-less double beta decay, cf.\ Eq.~\eqref{eq:m_ee}. For a definition of \emph{form dominance}, see \ref{sec:seesaw_FD}.

\begin{theorem}
\emph{$0\nu\beta\beta$ cancellation theorem}~\cite{Blennow:2010th,Barry:2011fp}\\
In a seesaw type~I situation where all RH neutrinos have mass eigenvalues below $100$~MeV, the effective neutrino mass $|m_{ee}|$ measured in neutrino-less double beta decay vanishes exactly. If the neutrino mass matrix is in addition form dominant, then this cancellation happens piecewise for each generation. In such a case, even having only a few RH neutrino masses below $100$~MeV leads to partial cancellations in $|m_{ee}|$.
\end{theorem}

\begin{proof}
Whenever the mass $m_k$ of a neutrino is below the nuclear momentum transfer in $0\nu\beta\beta$, $|\vec{q}| \simeq 100$~MeV, it contributes to the decay amplitude by light neutrino exchange~\cite{Rodejohann:2011mu}, and this contribution is given by $\overline{U}_{e k}^2 m_k$. Hence, for three LH and three RH neutrinos with all masses below $|\vec{q}|$, the effective mass is given by $|m_{ee}|=\left|\sum^3_{k=1} \overline{U}^2_{e k} m_k + \sum^3_{k=1} \overline{U}^2_{e,3+k} M_k \right| = \left[M_\nu^{6 \times 6}\right]_{ee}$, which is zero by definition in a type~I seesaw setting. If the mass matrix is in addition form dominant, cf.\ \ref{sec:seesaw_FD}, then Eq.~\eqref{eq:FD_9} holds in the basis where $M_R$ is diagonal. One can immediately conclude that $(m_D^*)_{e k} = i \sqrt{m_k M_k} U_{e k}$ and that $\overline{U}_{e,3+k} = \theta_{e k} = i \sqrt{\frac{m_k}{M_k}} U_{e k}$, cf.\ Eq.~\eqref{eq:theta1_def}. One readily obtains $\overline{U}^2_{e,3+k} M_k = \theta_{e k}^2 M_k = - \frac{m_k}{M_k} U_{e k}^2 M_k = - U_{e k}^2 m_k$.
\end{proof}

Note that this theorem implies in particular that, in models with one keV sterile neutrino $N_1$ and two other neutrinos $N_{2,3}$ with $M_{2,3} > |\vec{q}| \simeq 100$~MeV, the effective mass for a form dominant setting is given by~\cite{keV0nbb}
\begin{equation}
 |m_{ee}| = |m_2 s_{12}^2 c_{13}^2 + m_3 s_{13}^2 e^{i(\beta-\alpha)}|,
 \label{eq:0nbb-canc}
\end{equation}
and similar if the keV sterile neutrino originates from a different generation. This cancellation was not recognized in Refs.~\refcite{Bezrukov:2005mx,Asaka:2011pb}, which contained the first discussion of $0\nu\beta\beta$ in the context of the $\nu$MSM.\footnote{In particular, this references focused on the ``standard'' production mechanisms (DW and SF) of keV sterile neutrino DM, in which case the lower bounds on the keV mass are strongest.} Instead, it was shown that the contribution from the keV sterile neutrino is tiny, which is why it was neglected. Hence, the cancellation could not have possibly been taken into account in these references. In particular, the formula for $|m_{ee}|$ obtained in in Ref.~\refcite{Bezrukov:2005mx} for the case of inverted mass ordering would be wrong if indeed $N_1$ was the keV neutrino and form dominance held: instead of the standard contribution, as obtained in Ref.~\refcite{Lindner:2005kr}, the correct expression would be given by $|m_{ee}| \approx \sqrt{|\Delta m_{31}^2|} s_{12}^2 c_{13}^2$, which is obtained from Eq.~\eqref{eq:0nbb-canc} for the limit $m_3 < m_1 < m_2$ with $m_3 \to 0$.

\subsection{\label{sec:keV_general_NoGo}A little No-Go theorem}	

The final theorem has appeared in a slightly simplified version in Refs.~\refcite{Asaka:2005an,Boyarsky:2006jm} and was later on generalized in Ref.~\refcite{Merle:2012xq}. Note that in the theorem and in the proof we will make use of the so-called \emph{Casas-Ibarra parametrization}, the details of which are explained in \ref{sec:seesaw_CI}.

\begin{theorem}
\emph{No-Go theorem}~\cite{Asaka:2005an,Boyarsky:2006jm,Merle:2012xq}\\
A quasi-degenerate light neutrino spectrum and keV sterile neutrino Dark Matter contradict each other in a seesaw type~I setting with a real Casas-Ibarra matrix $R$, if the sterile neutrino mainly decays via $N_1 \to \nu \gamma$.
\end{theorem}

\begin{proof}
Inverting the known Casas-Ibarra formula, Eq.~\eqref{eq:CI_5}, one obtains $m_D = i U^* {\rm diag}(\sqrt{m_1}, \sqrt{m_2}, \sqrt{m_3}) R^T {\rm diag}(\sqrt{M_1}, \sqrt{M_2}, \sqrt{M_3})$. Using Eq.~\eqref{eq:theta1_simp} yields $\theta_{\alpha i} = - i \sum_{k,l,m} U_{\alpha k} \sqrt{m_k} \delta_{kl} R^\dagger_{lm} \frac{1}{\sqrt{M_m}} \delta_{mi} = - i \sum_{k,l} \sqrt{\frac{m_k}{M_i}} U_{\alpha k} \delta_{kl} R^*_{il} = - i \sum_k \sqrt{\frac{m_k}{M_i}} U_{\alpha k} R^*_{ik}$. With the help of Eq.~\eqref{eq:theta1_def}, one can show that $\theta_i^2 \equiv \sum_\alpha |\theta_{\alpha i}|^2 = \sum_\alpha \sum_{k,l} \frac{\sqrt{m_k m_l}}{M_i} U_{\alpha k} U_{\alpha l}^* R^*_{ik} R_{il} = \sum_{k,l} \frac{\sqrt{m_k m_l}}{M_i} \left( \sum_\alpha U_{l \alpha}^\dagger U_{\alpha k} \right) R_{il} R^\dagger_{ki} = \sum_k \frac{m_k}{M_i} R_{ik} R^\dagger_{ki}$. For quasi-degenerate light neutrinos, one has $m_k \simeq m_0$ for $k = 1, 2, 3$, and if the orthogonal matrix $R$ is real then $R^\dagger = R^T$ and $\sum_k R_{ik} R^\dagger_{ki} = \delta_{ii} = 1$. For $i=1$, one finally arrives at $\theta_1^2 = \frac{m_0}{M_1}$. Applying the X-ray bound, Eq.~\eqref{eq:X_simp} yields a relation that cannot be fulfilled: the lower limits on $M_1$ from structure formation require $m_0 \ll \sqrt{\Delta m^2_A}$ [cf.\ Sec.~\ref{sec:astro_WCDM} and Eq.~\eqref{eq:X_simp}], while quasi-degeneracy would require $m_0 \gg \sqrt{\Delta m^2_A}$.
\end{proof}

The detailed implications of and ways around the No-Go theorem are discussed in Ref.~\refcite{Merle:2012xq}. The important point we want to stress is that $R$ is real in particular if $CP$ invariance holds. Hence, models for keV sterile neutrinos with $CP$ invariance will be in trouble if they predict a quasi-degenerate light neutrino spectrum. On the other hand, if the light neutrinos were experimentally known to have a quasi-degenerate mass pattern and if at the same time keV sterile neutrinos were the Dark Matter, this would implicitly prove the existence of $CP$ violation in the neutrino sector, as long as the light neutrino mass is generated by a seesaw type~I mechanism.\\

Equipped with these general considerations and theorems, we are now ready to enter the central discussion of this review, exploring the known models for keV sterile neutrino Dark Matter.

\section{\label{sec:keV}Models for keV Sterile Neutrinos}	

We will now turn to the core chapter of this review, where we will give a supposed-to-be-complete discussion of the models of keV neutrinos on the market. Hereby, we will not follow the historical order in which the different models have been developed, but we will rather apply a more pedagogical way through the jungle of possibilities, in order to give the less experienced reader a solid guideline, and to provide the more experienced reader with a maybe slightly different viewpoint.

In the course of the chapter, we will strongly distinguish in terminology between \emph{models} and \emph{scenarios}. In the definition used here, we will call a certain setting \emph{scenario}, whenever it can accommodate for a keV sterile neutrino, but does not give any \emph{explanation} for the appearance of the keV scale. One typical example for a \emph{scenario} is the $\nu$MSM~\cite{Asaka:2005an}, although it is actually termed ``model'' in the corresponding reference. On the other hand, we call a certain setting a \emph{model} whenever there is an explanation for the appearance of the keV scale or, rather, for a suitable mass hierarchy or the existence of a suitable new scale. One example of what we would call a \emph{model} would be a $L_e - L_\mu - L_\tau$ flavor symmetry extension of the $\nu$MSM, as presented in Refs.~\refcite{Shaposhnikov:2006nn,Lindner:2010wr,Merle:2012ya}.

The reason for this distinction in terminology is motivated by the usage of the terms in contemporary elementary particle physics, however, admittedly the terms can overlap in certain cases. We nevertheless try to stick to this terminology as strictly as possible, in order to give better guidance to less experienced readers. But we also stress that the distinction in terminology does not involve any kind of judgement. Indeed, both concepts are extremely useful in their respective contexts: \emph{scenarios}, on the one hand, are very useful for phenomenological studies, as the recent very complete treatment of the phenomenology of the $\nu$MSM clearly illustrates~\cite{Canetti:2012kh}. \emph{Models}, on the other hand, add the necessary building block of an explanation for the existence of the keV scale (or, rather, of a suitable scale hierarchy), which is important since the keV scale is to our current knowledge unrelated to any other fundamental scale in physics. Such models typically have ``back reactions'' on several phenomenological parameters, as \emph{e.g.}\ explained in Ref.~\refcite{Barry:2011wb}.

Naturally, the ultimate goal we should work towards is to combine scenarios and models, in order to identify the most promising theory for keV sterile neutrino Dark Matter. This ``ultimate'' theory should involve a convincing production mechanism as well as a proper explanation for the appearance of the keV scale, and its phenomenology should agree with all cosmological observations and particle physics experiments. On top of that, any candidate theory should yield testable predictions which can be probed in terrestrial and/or extra-terrestrial experiments. While \emph{e.g.} the $\nu$MSM does not have too many direct tests, certain models manage to entangle the existence of the keV scale with solid predictions for low energy neutrino observables such as mixing angles or effective masses. These models offer testability, not only by their phenomenological predictions but already when trying to combine them with certain production mechanisms.

\subsection{\label{sec:keV_schemes}General mass shifting schemes for keV neutrinos}	

Before discussing actual models, we first illustrate the two main approaches to \emph{explain} the appearance of the keV scale. As we have already mentioned, a mass of a few keV is, to the best of our current knowledge, not connected to any ``fundamental'' energy scale in Nature, such as the Planck scale or the electroweak scale. This is a pity, because in a quantum field theory it is intrinsically difficult to predict absolute scales. What we can predict, however, are hierarchies of scales, \emph{i.e.}, why a certain scale is larger or smaller than another one. While some of these explanations may sound a bit artificial, this is nevertheless the best way we know to predict scales at all. Furthermore, it is a well-known principle to relate ``new'' energy scales to known ones, just as the energy levels of an atom are intrinsically connected to the size of the Rydberg energy.

The two main \emph{mass shifting schemes} that are, in various versions, generically used to ``explain'' (or at least ``motivate'') the existence of a keV scale are depicted in Fig.~\ref{fig:keV-schemes}. For obvious reasons, we will call these two main schemes the \emph{Bottom-up scheme} and the \emph{Top-down scheme}. These two schemes nicely illustrate the two main ways to arrive at the ``in-between'' scale of a few keV: either the mass of the sterile neutrino that is supposed to be the Dark Matter is actually zero, but this \emph{natural} value is corrected by some model-specific mechanism to yield a non-zero mass of $\mathcal{O}({\rm keV})$. Or the fundamental mass $M_R$ of the RH-neutrinos is actually much higher, and for some reason the mass of one sterile neutrino is suppressed to yield a physical mass of only a couple of keV.

\begin{figure}[pb]
\centerline{
\begin{tabular}{lr}
\psfig{file=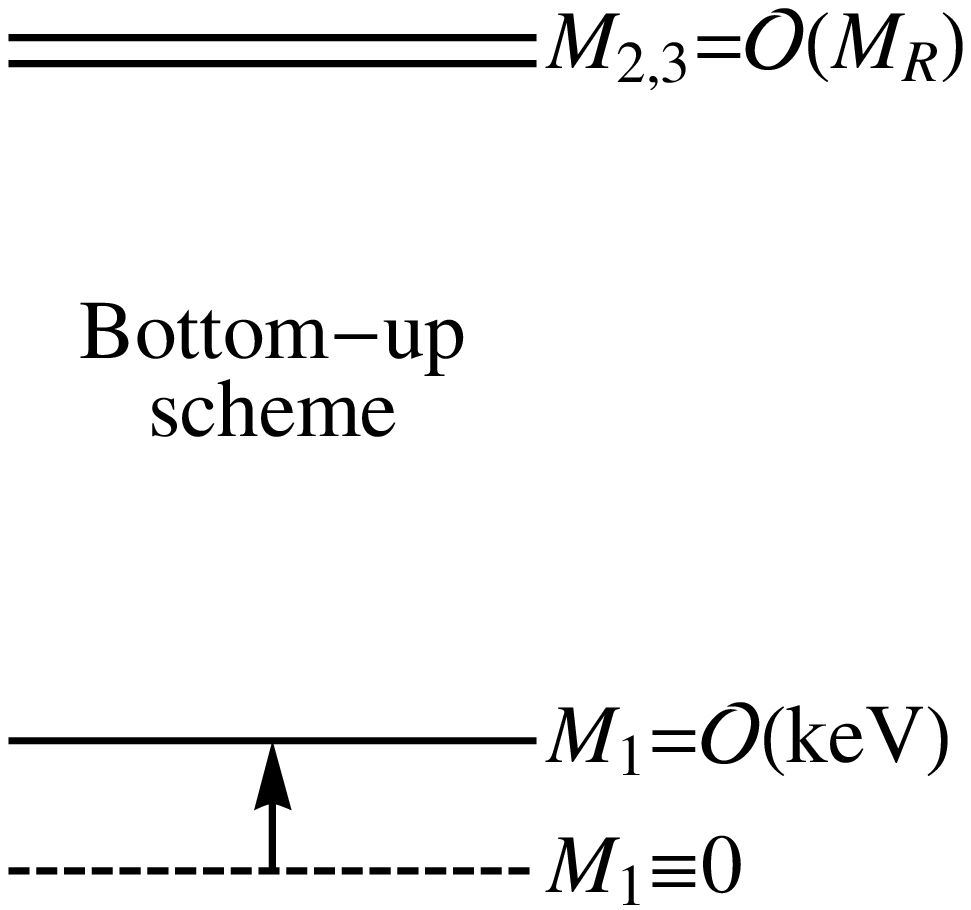,width=5.5cm} & \psfig{file=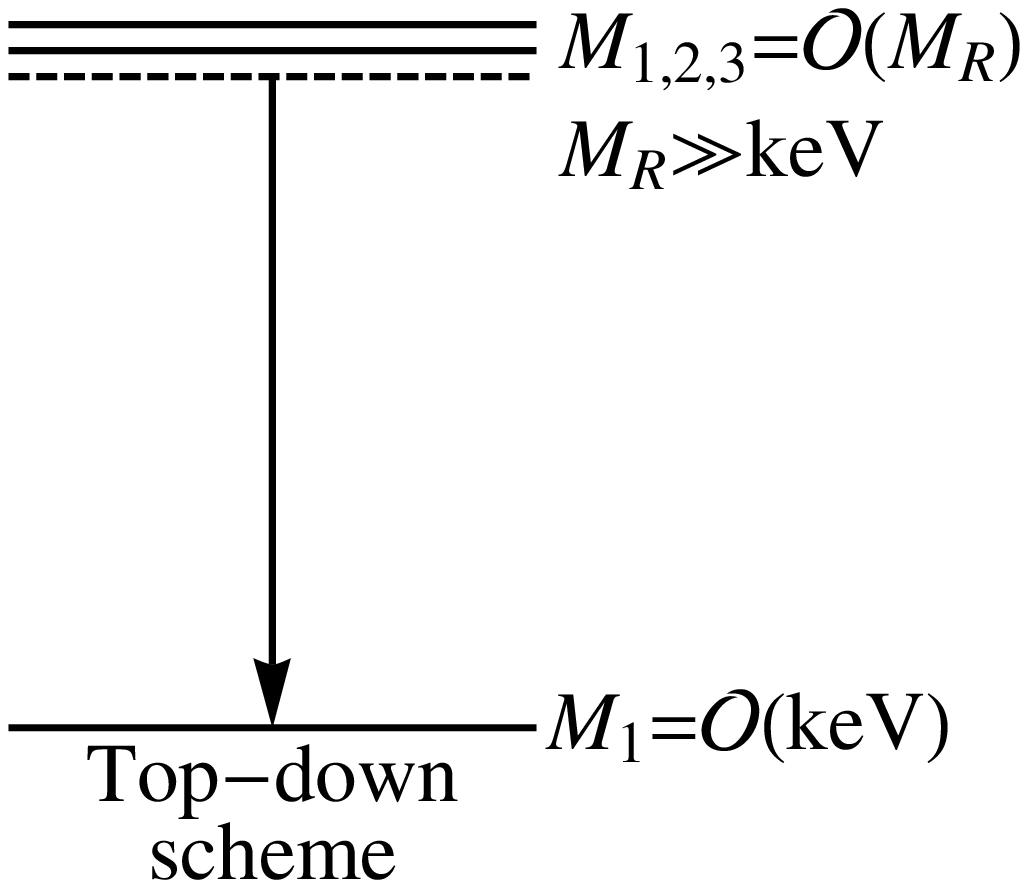,width=5.5cm}
\end{tabular}
}
\vspace*{8pt}
\caption{\label{fig:keV-schemes}The two generic mass shifting schemes for keV sterile neutrinos, in a setting with three right-handed neutrinos. Typically, $N_1$ is taken to be the keV sterile neutrino, but there are models where this is not true.}
\end{figure}

Note that we have, in passing, used the terms \emph{right-handed neutrino} and \emph{sterile neutrino} to be practically equivalent. This is okay, as long as we are talking about SM charges only: the right-handed neutrino is a total singlet under the SM gauge group, and hence it is sterile. However, we have to keep in mind that there are two loopholes in this terminology:
\begin{itemize}

\item First, the physical particle (which is a mass eigenstate!) will actually be neither purely right-handed nor purely sterile, due to the structure of the mass terms and the active-sterile mixing, cf.\ \ref{sec:seesaw_diag}. This implies that any mass eigenstate is always a superposition of a left-handed [active/$SU(2)$ doublet] and a right-handed [sterile/$SU(2)$ singlet] state. Nevertheless, due to the small mixing between active and sterile states, it is common to refer to the SM-like neutrino mass eigenstates $\nu_i$ ($i=1,2,3$) as \emph{active neutrinos} (even though they do have small sterile admixtures) and to refer to the additional (often heavier) mass eigenstates $N_i$ ($i=1,2,3$) as \emph{sterile neutrinos} (even though they do have small active admixtures). In turn, the fields $\nu_{L \alpha}$ ($\alpha=1,2,3$) and $N_{R \alpha}$ ($\alpha=1,2,3$) in the Lagrangian (which are not the physical particles but only the fundamental ingredients of the theory) are referred to as \emph{left-} and \emph{right-handed neutrinos}, respectively, according to the standard terminology. Although this terminology is unambiguous, the terms are often used in a more or less equivalent manner in the literature, and in many cases one has to conclude from the context which physical meaning is actually referred to.

\item Second, even though the term \emph{sterile} is used, this only refers to SM-interactions. As soon as we go beyond the SM by extending the gauge group, the RH-neutrinos will not be total singlets anymore, in general. For example, in left-right symmetric models~\cite{Mohapatra:1974hk,Deshpande:1990ip}, where the $SU(2)_L \times U(1)_Y$ symmetry of the SM is extended to $SU(2)_L \times SU(2)_R \times U(1)_{B-L}$, right-handed neutrinos are charged non-trivially under the $SU(2)_R \times U(1)_{B-L}$ subgroup. Furthermore, even if the gauge group is not extended, the RH neutrinos can easily have relatively strong interactions with non-SM scalar particles, as for example in the scotogenic model~\cite{Ma:2006km}.

\end{itemize}

Nevertheless, as done in the bulk of the literature, we will use both terms in a pretty analogous manner, in order not to hinder the flow of the text. Keeping the above two remarks in mind is perfectly sufficient for the purpose of understanding the models.

In addition we want to remark that in the schemes in Fig.~\ref{fig:keV-schemes} we have assumed the existence of three right-handed neutrinos $N_{1,2,3}$, which have masses $M_{1,2,3}$. Out of these, our goal is to find an explanation for why $M_1$ (or whichever is the keV sterile neutrino mass) has a value of a few keV. Then, the corresponding mass eigenstate field $N_1$ (or $N_{2,3}$) is referred to as the \emph{keV sterile neutrino}. However, in the literature one can also find examples with more or less than three right-handed neutrinos (see Ref.~\refcite{Abazajian:2012ys} for an exhaustive collection). While in such cases the schemes in Fig.~\ref{fig:keV-schemes} would of course have to be altered (by adding or removing some of the energy levels), the basic principles illustrated in the two schemes nevertheless remain the same: in practically all cases, either a zero mass is lifted or a larger mass is suppressed, unless no completely new scale unrelated to anything else is introduced and then more or less arbitrarily set to $\mathcal{O}({\rm keV})$.

Equipped with the imagination of the schemes, as well as the terminology to be used, we are now prepared to enter the zoo of the actual models, which all attempt to find an explanation for the existence of the keV scale.

\subsection{\label{sec:keV_FN}Models involving the Froggatt-Nielsen mechanism}	

First we will discuss models involving the Froggatt-Nielsen mechanism. There exist \emph{pure} FN models, which use only FN, and \emph{mixed} FN models, which employ the FN mechanism in combination with, \emph{e.g.}, a discrete flavor symmetry.

\subsubsection{\label{sec:keV_FN_pure}Pure FN models}	

A very complete discussion on how to build models for keV sterile neutrinos using the FN mechanism was given in Ref.~\refcite{Merle:2011yv}. As we have already seen in Sec.~\ref{sec:neutrino_modeling_FN}, the FN mechanism can modify a mass matrix or Yukawa coupling matrix by powers of a small number $\lambda$. Hence, we can immediately conclude that FN inspired models resemble a typical top-down setting: the ``natural'' values of the masses receive generation-dependent suppressions, which lead to small values of the physical masses, as illustrated in Fig.~\ref{fig:FN-scheme}.

\begin{figure}[pb]
\centerline{
\psfig{file=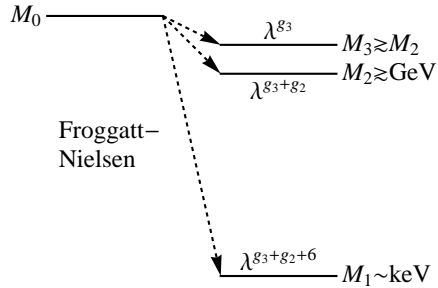,width=6.0cm}
}
\vspace*{8pt}
\caption{\label{fig:FN-scheme}Mass shifting scheme of the FN mechanism. The generation-dependent FN-charges $g_i$ strongly suppress the lightest sterile neutrino mass $M_1$, which clearly resembles the top-down scheme, cf.\ right panel of Fig.~\ref{fig:keV-schemes}. (Figure similar to Fig.~2 in Ref.~\cite{Merle:2011yv}.)}
\end{figure}

The relatively general approach discussed in Ref.~\refcite{Merle:2011yv} uses the same ingredients as the models presented in Refs.~\refcite{Kanemura:2007yy,Kamikado:2008jx}. The minimal particle content to arrive at a model with the power of at least semi-analytical predictions is given by the SM leptons, together with three right-handed neutrinos $N_{1,2,3}$ and two flavons $\Theta_{1,2}$, which are charged under the $U(1)_{\rm FN}$ symmetry as well as an auxiliary $Z_2$ parity:\footnote{As explained in Ref.~\refcite{Kanemura:2007yy}, this auxiliary symmetry is necessary to have a model with $CP$ violation.}
\begin{eqnarray}
 \Theta_{1,2}: && (\theta_1, \theta_2;+,-),\nonumber\\
 L_{1,2,3}: && (f_1, f_2, f_3;+,+,-),\nonumber\\
 \overline{e_{1,2,3}}: && (k_1, k_2, k_3;+,+,-),\nonumber\\
 \overline{N_{1,2,3}}: && (g_1,g_2, g_3; +,+,-).
 \label{eq:FN-assignments_1}
\end{eqnarray}
The most general (\emph{i.e.}, seesaw type~II) Lagrangian that leads to masses in the lepton sector is then given by
\begin{eqnarray}
 \mathcal{L} &=& 
-\sum_{a,b,i,j}^{a+b=k_i+f_j} Y_e^{ij}\,\overline{e_{iR}}\,H\,L_{jL}\,\lambda_1^a \lambda_2^b +h.c.\
-\sum_{a,b,i,j}^{a+b=g_i+f_j} Y_D^{ij}\,\overline{N_{iR}}\,\tilde{H}\,L_{jL}\,\lambda_1^a \lambda_2^b +h.c. \\
 && - \sum_{a,b,i,j}^{a+b=f_i+f_j} \frac{1}{2} \overline{(L_{iL})^c}\,\tilde m_L^{ij}\,L_{jL}\,\lambda_1^a \lambda_2^b +h.c.\
- \sum_{a,b,i,j}^{a+b=g_i+g_j} \frac{1}{2}\overline{(N_{iR})^c}\,\tilde M_R^{ij}\,N_{jR}\,\lambda_1^a \lambda_2^b+h.c.\,,
\nonumber
 \label{eq:FN-Lagrangian}
\end{eqnarray}
where all Yukawa and mass matrix entries denote their natural values, while all suppressions come from FN factors. These suppression factors depend on the flavon VEVs $\langle \Theta_i \rangle$ as well as on a high energy scale $\Lambda$, cf.\ Eq.~\eqref{eq:FN_ex_5},
\begin{equation}
 \lambda_i = \frac{\langle \Theta_i \rangle}{\Lambda}.
 \label{eq:FN_sups}
\end{equation}
Although the flavon VEVs can be complex, all but one phase can be rotated away~\cite{Kanemura:2007yy}. Thus it is sufficient to use one real parameter $\lambda$ and one complex parameter $R$,
\begin{equation}
 \lambda=\frac{\langle \Theta_1 \rangle}{\Lambda},\ \ R=\frac{\langle \Theta_1 \rangle}{\langle \Theta_2 \rangle} = 
R_0 e^{i\alpha_0},
 \label{eq:lambda_R}
\end{equation}
where $R_0$ and $\alpha_0$ are real numbers. Note that $\alpha_0$ is responsible for all (Dirac \emph{and} Majorana) $CP$ violation in the model. Using this parametrization, we get
\begin{equation}
 \lambda_1^a \lambda_2^b \equiv \left( \frac{\langle \Theta_1 \rangle}{\Lambda} \right)^a \left( \frac{\langle \Theta_2 \rangle}{\Lambda} \right)^b = 
\lambda^{a+b} R^b.
 \label{eq:L-ratios}
\end{equation}

The goal of any model is to reproduce the known data, in our case the charged lepton masses. A successful choice of FN charges (with $a=0,1$) is given by\cite{Kanemura:2007yy,Asaka:2003fp}
\begin{eqnarray}
 \Theta_{1,2}: && (-1, -1),\nonumber\\
 L_{1,2,3}: && (a+1, a, a),\nonumber\\
 \overline{e_{1,2,3}}: && (3, 2, 0).
 \label{eq:FN-assignments_2}
\end{eqnarray}

The important point is that these assignments do \emph{not} constrain the FN charges of the RH-neutrinos, $(g_1,g_2, g_3)$. Indeed, it is trivial (though a bit cumbersome without the use of a computer algebra package) to show that the light neutrino mass matrix arising from the Lagrangian in Eq.~\eqref{eq:FN-Lagrangian} is independent of $g_i$. The explicit proof for the case of only one flavon field is given in Ref.~\refcite{Merle:2011yv}. As explained in Refs.~\refcite{Kanemura:2007yy,Choi:2001rm}, the reason is simply that the seesaw mass term, due to its Majorana nature, breaks any global symmetry of the RH-neutrino fields. In particular this is true for any assignment of a lepton number, which is typically some kind of global $U(1)$ symmetry~\cite{Davidson:2006bd} and causes the corresponding mass term to be lepton number violating. In the same way, however, the $U(1)_{\rm FN}$ is broken by that term, and hence the RH-neutrino charges $(g_1,g_2, g_3)$ cannot appear in the seesaw formula.

It is exactly this fact that was exploited in Ref.~\refcite{Merle:2011yv}: the charges $(g_1,g_2, g_3)$ can be freely chosen to generate a suitable RH-neutrino mass pattern, and the seesaw formula is always \emph{guaranteed} to work, in spite of the presence of potentially very light RH-neutrinos whose masses are ``divided by'' in the seesaw formula. This is the key point why FN-inspired models are extremely well suited to explain low-scale sterile neutrinos, and in particular keV sterile neutrinos.

In the $\nu$MSM~\cite{Asaka:2005an}, which is more or less the minimal scenario for keV sterile neutrinos, the ``heavy'' sterile neutrinos $N_{2,3}$ need to have masses of a few GeV. Depending on the scenario under consideration, one might also be interested in higher masses. In any case a suitable mass pattern may arise for the second and third generation charges being larger than $g_1$ by at least three units, $g_1 \geq g_{2,3} + 3$.\footnote{Note that this statement may be changed if one attempts to find models including sterile neutrinos at the eV scale. Such examples are discussed in, e.g.,\ Refs.~\refcite{Barry:2011wb,Chen:2011ai}. This is one instance where the specific requirements for models with eV or keV sterile neutrinos can exhibit considerable differences.} Having this requirement in mind, two ``minimal'' scenarios were identified in Ref.~\refcite{Merle:2011yv}: Scenario~A where $(g_1,g_2,g_3)=(3,0,0)$ and Scenario~B where $(g_1,g_2,g_3)=(4,1,0)$. Denoting the natural RH mass scale by $M_0$ and assuming the natural mass matrices to be \emph{democratic}, $\tilde M_R^{ij}=M_0\;\; \forall i,j$, these scenarios immediately lead to the following sterile neutrino mass patterns:
\begin{eqnarray}
 {\rm A}(3,0,0): && M_1 = M_0 \lambda^6\ 2 R_0^2 \sqrt{1+ R_0^4 + 2 R_0^2 \cos (2\alpha_0)}\,, \nonumber\\
                 && M_2 = M_0\,, \nonumber\\
                 && M_3 = M_0 \left(1 + \lambda^6 [ 1 + R_0^2 (3 \cos (2 \alpha_0) + 3 R_0^2 \cos (4\alpha_0) + 
R_0^4 \cos (6 \alpha_0) ] \right)\,, \nonumber\\
 {\rm B}(4,1,0): && M_1 = M_0 \lambda^8\ 2 R_0^4 \sqrt{1 + R_0^8 - 2 R_0^4 \cos (4 \alpha_0)}\,,\nonumber\\
                 && M_2 = M_0 \lambda^2\,, \nonumber\\
                 && M_3 = M_0 \ \left( 1+ R_0^2 \lambda^2 \cos (2\alpha_0) \right)\,.
   \label{eq:RH_masses}
\end{eqnarray}
Indeed, these two scenarios reveal the desired structure: if $\lambda \sim 0.1$, then the requirement $M_1 = \mathcal{O}(1~{\rm keV})$ leads to $M_0 \sim (10^6~{\rm keV}\sim 1~{\rm GeV}, 10^8~{\rm keV}\sim 100~{\rm GeV})$ for scenario (A,B). Even larger values of the FN charges would easily be able to push $M_0$ to even higher values, due to the exponential dependence on the FN charges: already for $g_1=5$, one would achieve $M_0\sim 10$~TeV.

The two scenarios (A,B) are combined with the two possible charge assignments for the LH lepton doublets, Assignment~1 ($a=0$) and Assignment~2 ($a=1$), cf.\ Eq.~\eqref{eq:FN-assignments_2}, which lead to two different possibilities for the charged lepton mass matrix,
\begin{equation}
 M_e^{(1,2)} = v \begin{pmatrix}
 Y_e^{11} B_{2,4} \lambda^{3,4}\hfill \hfill & Y_e^{12} B_2 \lambda^{2,3}\hfill \hfill & Y_e^{13} B_{0,2} R \lambda^{2,3}
\hfill \hfill\\
 Y_e^{21} B_2 \lambda^{2,3}\hfill \hfill & Y_e^{22} B_{0,2} \lambda^{1,2}\hfill \hfill & Y_e^{23} R \lambda^{1,2}\hfill \hfill\\
 Y_e^{31} R \lambda^{1,2} \hfill \hfill & 0,\ Y_e^{32} R \lambda \hfill \hfill  & Y_e^{33} \lambda^{0,1} \hfill \hfill \hfill
 \end{pmatrix}.
 \label{eq:charged_matrices_12}
\end{equation}
Again assuming the natural mass matrices to be democratic, $Y_e^{ij}=Y_e\;\; \forall i,j$, one obtains \emph{e.g.} for Assignment~1 the charged lepton masses
\begin{equation}
 \left\{
 \begin{array}{lcl}
 m_e &=& m_0 \lambda ^3 \; R_0^2\,,\\
 m_\mu &=& m_0 \lambda \; \left( 1+ \lambda^2 \left[ R_0^2 \cos (2 \alpha_0 )+\frac{R_0^4-R_0^2+3}{2} \right] \right)\,,\\
 m_\tau &=& m_0 \; \left( 1 + \frac{3}{2} R_0^2 \lambda^2 \right)\,,
 \end{array}
 \right.
 \label{eq:CL_masses}
\end{equation}
where $m_0 = v Y_e$. Using the measured values for $m_e$, $m_\mu$, and $m_\tau$, one can calculate the mass ratios to obtain $\lambda \simeq 0.06$ and $R_0 \simeq 1.18$ at lowest order in $\lambda$. The phase $\alpha_0$ is not constrained to that order, but choosing $\alpha_0=0.67$ causes the $\mathcal{O}(\lambda^3)$ correction to $m_\mu/m_\tau$ to vanish exactly. As typical for seesaw models\cite{Sato:2000ff}, the optimal choice for $\lambda$ turns out to be a bit smaller than the standard choice of $0.22$~\cite{Datta:2005ci}.

As visible in Eqs.~\eqref{eq:RH_masses} or~\eqref{eq:CL_masses}, the FN mechanism is extremely well suited to predict the desired mass hierarchy. However, the FN symmetry is a $U(1)$ group which has only 1-dimensional representations. Unfortunately, these are not sufficient to predict more eleborate structures of the mass matrix as, \emph{e.g.}, forcing certain entries to be exactly equal. One way out is to depart from the assumption of democracy, as illustrated in Ref.~\refcite{Merle:2011yv}. In practice, the most convenient way is to simply write numerical coefficients of  $\mathcal{O}(1)$ in front of all Yukawa and mass matrix elements. These coefficients are not \emph{predicted} by the FN mechanism, which can be seen as disadvantage. On the other hand, the FN mechanism nevertheless imposes a certain structure on the mass matrices, which makes it actually very easy to find suitable coefficients that lead to consistency with all experimental values.

To give one explicit example, we want to shortly illustrate Model~1AI from Ref.~\refcite{Merle:2011yv}, which is based on Assignment~1, Scenario~A, and a seesaw type~I mechanism. By imposing the measured values of the charged lepton masses, of the neutrino mass square differences, as well as the previously determined values of the parameters $(\lambda, R_0, \alpha_0)$, the numerical coefficients have been determined such that mixing angles within their (at that time) 3$\sigma$ ranges were predicted. In the case of Model~1AI, the mass matrices found are given by:
\begin{eqnarray}
 M_e^{(1)} = v \begin{pmatrix}
 Y_e^{11} B_2 \lambda^3 \hfill \hfill & Y_e^{12} B_2 \lambda^2 \hfill \hfill & Y_e^{13} B_0 R \lambda^2 \hfill \hfill\\
 Y_e^{21} B_2 \lambda^2 \hfill \hfill & Y_e^{22} B_0 \lambda \hfill \hfill & Y_e^{23} R \lambda \hfill \hfill\\
 Y_e^{31} R \lambda \hfill \hfill & 0 \hfill \hfill  & Y_e^{33} \hfill \hfill \hfill
 \end{pmatrix}
 &\to& M_{e0} \left(
\begin{array}{lll}
 0.81 B_2 \lambda^3 & 1.44 B_2 \lambda^2 & 0.29 R \lambda^2 \\
 2.00 B_2 \lambda^2 & 1.13 \lambda  & 2.50 R \lambda  \\
 3.71 R \lambda  & 0 & 0.35
\end{array}
\right), \nonumber\\
 M_R^{\rm (A)} = \begin{pmatrix}
 \tilde M_R^{11} B_6 \lambda^6 & \tilde M_R^{12} B_2 \lambda^3 & \tilde M_R^{13} R B_2 \lambda^3\\
 \bullet & \tilde M_R^{22} B_0 & 0 \hfill \hfill \\
 \bullet & \bullet  & \tilde M_R^{33}\hfill \hfill
 \end{pmatrix}
 &\to& M_0 \left(
 \begin{array}{lll}
 0.38 B_6 \lambda^6 & 0.31 B_2 \lambda^3 & 1.26 B_2 R \lambda^3 \\
 0.31 B_2 \lambda^3 & 4.18 & 0 \\
 1.26 B_2 R \lambda^3 & 0 & 4.81
 \end{array}
 \right), \nonumber\\
 m_D^{\rm (1A)} = v \begin{pmatrix}
 Y_D^{11} B_4 \lambda^4 \hfill \hfill & Y_D^{12} B_2 \lambda^3 \hfill \hfill & Y_D^{13} R B_2 \lambda^3 \hfill \hfill \\
 Y_D^{21} B_0 \lambda \hfill \hfill & Y_D^{22} \hfill \hfill & 0 \hfill \hfill \\
 Y_D^{31} R \lambda \hfill \hfill & 0 \hfill \hfill  & Y_D^{33} \hfill \hfill
 \end{pmatrix}
 &\to& m_{D0} \left(
\begin{array}{lll}
 0.75 B_4 \lambda^4 & 0.15 B_2 \lambda^3 & 1.42 B_2 R \lambda^3 \\
 0.51 \lambda  & 0.13 & 0 \\
 3.32 R \lambda  & 0 & 2.93
\end{array}
\right),\nonumber
\end{eqnarray}
where $B_{2n}=1+R^2+...+R^{2n}$, and $M_{e0}$, $M_0$, and $m_{D0}$ denote the characteristic mass scales of the charged leptons, RH masses, and Dirac masses, respectively. Note that some entries are zero, since they are forbidden by the auxiliary $Z_2$ symmetry.

While Model~1AI still does not predict the absolute light neutrino mass scale $m_\nu$, it is powerful enough to predict their mass ratios, their mass ordering, as well as all leptonic mixing angles and phases. Calculating all these quantities from the mass matrices of the model, one obtains\cite{Merle:2011yv}: $\sin^2 \theta_{12} = 0.28$, $\sin^2 \theta_{13} = 0.018$, $\sin^2 \theta_{23} = 0.54$, $\delta = 4.21$, $\alpha = 0.58$, $\beta = 1.25$, $m_1/m_\nu = 0.0014$, $m_2/m_\nu = 0.19$, and $m_3/m_\nu = 1.03$, which means normal mass ordering.

Note that, although the FN mechanism is not very well suited to predict more complicated leptonic mixing patterns, it nevertheless intrinsically predicts that $\theta_{13} \neq 0$ in the absence of extremely unlikely cancellations, in contrast to some discrete flavor symmetries with a tendency to predict vanishing $\theta_{13}$. This is particularly remarkable since the predictions from Ref.~\refcite{Merle:2011yv} appeared \emph{before} the actual experimental determination of $\theta_{13}$. Nevertheless, even today, Model~1AI is valid within the 3$\sigma$ ranges of the leptonic mixing parameters.

Let us end this section by mentioning an interesting feature of the FN mechanism, which has been discussed in Ref.~\refcite{Merle:2011yv}. While the $U(1)_{\rm FN}$ charge assignments seem quite arbitrary at first sight, there nevertheless exist many settings that do \emph{not} work in connection with FN, in particular in the context of keV sterile neutrinos. Out of all the requirements mentioned in Ref.~\refcite{Merle:2011yv}, we will shortly comment on three points to illustrate how restrictive the FN mechanism actually is:
\begin{itemize}

\item \emph{No LR symmetry}:\\
Left-right symmetry~\cite{Senjanovic:1975rk,Mohapatra:1977mj,Mohapatra:1979ia,Mohapatra:1980yp} is by itself very restrictive: the RH charged leptons and neutrinos are members of the same doublet, which implies $k_i = g_i$. Hence, the symmetry forces the LH and RH fermion doublets of each generation to have equal FN charges, $f_i = g_i$. With such a strong restriction, no semi-realistic pattern like in Eq.~\eqref{eq:FN-assignments_2} can be obtained. This problem is particularly interesting since one mechanism to produce a suitable abundance of keV sterile neutrinos, thermal overproduction with subsequent entropy dilution, was mainly studied for LR-symmetry~\cite{Bezrukov:2009th,Nemevsek:2012cd}.

\item \emph{Difficulties with $SO(10)$}:\\
The FN mechanism is typically applied in the context of \emph{Grand Unified Theories} (GUTs). In the case of $SU(5)$, the RH neutrinos are total singlets, which leaves their FN charges perfectly unconstrained, just as required to achieve the desired RH neutrino mass patterns, cf.\ Eq.~\eqref{eq:RH_masses}. However, in theories based on $SO(10)$, the RH neutrinos are members of a $\mathbf{16}_i$ representation, for each generation $i$~\cite{Ross:1985ai}. Hence, their values would be strongly constrained, which destroys the freedom needed to generate a strong hierarchy among the sterile neutrino masses.

\item \emph{No need to run}:\\
Sometimes, non-fitting FN-based models are claimed to have a better match with the data when the couplings are evolved to low energies by renormalization group running, especially in cases where the models under study do not lead to a good fit with data~\cite{Kamikado:2008jx}. While this is a natural thought, since the predictions obtained by the FN mechanism should only be strictly true at a high energy scale, the running of the couplings is nevertheless fully negligible for low scale seesaw mechanisms as required by typical scenarios for keV sterile neutrinos. This conclusion could only be altered by having strong degeneracies between two or more light neutrino masses, at a relative level of roughly $10^{-5}$. However, since the FN mechanism generically leads to hierarchies and not to degeneracies, such effects can hardly appear in the models under consideration.

\end{itemize}

As we have seen, the FN mechanism is well suited to obtain scale suppressions and to generate hierarchies between masses. We have also seen that, although FN is actually more predictive than one might naively think, it is not powerful enough to predict complicated mixing patterns by itself, without varying coefficients of $\mathcal{O}(1)$. Furthermore, the generic FN-problems, and in particular the lack of a UV-completion, persist. A partial way out is to combine the FN mechanism with an additional discrete flavor symmetry, which is the next point to be discussed.

\subsubsection{\label{sec:keV_FN_mixed}Mixed FN models}	

We have now seen how one can build a model based exclusively on the FN mechanism, apart from an auxiliary $Z_2$ symmetry that is ``only'' important if manifest $CP$ violation is desired. The alternative type of models on the market uses the FN mechanism only as one among several ingredients, typically in combination with a non-Abelian (\emph{i.e.}, non-commuting) discrete flavor symmetry. This approach has the great advantage of being able to exploit the powerful FN suppressions, while at the same time making use of the structure of the discrete group in order to predict more sophisticated mixing patterns in the lepton sector. However, as always in model building one pays a price, which is in the case at hand the loss of minimality.

The models on the market make use of an $A_4$ symmetry, along with an additional small auxiliary symmetry $Z_3$,\footnote{$Z_3 = \left\{e, a, a^2 \right\}$ is the cyclic group of three elements. Hereby, $e$ denotes the neutral element and $a$ is a non-trivial element with $a^3 = e$.} in combination with the FN mechanism~\cite{Barry:2011wb,Barry:2011fp}. The model we are going to discuss in some detail here is the \emph{$A_4$ seesaw model with one keV sterile neutrino} (in the version with the other two sterile neutrinos being considerably heavier, with masses of at least a few GeV) presented in Ref.~\refcite{Barry:2011fp}. This model is maybe the prime example in the literature to illustrate how to combine the FN mechanism with a non-Abelian symmetry in a smart way, in order to obtain a model with a keV sterile neutrino.

When looking into the literature, typically the first piece of information given about a model is the particle content as well as the charge assignments, although obtaining this knowledge typically requires a lot of work and failed attempts. For the model under consideration, this information is presented in Tab.~\ref{tab:A4BRZ}. We first of all see that, besides the necessary lepton doublets and singlets, two Higgs fields $H_{u,d}$ are needed in the model, as well as seven flavons $\varphi$, $\varphi'$, $\varphi''$, $\xi$, $\xi'$, $\xi''$, and $\Theta$. In addition, Tab.~\ref{tab:A4BRZ} gives us the information of how the fields transform under the different symmetries, either by directly indicating the representation [for $SU(2)_L$ and for $A_4$] or by indicating the charges [for $Z_3$ and $U(1)_{\rm FN}$]. Note that the field $\Theta$ is a total singlet under all symmetries but $U(1)_{\rm FN}$, so $\Theta$ is nothing else than the FN flavon. The other flavons, however, are triplet or singlet flavons, under $A_4$, and they induce most of the structure. Finally, the sole purpose of the $Z_3$ assignment is to switch off certain terms that would not be desired, for example because they could spoil the mixing pattern.

Next, we write down the leading order Lagrangian, where only terms that are singlets under all symmetries are allowed. Glancing at Tab.~\ref{tab:A4BRZ}, making use of the $A_4$ group theory (cf.\ Sec.~\ref{sec:neutrino_modeling_flavor}), and realizing that $\omega^3 = 1$, the most general leptonic Yukawa coupling Lagrangian is given by:
\begin{eqnarray}
 \mathcal{L}_Y &=&  -\frac{y_e}{\Lambda}\lambda^3 \overline{e_R} \left(\varphi H_d L \right)_{\mathbf{1}} - \frac{y_\mu}{\Lambda} \lambda \overline{\mu_R} \left(\varphi H_d L \right)_{\mathbf{1'}} - \frac{y_\tau}{\Lambda} \overline{\tau_R} \left( \varphi H_d L \right)_{\mathbf{1''}} \label{eq:BRZ_lag} \\
 && -  \frac{y_1}{\Lambda}\lambda^{g_1} \overline{N_{1R}} (\varphi H_u L)_{\mathbf{1}} - \frac{y_2}{\Lambda}\lambda^{g_2} \overline{N_{2R}} (\varphi' H_u L)_{\mathbf{1''}} - \frac{y_3}{\Lambda}\lambda^{g_3} \overline{N_{3R}} (\varphi'' H_u L)_{\mathbf{1}} \nonumber \\
 &&- \frac{1}{2} \left[w_1^* \lambda^{2 g_1}\xi^* \overline{(N_{1R})^c} N_{1R} - w_2^* \lambda^{2 g_2} {\xi'}^* \overline{(N_{2R})^c} N_{2R} - w_3^* \lambda^{2 g_3} {\xi''}^* \overline{(N_{3R})^c} N_{3R}
 \right] + h.c., \nonumber
\end{eqnarray}
where the notation $(...)_{\rm \bf irrep}$ refers to the corresponding combination under $A_4$. Note that the FN mechanism has already been applied in Eq.~\eqref{eq:BRZ_lag}, which means that the invariance under $U(1)_{\rm FN}$ is only visible once the different powers of the small parameters $\lambda \equiv \langle
\Theta\rangle/\Lambda$ are taken into account. Note further that the notation has been slightly changed in order to be consistent with the rest of this review, which is why the Lagrangian presented here looks a bit different from the one in Ref.~\refcite{Barry:2011fp}.\footnote{Note also the typo in Ref.~\refcite{Barry:2011fp} in the third charged lepton Yukawa coupling.}

\begin{table}[h]
\tbl{Particle content and representation/charge assignments (in our notation) of the $SU(2)_L \times A_4 \times Z_3 \times U(1)_{\rm FN}$ model with one keV sterile neutrino~\cite{Barry:2011fp}. $SU(2)_L$ is the usual SM gauge symmetry, and $Z_3$ is an additional auxiliary symmetry in order for $CP$ violation to survive, similar to the $Z_2$ used in Sec.~\ref{sec:keV_FN_pure}. For $SU(2)_L$ and for $A_4$ the representations are given, while for $Z_3$ and $U(1)_{\rm FN}$ the charge is indicated. Note that $\omega = e^{2\pi i /3}$.}
{\begin{tabular}{@{}lccccccccccccccc@{}} \toprule
Field & $L_{1,2,3}$ & $\overline{e_R}$ & $\overline{\mu_R}$ & $\overline{\tau_R}$ & $\overline{N_1}$ & $\overline{N_2}$ & $\overline{N_3}$ & $H_{u,d}$ & $\varphi$ & $\varphi'$ & $\varphi''$ & $\xi$ & $\xi'$ & $\xi''$ & $\Theta$ \\
\colrule
$SU(2)_L$ & $\mathbf{2}$ & $\mathbf{1}$ & $\mathbf{1}$ & $\mathbf{1}$ & $\mathbf{1}$ & $\mathbf{1}$ & $\mathbf{1}$ & $\mathbf{2}$ & $\mathbf{1}$ & $\mathbf{1}$ & $\mathbf{1}$ & $\mathbf{1}$ & $\mathbf{1}$ & $\mathbf{1}$ & $\mathbf{1}$ \\
$A_4$ & $\mathbf{3}$ & $\mathbf{1}$ & $\mathbf{1''}$ & $\mathbf{1'}$ & $\mathbf{1}$ & $\mathbf{1'}$ & $\mathbf{1}$ & $\mathbf{1}$ & $\mathbf{3}$ & $\mathbf{3}$ & $\mathbf{3}$ & $\mathbf{1}$ & $\mathbf{1'}$ & $\mathbf{1}$ & $\mathbf{1}$ \\
$Z_3$ & $\omega$ & $\omega^2$ & $\omega^2$ & $\omega^2$ & $1$ & $1$ & $\omega$ & $\omega^2$ & $\omega^2$ & $\omega$ & $1$ & $1$ &  $\omega^2$ & $\omega$ & $1$  \\
$U(1)_{\rm FN}$ & $0$ & $3$ & $1$ & $0$ & $g_1$ & $g_2$ & $g_3$ & $0$ & $0$ & $0$ & $0$ & $0$ & $0$ & $0$ & $-1$ \\
\botrule
\end{tabular}
\label{tab:A4BRZ}}
\end{table}

Again, one has to choose a certain VEV alignment. In the case at hand, this alignment can, \emph{e.g.}, be achieved in a supersymmetric context by adding further auxiliary fields~\cite{Altarelli:2005yx}, but also other origins might be plausible. As we had mentioned in Sec.~\ref{sec:neutrino_modeling_flavor}, we will not go very far into this topic in this review, as achieving VEV alignments is nearly a discipline by itself, but the principles behind it are not very different from the ones in the easy example discussed in \ref{sec:vacuum}. In some sense, the alignment belongs to the ``dirty'' part of model building, since it typically involves an essentially untestable sector and many fields that are typically shifted towards higher energies. However, an alternative view would be that there are indeed many possibilities to achieve a certain VEV alignment, and the mere number of working possibilities is indeed a bonus of flavor models.

Ending this discussion to enter the actual physics, the alignment chosen in Ref.~\refcite{Barry:2011fp} for the first triplet flavon field is $\langle \varphi \rangle =
(v_\varphi,0,0)$, which leads to a charged lepton mass matrix given by
\begin{equation}
 M_e = \frac{v_d v_\varphi}{\Lambda}
 \begin{pmatrix} y_e \lambda^{3} &
 0 & 0 \\
 0 & y_\mu\lambda & 0\\
 0 & 0 & y_\tau
 \end{pmatrix},
 \label{eq:BRZ_1}
\end{equation}
where $v_d = \langle H_d \rangle$ is one of the two Higgs VEVs. Note the FN suppression, which is clearly visible in Eq.~\eqref{eq:BRZ_1}: although different assignments have been used, the actual pattern of the charged lepton masses clearly resembles the one obtained from pure FN models, cf.\ Eq.~\eqref{eq:CL_masses}.

Up to now, the FN charges $g_{1,2,3}$ of the RH neutrinos have not yet been fixed. With a VEV $u= \langle \xi \rangle$, the mass of the lightest sterile neutrino is given by 
\begin{equation}
 M_1 = w_1 u \lambda^{2 g_1},
 \label{eq:M1_BRZ}
\end{equation} 
where $w_1$ can always be chosen to be real and positive by absorbing any phase into the fields. Taking $g_1$ large enough rapidly decreases this mass down to the keV scale, while the parameter $w_1$ is still left for finer adjustments. For example, in Ref.~\refcite{Barry:2011fp} this charge is suggested to be $g_1=9$, which leads to $M_1 \approx 1$~keV for $\lambda \approx 0.1$ and $u \simeq 10^{12}$~GeV. However, one has to remark that such a large $g_1$ can also be problematic, as the FN charge actually \emph{enhances} the active-sterile mixing, $\theta_{e1} \simeq \frac{y_1 v_\varphi v_u}{w_1 u \Lambda} \lambda^{-g_1}$. Hence, one has to be extremely careful that a too large value of $g_1$ does not spoil the validity of the model.

Furthermore, Ref.~\refcite{Barry:2011fp} states that this lightest sterile neutrino decouples from the seesaw formula due to the strong X-ray bound, $\theta_1^2 \lesssim 10^{-8}$ for $M_1 \approx 10$~keV, cf.\ Sec.~\ref{sec:keV_general_X}. While this is certainly true in some approximation, one could also argue that one does not have to care about decoupling the $N_1$, since the FN charges $g_{1,2,3}$ in any case cancel in the seesaw formula, cf.\ Sec.~\ref{sec:keV_FN_pure} and Ref.~\refcite{Merle:2011yv}. In particular, the quantities $w_1$ and $u$ from Eq.~\eqref{eq:M1_BRZ} will actually enter the seesaw formula and cannot necessarily be neglected. Alternatively, one could argue that according to the \emph{seesaw practicality theorem}, cf.\ Sec.~\ref{sec:keV_general_seesaw}, we could in any case make use of the seesaw formula.

In addition, in the approximation of a fully decoupled $N_1$, one of the light neutrinos must actually be massless at tree-level: the \emph{seesaw fair play rule}~\cite{Xing:2007uq}, cf.\ Sec.~\ref{sec:keV_general_seesaw}, states that two massive right-handed neutrinos can lead to at most two massive light neutrinos which might contradict future experiments. Hence, even though the decoupling of $N_1$ is true to some extent, it is not always a good approximation.

Let us follow the path of Ref.~\refcite{Barry:2011fp} and see where it is leading us. The remaining task is to determine the mass and mixing pattern. As often the case for FN-inspired models, normal or inverted neutrino mass ordering are both possible. Ref.~\refcite{Barry:2011fp} suggests to take the alignment $\langle \varphi' \rangle = (v'_\varphi,v'_\varphi,v'_\varphi)$ for both cases, while $\langle \varphi'' \rangle = (0,v''_\varphi,-v''_\varphi)$ and $\langle \varphi'' \rangle = (2 v''_\varphi,-v''_\varphi,-v''_\varphi)$ lead to normal and inverted ordering, respectively. With these choices, one can straightforwardly write down the Dirac mass matrices,
\begin{equation}
 m^{({\rm NO})}_D = \frac{v_u}{\Lambda} \begin{pmatrix}
 y_2 v'_\varphi \lambda^{g_2} & 0\\
 y_2 v'_\varphi \lambda^{g_2} & -y_3 v''_\varphi \lambda^{g_3}\\
 y_2 v'_\varphi \lambda^{g_2} &  y_3 v''_\varphi \lambda^{g_3}
 \end{pmatrix}, \ \ \ m^{({\rm IO})}_D = \frac{v_u}{\Lambda} \begin{pmatrix}
 y_2 v'_\varphi \lambda^{g_2} & 2 y_3 v''_\varphi \lambda^{g_3}\\
 y_2 v'_\varphi \lambda^{g_2} & -y_3 v''_\varphi \lambda^{g_3}\\
 y_2 v'_\varphi \lambda^{g_2} &  -y_3 v''_\varphi \lambda^{g_3}
 \end{pmatrix},
 \label{eq:Dirac_BRZ}
\end{equation} 
where $v_u = \langle H_u \rangle$, and the RH Majorana mass matrix,
\begin{equation}
 M_R = \begin{pmatrix}
 w_2 u' \lambda^{2 g_2} & 0\\
 0 & w_3 u'' \lambda^{2 g_3}
 \end{pmatrix}.
 \label{eq:Majorana_BRZ}
\end{equation} 
Using Eqs.~\eqref{eq:Dirac_BRZ} and~\eqref{eq:Majorana_BRZ}, one obtains the full $5\times 5$ neutrino mass matrices for both orderings. These matrices can be fully diagonalized to yield the eigenvalues
\begin{align}
 m_1^{\rm NO} &= 0 , & m_1^{\rm IO} &= m_1^{(0)}\left(1 -6\epsilon_2^2\right), \nonumber\\
 m_2^{\rm NO} &= m_2^{(0)}\left(1 -3\epsilon_1^2\right) , & m_2^{\rm IO} & = m_2^{(0)}\left(1 -3\epsilon_1^2\right), \nonumber\\
 m_3^{\rm NO} &= m_3^{(0)}\left(1-2\epsilon_2^2\right) , & m_3^{\rm IO} &= 0, \nonumber\\
 m_4^{\rm NO} &= w_2 u' \lambda^{2 F_2} - m_2^{(0)}\left(1-3\epsilon_1^2\right) , & m_4^{\rm IO} &= w_2 u' \lambda ^{2 F_2} - m_2^{(0)}\left(1-3\epsilon_1^2\right), \nonumber\\
 m_5^{\rm NO} &= w_3 u'' \lambda^{2 F_3} - m_3^{(0)}\left(1 - 2\epsilon_2^2\right) , & m_5^{\rm IO} &= w_3 u'' \lambda ^{2 F_3} - m_1^{(0)}\left(1-6\epsilon_2^2\right),
 \label{eq:BRZ_eignvals}
\end{align}
for normal and inverted ordering, respectively, where we have used the abbreviations
\begin{eqnarray}
 && m_1^{(0)} \equiv -\frac{6 y_3^2 v''^2 v_u^2 }{w_3 u'' \Lambda^2}\ , \ \ m_2^{(0)} \equiv -\frac{3  y_2^2 v'^2 v_u^2 }{w_2 u' \Lambda^2}\ , \ \ m_3^{(0)} \equiv -\frac{2 y_3^2 v''^2 v_u^2  }{w_3 u'' \Lambda ^2}\ , \nonumber\\
 && \epsilon_1 \equiv \frac{y_2 v_\varphi' v_u}{w_2 u' \Lambda}\lambda^{-g_2}\ , \ \ \epsilon_2 \equiv \frac{y_3 v_\varphi'' v_u}{w_3 u'' \Lambda} \lambda^{-g_3}. 
 \label{eq:abbr_BRZ}
\end{eqnarray}
Again, one has to be careful in the choice of the FN charges $g_{2,3}$ in order not to spoil the smallness of $\epsilon_{1,2}$. Still, the above patterns are somewhat compatible with the experimental data, and they can be used to calculate predictions for neutrino observables.

The corresponding diagonalization matrices $U_\nu$ are given by
\begin{equation}
 U_\nu^{({\rm NO, IO})} \simeq
 \begin{pmatrix}
 \frac{2}{\sqrt{6}} & \frac{1}{\sqrt{3}} & 0 & 0 & 0 \\
 -\frac{1}{\sqrt{6}} & \frac{1}{\sqrt{3}} & -\frac{1}{\sqrt{2}} & 0 & 0 \\
 -\frac{1}{\sqrt{6}} & \frac{1}{\sqrt{3}} & \frac{1}{\sqrt{2}} & 0 & 0 \\ 
 0 & 0 & 0 & 1 & 0 \\ 
 0 & 0 & 0 & 0 & 1 
 \end{pmatrix} + 
 \begin{pmatrix} 
 0 & 0 & 0 & \epsilon_1 & 0, 2 \epsilon_2 \\ 
 0 & 0 & 0 & \epsilon_1 & -\epsilon_2 \\ 
 0 & 0 & 0 & \epsilon_1 & \epsilon_2, -\epsilon_2 \\ 
 0 & -\sqrt{3} \epsilon_1 & 0 & 0 & 0\\ 
 0, -\sqrt{6}\epsilon_2 & 0 & -\sqrt{2}\epsilon_2, 0 & 0 & 0
 \end{pmatrix} + \mathcal{O}(\epsilon_i^2).
 \label{eq:mix_BRZ}
\end{equation}
This matrix already yields the physical mixing parameters, since the charged lepton mass matrix is already diagonal, cf.\ Eq.~\eqref{eq:BRZ_1}. Note that the upper left $3\times 3$-block of the leading contribution is \emph{tri-bimaximal}~\cite{Harrison:2002er},
\begin{equation}
 U_{\rm TBM} =
 \begin{pmatrix}
 \frac{2}{\sqrt{6}} & \frac{1}{\sqrt{3}} & 0 \\
 -\frac{1}{\sqrt{6}} & \frac{1}{\sqrt{3}} & -\frac{1}{\sqrt{2}} \\
 -\frac{1}{\sqrt{6}} & \frac{1}{\sqrt{3}} & \frac{1}{\sqrt{2}} 
 \end{pmatrix}.
 \label{eq:TBM}
\end{equation}
This yields $\theta_{13}=0$ in particular, which is now known to be excluded, cf.\ Eqs.~\eqref{eq:angles_NH} and~\eqref{eq:angles_IH}. Although the leading order matrix $U_\nu^{({\rm NO, IO})}$ is corrected by terms of $\mathcal{O}(\epsilon_i)$ and higher, the first (and actually also the second) order contributions do not do the job of generating a large enough $\theta_{13}$. Even higher orders are likely to be too small to yield a reasonably sized $\theta_{13}$. However, it is argued in Ref.~\refcite{Barry:2011fp} that this problem can be overcome by including higher-order terms in the Lagrangian, Eq.~\eqref{eq:BRZ_lag}. Since such corrections are unavoidable, this is a good way of solving the problem, and the reference further shows that relatively large values of $\theta_{13}$ are in fact possible.

Summing up, the model presented in Ref.~\refcite{Barry:2011fp} comprises a fully working setting, as long as the parameters are carefully chosen. Nevertheless, although FN-inspired models are seemingly perfect to explain keV neutrinos, choosing a too large FN charge can lead to problems with the strong observational bound on the active-sterile mixing.

\subsection{\label{sec:keV_flav}Models based on flavor symmetries}	

Apart from the FN mechanism, which is a simple kind of flavor symmetry, one can also use other symmetries. In the literature, one currently finds models for keV sterile neutrinos based on two relatively simple symmetries, $L_e - L_\mu - L_\tau$ and $Q_6$, although other possibilities are not excluded. The important property which the existing models based on flavor symmetries have in common is that the keV sterile neutrino mass is completely forbidden at leading order. Hence, the ``natural'' value of mass of the DM neutrino $N_1$ is exactly zero, but higher order corrections lift it to a non-zero value, which clearly follows a bottom-up scheme, cf.\ left panel of Fig.~\ref{fig:keV-schemes}. Furthermore, if the other sterile neutrinos $N_{2,3}$ obtain a mass already at leading order, these types of models simultaneously yield an explanation for $M_1 \ll M_{2,3}$. Of course, this does not necessarily have to lead to the absolute scale of keV, but the models indeed do explain a certain structure (or hierarchy) in the RH neutrino sector.

\subsubsection{\label{sec:keV_flav_Le}$L_e - L_\mu - L_\tau$ symmetry}	

The most straightforward symmetry to use is, similar to FN, a particular $U(1)$ symmetry which essentially plays the role of a generalized lepton number. The symmetry under consideration is called $L_e - L_\mu - L_\tau$, for reasons that will become clear soon. For simplicity, we will abbreviate $\mathcal{F} \equiv L_e - L_\mu - L_\tau$. Based on considerations in Ref.~\refcite{Petcov:1982ya}, this $\mathcal{F}$ symmetry was proposed in Ref.~\refcite{Lavoura:2000ci} (two RH neutrinos) and in Refs.~\refcite{Barbieri:1998mq,Mohapatra:2001ns} (three RH neutrinos). In general, the symmetry leads to a spectrum $(0,m,m)$ for light neutrinos in the limit of $\mathcal{F}$ being conserved. However, while for Ref.~\refcite{Lavoura:2000ci} this could be argued to originate from the seesaw fair play rule, in particular Ref.~\refcite{Mohapatra:2001ns} showed that the symmetry intrinsically predicts an analogous pattern $(0,M,M)$ in the heavy neutrino sector where no seesaw is at work, if an additional $\mu$-$\tau$ symmetry is present. A first application of this symmetry (and of similar patterns) to the case of keV sterile neutrinos was presented in Ref.~\refcite{Shaposhnikov:2006nn}, and a recent more detailed study was performed in Ref.~\refcite{Lindner:2010wr}.

The charge assignments of this model are illustrated in Tab.~\ref{tab:LeLmLt}. Note that $\mathcal{F}$ is a special type of $U(1)$ symmetry, which obviously has its name $L_e - L_\mu - L_\tau$ from the charge assignments. At first sight, this model looks considerably more economical than, \emph{e.g.}, the one in Sec.~\ref{sec:keV_FN_mixed}. However, the reason for this is that the exact breaking of the symmetry was not specified in Ref.~\refcite{Lindner:2010wr}. This may be seen as incompleteness or inconsistency. On the other hand, the effect on the spectrum discussed in the reference is general enough that it is applicable for more than one type of breaking. Nevertheless, a complete model would need to specify these details.

\begin{table}[h]
\tbl{Particle content and representation/charge assignments (in our notation) of the $SU(2)_2 \times \mathcal{F}$ model~\cite{Lindner:2010wr}. $SU(2)_L$ is the usual SM gauge symmetry. For $SU(2)_L$ the representations are given, while for $\mathcal{F}$ the charge is indicated.}
{\begin{tabular}{@{}lccccccccccc@{}} \toprule
Field & $L_1$ & $L_2$ & $L_3$ & $\overline{e_R}$ & $\overline{\mu_R}$ & $\overline{\tau_R}$ & $\overline{N_1}$ & $\overline{N_2}$ & $\overline{N_3}$ & $H$ & $T$ \\
\colrule
$SU(2)_L$ & $\mathbf{2}$ & $\mathbf{2}$ & $\mathbf{2}$ & $\mathbf{1}$ & $\mathbf{1}$ & $\mathbf{1}$ & $\mathbf{1}$ & $\mathbf{1}$ & $\mathbf{1}$ & $\mathbf{2}$ & $\mathbf{3}$ \\
$\mathcal{F}$ & $1$ & $-1$ & $-1$ & $-1$ & $1$ & $1$ & $-1$ & $1$ & $1$ & $0$ & $0$ \\
\botrule
\end{tabular}
\label{tab:LeLmLt}}
\end{table}

With these assignments, one can immediately write down all allowed terms in the Lagrangian, in the most general setting of seesaw type~II, cf.\ Eq.~\eqref{eq:full_mass}. However, we will for illustration focus on the RH terms first, which are given by
\begin{eqnarray}
 \mathcal{L}_M &=& -M^{12}_R \overline{(N_{1 R})^c} N_{2 R} - M^{13}_R \overline{(N_{1 R})^c} \, N_{3 R} + h.c. \nonumber\\
 && = -\frac{1}{2} \overline{(N_R)^c}
 \begin{pmatrix}
 0 & M^{12}_R & M^{13}_R\\
 M^{12}_R & 0 & 0\\
 M^{13}_R & 0 & 0
 \end{pmatrix} N_R + h.c. \equiv -\frac{1}{2} \overline{(N_R)^c} M_R N_R + h.c.
 \label{eq:Le_1}
\end{eqnarray}
The key point is that this matrix must have (at least) one zero eigenvalue, as one can easily see by computing the determinant:
\begin{equation}
 {\rm det}\ M_R = 0\cdot 0 \cdot 0 + M^{12}_R\cdot 0 \cdot M^{13}_R + M^{13}_R \cdot M^{12}_R \cdot 0 - M^{13}_R \cdot 0 \cdot M^{13}_R - 0\cdot 0 \cdot 0 - 0\cdot M^{12}_R \cdot M^{12}_R = 0.
 \label{eq:Le_2}
\end{equation}
Indeed, the eigenvalues of the matrix $M_R$ turn out to be $(0, + M, - M)$ so that one sterile neutrino is exactly massless, while the other two are degenerate with masses $M\equiv \sqrt{(M^{12}_R)^2 + (M^{13}_R)^2}$.

Applying the same logic to the full $6 \times 6$ neutrino mass matrix, one obtains
\begin{equation}
 \mathcal{L}_\nu=-\frac{1}{2} 
 \overline{\Psi^c} M_\nu \Psi^c +h.c.,\ \ \ {\rm where}\ \ \ M_\nu =
 \begin{pmatrix}
 \begin{array}{c|c}
 \begin{matrix}
 0 & m^{e \mu}_L & m^{e \tau}_L \\
 m^{e \mu}_L & 0 & 0 \\
 m^{e \tau}_L & 0 & 0 
 \end{matrix}
 &
 \begin{matrix}
 m^{e 1}_D & 0 & 0 \\
 0 & m^{\mu 2}_D & m^{\mu 3}_D \\
 0 & m^{\tau 2}_D & m^{\tau 3}_D
 \end{matrix}\\\hline
 \begin{matrix}
 m^{e 1}_D & 0 & 0 \\
 0 & m^{\mu 2}_D & m^{\tau 2}_D \\
 0 & m^{\mu 3}_D & m^{\tau 3}_D
 \end{matrix}
 &
 \begin{matrix}
 0 & M^{12}_R & M^{13}_R \\ 
 M^{12}_R & 0 & 0 \\
 M^{13}_R & 0 & 0
 \end{matrix}
 \end{array}
 \end{pmatrix},
 \label{eq:Le_3}
\end{equation}
and $\Psi = (\nu_{1L}, \nu_{2L}, \nu_{3L}, N_{1R}^c, N_{2R}^c, N_{3R}^c)$. In the seesaw type~II limit, $m^{\alpha \beta}_L \ll m^{\alpha i}_D \ll M^{i j}_R$, this matrix can be analytically diagonalized under the assumption $m^{\alpha i}_D \sim m_D$. Even for this case one obtains two zero eigenvalues, one in the LH and one in the RH sector, while the other two pairs of eigenvalues are given by $\lambda_\pm = \pm \sqrt{\frac{2(M^{12}_R-M^{13}_R)^2}{(M^{12}_R)^2+(M^{13}_R)^2}} m_D + \mathcal{O}(m^3_D)$ and $\Lambda_\pm=\pm \sqrt{(M^{12}_R)^2+(M^{13}_R)^2} + \mathcal{O}(m^2_D)$ for the light and heavy neutrinos, respectively. We can see that we furthermore need $M^{12}_R \simeq M^{13}_R$ to a good precision for the ``small'' eigenvalues $\lambda_\pm$ to be small enough. This, together with the assumption on the entries of $m_D$, is in effect nothing else than introducing an additional $\mu$-$\tau$ symmetry~\cite{Harrison:2002et}, in accordance with Ref.~\refcite{Mohapatra:2001ns}.

While the eigenvalue spectrum looks promising, in the sense that a strong hierarchy in the heavy sector is generated, we still need to justify a non-zero sterile neutrino mass $M_3$.\footnote{In this model, it is actually the third generation neutrino which is massless. Hence, for the light neutrinos an inverted hierarchy will be predicted, while in the heavy sector it is $N_3$ which plays the role of the keV sterile neutrino.} At the same time, the diagonalization of the light neutrino mass matrix is achieved with the so-called \emph{bimaximal} mixing matrix,
\begin{equation}
 U_{\rm BM}=
 \begin{pmatrix}
 \frac{1}{\sqrt{2}} & \frac{1}{\sqrt{2}} & 0\\
 -\frac{1}{2} & \frac{1}{2} & \frac{1}{\sqrt{2}} \\
 \frac{1}{2} & -\frac{1}{2} & \frac{1}{\sqrt{2}}
 \end{pmatrix},
 \label{eq:Le_4}
\end{equation}
which predicts the non-acceptable values $\theta_{12}^\nu = \frac{\pi}{2}$ and $\theta_{13}^\nu = 0$ in case they are not modified by corrections from the charged lepton sector. This is a general feature of $\mathcal{F}$ symmetry~\cite{Petcov:2004rk}. Both these problems can be cured by \emph{breaking} the symmetry. The exact mechanism of this breaking has not been discussed in Ref.~\refcite{Lindner:2010wr}, but instead several \emph{soft breaking terms} have been assumed, \emph{i.e.}, terms which suitably break the symmetry but do not introduce dangerous quadratic divergences (\emph{i.e.}, the effects on the high energy sector of the theory are assumed to be so weak that no new problems with divergences of the Higgs or other scalar masses arise, similar to the \emph{hierarchy problem} in the SM~\cite{Djouadi:2005gi}). However, the point is that no matter which breaking pattern is assumed, the resulting effect is very generic.

Before we show how to cure the mass and mixing pattern by this breaking, let us note that by the Goldstone theorem~\cite{Goldstone:1961eq,Goldstone:1962es} breaking of a global and continuous symmetry would lead to a massless Goldstone boson,\footnote{In the case at hand, the boson would actually be a so-called \emph{Majoron}~\cite{Chikashige:1980ui,Gelmini:1980re,Georgi:1981pg}.} which is phenomenologically problematic. Possible ways out would be to gauge the symmetry, which would generate a new massive gauge boson, or to use the discretized version of the symmetry instead: since only the $U(1)$ charges $\pm 1$ are used in Tab.~\ref{tab:LeLmLt}, one could instead use, \emph{e.g.}, a cyclic $Z_4 = \left\{ \pm 1, \pm i \right\}$ symmetry. The reader is invited to verify that the change of assignments from $(1, -1, -1)_{U(1)}$ to $(i, -i, -i)_{Z_4}$ for the first, second, and third generation leptons exactly reproduces the mass matrix given in Eq.~\eqref{eq:Le_3}.

We will now turn to a solution of the problem of massless neutrinos. In Ref.~\refcite{Lindner:2010wr} it was suggested to introduce soft breaking terms $s^{e e, \mu \mu, \tau \tau}_L \ll m^{\alpha \beta}_L$ and $S^{e e, \mu \mu, \tau \tau}_R \ll M^{i j}_R$ on the diagonal part of the matrix,
\begin{equation}
 M_\nu \to M_\nu + {\rm diag} (s^{e e}_L, s^{\mu \mu}_L, s^{\tau \tau}_L, S^{e e}_R, S^{\mu \mu}_R, S^{\tau \tau}_R).
 \label{eq:Le_5}
\end{equation}
These new terms must be smaller than the ``old'' ones, since they arise from symmetry breaking. In fact, even when viewed as explicit breaking terms, the whole consideration of a symmetry makes only sense in this limit.\footnote{This is very similar to the isospin symmetry between protons and neutrons in nuclear physics.} The effect of these breaking terms can be studied analytically in the limit $s^{\alpha \alpha}_L \simeq s$, $S^{ii}_R \simeq S$, and $M^{12}_R\simeq M^{13}_R \sim M_R$, which allows to compute the eigenvalues of the light neutrino mass matrix $m_\nu = m_L - m_D M_R^{-1} m_D^T$ to be $\{ \lambda^\prime_+, \lambda^\prime_-, \lambda_s, \Lambda^\prime_+, \Lambda^\prime_-, \Lambda_s \}$, where $\lambda^\prime_\pm= s\pm \sqrt{2} \left[ m_L - \frac{m_D^2}{M_R} \right] + \frac{5 m_D^2 S}{4 M_R^2} + \mathcal{O}\left(\frac{S^2}{M^3_R}\right)$, $\lambda_s=s$, $\Lambda^\prime_\pm=S\pm\sqrt{2} M_R$, and $\Lambda_s=S$. Indeed, $S \ll M_R$ justifies taking $S = \mathcal{O} ({\rm keV})$, while $M_R$ can be of $\mathcal{O} ({\rm GeV})$, or larger. In fact, using the experimental data on the mass square differences, one can use these results to fully predict the light neutrino mass spectrum to be $m_1=s+b$, $m_2=s-b$, and $m_3=s$, where $b=0.0489$~eV and $s=-3.9 \times 10^{-4}$~eV, which yields $|m_1|=0.0486$~eV, $|m_2|=0.0494$~eV, and $|m_3|=0.0004$~eV. This, in particular, confirms that inverted ordering is predicted.

It remains to show that also the neutrino mixings turn out to be okay. Sticking to our simple example, it has been shown in Ref.~\refcite{Lindner:2010wr} that indeed the diagonalization matrix for light neutrinos can be modified compared to Eq.~\eqref{eq:Le_4}. First of all, the light neutrino mass matrix is given by
\begin{equation}
 m_\nu = m_L - m_D M_R^{-1} m_D^T =
 \begin{pmatrix}
 s+\frac{m_D^2 S}{2 M_R^2} & m_L-\frac{m_D^2}{M_R}-\frac{m_D^2 S^2}{2 M_R^3} & m_L-\frac{m_D^2}{M_R}-\frac{m_D^2 S^2}{2 M_R^3} \\
 m_L-\frac{m_D^2}{M_R}-\frac{m_D^2 S^2}{2 M_R^3} & s+\frac{m_D^2 S}{M_R^2} & \frac{m_D^2 S}{M_R^2} \\
 m_L-\frac{m_D^2}{M_R}-\frac{m_D^2 S^2}{2 M_R^3} & \frac{m_D^2 S}{M_R^2} & s+\frac{m_D^2 S}{M_R^2}
 \end{pmatrix}.
 \label{eq:Le_6}
\end{equation}
This modifies the neutrino mixing matrix from bimaximal to
\begin{equation}
 U_{\rm BM} \to U_\nu =
 \begin{pmatrix}
 \frac{1}{\sqrt{2}}-\epsilon  & \frac{1}{\sqrt{2}} +\epsilon & 0 \\
 -\frac{1}{2} - \frac{\epsilon }{\sqrt{2}} & \frac{1}{2}-\frac{\epsilon }{\sqrt{2}} & \frac{1}{\sqrt{2}} \\
 \frac{1}{2} + \frac{\epsilon }{\sqrt{2}} & -\frac{1}{2}+\frac{\epsilon }{\sqrt{2}} & \frac{1}{\sqrt{2}}
 \end{pmatrix},\ \ \ {\rm where}\ \ \ \epsilon=\frac{3 m_D^2 S}{16 M_R^2 \left( m_L-\frac{m_D^2}{M_R} s\right)}.
 \label{eq:Le_7}
\end{equation}
However, this matrix still has to be combined with the mixing coming from the charged leptons. A suitable charged lepton mixing matrix is given by~\cite{Frampton:2004ud}
\begin{equation}
 U_e =
 \begin{pmatrix}
 1-\lambda^2/2 & \lambda & \lambda^3\\
 -\lambda & 1-\lambda^2/2 & \lambda^2\\
 \lambda^3 & -\lambda^2 & 1
 \end{pmatrix} + \mathcal{O}(\lambda^4),
 \label{eq:Le_8}
\end{equation}
where $\lambda \simeq 0.20$ is the deviation of $\theta_{12}$ from $\frac{\pi}{2}$. Building the PMNS matrix $U_{\rm PMNS} = U_e^\dagger U_\nu$ in this basis (with non-diagonal charged lepton mass matrix), one obtains for the mixing angles
\begin{equation}
 (\tan^2\theta_{12}, |U_{e3}|, \sin^2 2\theta_{23}) \simeq \left( 1- 2\sqrt{2}\lambda - 2 \sqrt{2}\lambda^3 + 4\lambda^4, \frac{\lambda}{\sqrt{2}}, 1-4\lambda^4 \right),
 \label{eq:LeMIX}
\end{equation}
or $(\sin^2 \theta_{12}, \sin^2 \theta_{13}, \sin^2 \theta_{23}) \simeq (0.295, 0.020, 0.460)$, which is all consistent with data at the 3$\sigma$ level, cf.\ Eq.~\eqref{eq:angles_IH}.

Finally, it remains to be shown that the charged lepton squared mass matrix necessary for Eq.~\eqref{eq:Le_8} to arise,
\begin{equation}
 M_e M_e^\dagger \simeq
 \begin{pmatrix}
 m^2_e+m^2_\mu\lambda^2 & m^2_\mu\lambda & 0\\
 m^2_\mu \lambda & m^2_\mu & 0\\
 0 & 0 & m^2_\tau
 \end{pmatrix} + \mathcal{O}(\lambda^3),
 \label{eq:Le_9}
\end{equation}
can be obtained from softly broken $\mathcal{F}$ symmetry. Indeed, if the symmetry preserving charged lepton mass matrix as derived from the assignments in Tab.~\ref{tab:LeLmLt} is extended by soft breaking terms, $s_e^{\alpha \beta}$, which would be zero in the symmetry preserving limit, one obtains
\begin{equation}
 M_e =
 \begin{pmatrix}
 M^{e e}_e & 0 & 0 \\
 0 & M^{\mu \mu}_e & M^{\mu \tau}_e \\
 0 & M^{\tau \mu}_e & M^{\tau \tau}_e
 \end{pmatrix} \to
 \begin{pmatrix}
 M^{e e}_e & s^{e \mu}_e & s^{ e \tau}_e \\
 s^{e \mu}_e & M^{\mu \mu}_e & M^{\mu \tau}_e \\
 s^{e \tau}_e & M^{\mu \tau}_e & M^{\tau \tau}_e
 \end{pmatrix}.
 \label{eq:Le_10}
\end{equation}
This yields Eq.~\eqref{eq:Le_9} for the identification $s^{e \mu}_e = M^{\mu \mu}_e \lambda$, $s^{e \tau}_e = 0$, $M^{e e}_e = m_e$, $M^{\mu \mu}_e = m_\mu$, and $M^{\tau \tau}_e=m_\tau$. However, it should be pointed out that the form of the matrix in Eq.~\eqref{eq:Le_10}, even though it can be obtained in accordance with $\mathcal{F}$ symmetry is not a strict prediction of the model.

The corresponding mass shifting scheme is depicted in Fig.~\ref{fig:Le-scheme}, and it clearly resembles the bottom-up scheme from the left panel of Fig.~\ref{fig:keV-schemes}. As illustrated in the figure, the masses of the two heavier neutrinos are exactly degenerate in the symmetry limit, and the lightest sterile neutrino is exactly massless. Once the symmetry is broken, the additional small soft terms will not only lift the degeneracy between the two heavier states, but they will also give a small mass, of the order of the breaking terms, to the lightest state. Hence, this model motivates a mass pattern $M_1 = \mathcal{O}({\rm keV}) \ll M_{2,3}$.

\begin{figure}[pb]
\centerline{
\psfig{file=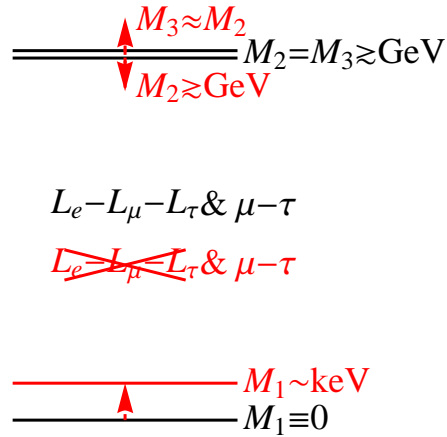,width=6.0cm}
}
\vspace*{8pt}
\caption{\label{fig:Le-scheme}Mass shifting scheme for softly broken $\mathcal{F}$ symmetry. Breaking the symmetry lifts the sterile neutrino mass $M_1$ to a non-zero value, which clearly resembles the bottom-up scheme, cf.\ left panel of Fig.~\ref{fig:keV-schemes}. (Figure similar to Fig.~1.a in Ref.~\cite{Merle:2011yv}.)}
\end{figure}

\subsubsection{\label{sec:keV_flav_Q6}$Q_6$ symmetry}	

Another proposal is a model for keV sterile neutrinos based on the symmetry group $Q_6$, which was presented in Ref.~\refcite{Araki:2011zg}. $Q_6$ is nothing else than the double cover of the dihedral group $D_3=S_3$,\footnote{$D_3$ is the full symmetry group of a triangle, \emph{i.e.}, three rotations and three reflections, which amounts to a total of six elements. This is why some authors use the notation ``$D_6$'' for this group, which is unfortunately a bit ambiguous as it could also be interpreted as the full symmetry group of a hexagon~\cite{Ramond}.} as explained, \emph{e.g.}, in Refs.~\refcite{Frampton:1994rk,Babu:2004tn,Kajiyama:2005rk}.

The required group theory of $Q_6$ is outlined in Ref.~\refcite{Araki:2011zg}, but we will comment on a couple of things in passing. This group is not much more difficult than $A_4$: it has twelve elements, and four different one-dimensional irreps $\mathbf{1}$, $\mathbf{1'}$, $\mathbf{1''}$, $\mathbf{1'''}$, as well as two different two-dimensional irrpes $\mathbf{2}$, $\mathbf{2'}$. Note that $\mathbf{2}$ is a \emph{pseudoreal} representation, \emph{i.e.}, it contains only matrices with real entries, but the squares of some of them are nevertheless negative. An important difference to $A_4$ is, however, that $Q_6$ has no three-dimensional irreps, so that a model based on $Q_6$ can unify at most two of the three known generations of matter. Because of this fact, one might criticize the aesthetic power of this group, but one could also argue that part of the data actually seems to favor a unification of only two generations of leptons. For example, when looking at charged lepton masses, $m_\mu$ and $m_\tau$ are much closer to each other than they are to $m_e$. One can always take on a very pragmatic point of view and judge a model by its consistency with experimental data rather than by requiring it to predict certain patterns that may or may not exist in Nature.

The model under consideration is relatively economic since only three singlet flavons $S_{x,y,z}$ and one doublet flavon $D$ are introduced in addition to the lepton sector. The detailed charge and representation assignment is outlined in Tab.~\ref{tab:Q6}. Also in this model, two auxiliary symmetries $Z_{2,3}$ have been introduced in order to suppress certain terms in the Lagrangian.

\begin{table}[ph]
\tbl{Particle content and representation/charge assignments (in our notation) of the $SU(2)_L \times Q_6 \times Z_2 \times Z_3$ model with one keV sterile neutrino~\cite{Barry:2011fp}. $SU(2)_L$ is the usual SM gauge symmetry and $Z_{2,3}$ are additional auxiliary symmetries. For $SU(2)_L$ and for $Q_6$ the representations are given, while for $Z_{2,3}$ the charge is indicated. Note that $\omega = e^{2\pi i /3}$ and that $\mathbf{2}^* = \mathbf{2}$ in $Q_6$.}
{\begin{tabular}{@{}lccccccccccc@{}} \toprule
Field & $L_1$ & $L_D\equiv L_{2,3}$ & $\overline{e_R}$ & $\overline{e_{DR}} = (\overline{\mu_R},\overline{\tau_R})$ & $\overline{N_1}$ & $\overline{N_D} \equiv \overline{N_{2,3}}$ & $H$ & $S_x$ & $S_y$ & $S_z$ & $D$ \\
\colrule
$SU(2)_L$ & $\mathbf{2}$ & $\mathbf{2}$ & $\mathbf{1}$ & $\mathbf{1}$ & $\mathbf{1}$ & $\mathbf{1}$ & $\mathbf{2}$ & $\mathbf{1}$ & $\mathbf{1}$ & $\mathbf{1}$ & $\mathbf{1}$ \\
$Q_6$ & $\mathbf{1}$ & $\mathbf{2'}$ & $\mathbf{1'}$ & $\mathbf{2}$ & $\mathbf{1''}$ & $\mathbf{2'}$ & $\mathbf{1}$ & $\mathbf{1''}$ & $\mathbf{1'''}$ & $\mathbf{1}$ & $\mathbf{2}$ \\
$Z_2$ & $+$ & $+$ & $+$ & $+$ & $+$ & $-$ & $+$ & $+$ & $+$ & $-$ & $-$ \\
$Z_3$ & $1$ & $1$ & $\omega$ & $\omega^2$ & $1$ & $1$ & $1$ & $\omega^2$ & $\omega^2$ & $1$ & $1$ \\
\botrule
\end{tabular}
\label{tab:Q6}}
\end{table}

Using the representation content as well as the group theory of $Q_6$, we can again write down the lowest order Lagrangian for the charged leptons,\footnote{Note that this is the Hermitean conjugate of the Lagrangian presented in Ref.~\refcite{Araki:2011zg}. Furthermore, we take the mass terms to be negative, as is conventional for a Dirac mass, although one could simply redefine the couplings to account for that. Finally, since we are effectively taking the complex conjugate, one has to keep in mind that $\mathbf{1''}^* = \mathbf{1'''}$ in $Q_6$.}
\begin{equation}
 \mathcal{L}_e = -\frac{Y_x^*}{\Lambda} (\overline{e_{DR}} H L_D)_{\mathbf{1''}} S_x^* - \frac{Y_y^*}{\Lambda} (\overline{e_{DR}} H L_D)_{\mathbf{1'''}} S_y^* - \frac{Y_e^2}{\Lambda^2} (\overline{e_R} H L_1)_{\mathbf{1'}} (S_x^*)^2 + h.c.
 \label{eq:Q6_1}
\end{equation} 
Here, we have made use of the composition rule for two doublets $\mathbf{2} = (x_1, x_2)^T$ and $\mathbf{2'} = (y_1, y_2)^T$, whose tensor product $\mathbf{2} \otimes \mathbf{2'}$ can be decomposed into two singlets, $\mathbf{1''} = (x_1 y_1 - x_2 y_2)$ and $\mathbf{1'''} = (x_1 y_1 + x_2 y_2)$, and one doublet $\mathbf{2} = (x_2 y_1, x_1 y_2)^T$. Furthermore, we have used $\mathbf{1'} \otimes \mathbf{1'} = \mathbf{1}$ and $\mathbf{1''} \otimes \mathbf{1''} = \mathbf{1'''} \otimes \mathbf{1'''} = \mathbf{1'}$. Note that the whole Lagrangian in Eq.~\eqref{eq:Q6_1} arises as correction to the leading order, and hence all the terms inside are non-renormalizable and suppressed by a high mass scale $\Lambda$. Assuming singlet flavon VEVs $\langle S_i \rangle = s_i$ for $i=x,y,z$, the charged lepton masses obtained after electroweak symmetry breaking are given by
\begin{equation}
 m_e = \left(\frac{s_x}{\Lambda}\right)^2 Y_e v,\ \ \ m_\mu = \frac{1}{\Lambda} \left( Y_x s_x + Y_y s_y \right)v,\ \ \ m_\tau = \frac{1}{\Lambda} \left( Y_x s_x - Y_y s_y \right) v.
 \label{eq:Q6_2}
\end{equation} 
Indeed, a tendency that had been anticipated reveals itself in these eigenvalues: the representation content was chosen such that the muon and the tau unify in doublets of $Q_6$, cf.\ Tab.~\ref{tab:Q6}, while the electron transforms as a singlet. Accordingly, the masses of $\mu$ and $\tau$ in Eq.~\eqref{eq:Q6_2} are identical up to a splitting $2 Y_y s_y v$, while the electron mass is suppressed (compared to $m_{\mu,\tau}$) by another power of the high energy scale $\Lambda$.

The corresponding Lagrangian for the neutrino sector turns out to be
\begin{eqnarray}
\mathcal{L}_\nu &=& - \frac{\alpha^*}{\Lambda} (\overline{N_D} D^*)_{\mathbf{1}} \tilde H L_1 - \frac{\beta^*}{\Lambda} (\overline{N_D} \tilde H L_D)_{\mathbf{2'}} D^* - \frac{\gamma^*}{\Lambda} (\overline{N_D} \tilde H L_D^*)_{\mathbf{1}} S_x \nonumber\\
 && - \frac{\delta^*}{\Lambda^3} (\overline{N_1})_{\mathbf{1'''}} \tilde H L_1 (S_x^*)^3 - \frac{\epsilon}{\Lambda^3} (\overline{N_D} \tilde H L_D)_{\mathbf{1'}} S_x^* S_y S_z^* \label{eq:Q6_3} \\
 && - \frac{M_a}{2} [\overline{(N_D)^c} N_D]_{\mathbf{1}} - \frac{1}{2} \frac{M_b}{\Lambda^2} [\overline{(N_D)^c} N_D]_{\mathbf{2'}} \left[(D^*)^2\right]_{\mathbf{2'}} - \frac{1}{2} \frac{M_c}{\Lambda^2} [\overline{(N_1)^c} N_1]_{\mathbf{1'}} S_x^* S_y + h.c., \nonumber
\end{eqnarray}
where the RH Majorana masses are taken to be real. Note that we have introduced a factor $\frac{1}{2}$ compared to the corresponding equation in Ref.~\refcite{Araki:2011zg}, in order to coincide with the notation in the remainder of this review. In Eq.~\eqref{eq:Q6_3}, we have made use of further tensor products, such as $\mathbf{1''} \otimes \mathbf{1'''} = \mathbf{1}$ or $\mathbf{2'} \otimes \mathbf{2'} = \mathbf{1} \oplus \mathbf{1'} \oplus \mathbf{2'}$. The reader is invited to work out the exact form and the invariance of the remaining terms in Eq.~\eqref{eq:Q6_3} her or himself, using the information on the group $Q_6$ given in Ref.~\refcite{Araki:2011zg}.

The most important aspect of Eq.~\eqref{eq:Q6_3} can be found in the last line of the equation: since $N_D$ transforms as a doublet $\mathbf{2'}$ under $Q_6$, the combination $\overline{(N_D)^c} N_D$ indeed contains a trivial singlet $\mathbf{1}$. To be precise, the tensor product of two doublets $\mathbf{2'} = (x_1, x_2)^T$ and $\mathbf{2'} = (y_1, y_2)^T$ contains a term $\mathbf{1} = (x_1 y_2 - x_2 y_1)$, and the conjugations inside $\overline{(N_1)^c}$ do not matter since $\mathbf{2'}$ is real. In conclusion, the corresponding term in the Lagrangian yields a tree-level mass for $N_{2,3}$,
\begin{equation}
 - \frac{M_a}{2} [\overline{(N_D)^c} N_D]_{\mathbf{1}} + h.c. = - \frac{M_a}{2} (\overline{N_2} N_3^c - \overline{N_3} N_2^c) + h.c.
 \label{eq:Q6_4}
\end{equation}
However, the first RH neutrino, $N_1$, transforms as a non-trivial singlet and, similar to what happens to the electron in Eq.~\eqref{eq:Q6_2}, its mass is suppressed by two more powers of the high scale $\Lambda$ compared to the masses of $N_{2,3}$. A natural hierarchy $M_1 \ll M_{2,3}$ in the sterile neutrino sector is predicted, and hence we have a good reason for taking $N_1$ to be the keV sterile neutrino.

Taking the VEV of the doublet flavon to be $\langle D \rangle = (d_1, d_2)^T$, one can immediately read off $M_R$ from Eq.~\eqref{eq:Q6_3},
\begin{equation}
 M_R=
\begin{pmatrix}
 0 & 0 & 0 \\
 0 & 0 & M_a \\
 0 & M_a & 0
\end{pmatrix}
+
\frac{1}{\Lambda^2}
 \begin{pmatrix}
 M_c s_x s_y & 0 & 0 \\
 0 & M_b d_2^2 & 0 \\
 0 & 0 & M_b d_1^2
\end{pmatrix},
 \label{eq:Q6_5}
\end{equation}
whose eigenvalues are given by $M_1 = \frac{s_x s_y}{\Lambda^2} M_c$ and $M_{2,3} \simeq M_a \mp \frac{M_b}{2\Lambda^2}(d_1^2 + d_2^2)$. This can be viewed as a realization of the Bottom-up scheme, cf.\ Fig.~\ref{fig:Q6-scheme}: while the zeroth and first order contributions to $M_1$ vanish exactly, the second-order correction in $1/\Lambda$ corrects this natural value (as well as the exact degeneracy between $M_2$ and $M_3$) to yield a finite mass $M_1 \ll M_{2,3}$.

The Dirac mass matrix, in turn, is given by
\begin{equation}
 m_D^T =
 \frac{1}{\Lambda}
 \begin{pmatrix}
 0 & \alpha d_2 & \alpha d_1 \\
 0 & \beta d_1  & \gamma s_z \\
 0 & \gamma s_z & \beta d_2
 \end{pmatrix} v
 +
 \frac{1}{\Lambda^3}
 \begin{pmatrix}
 \delta s_x^3 & 0 & 0 \\
 0 & 0 & \epsilon s_x s_y s_z \\
 0 & -\epsilon s_x s_y s_z & 0
 \end{pmatrix} v.
 \label{eq:Q6_6}
\end{equation}
Due to $M_1 \ll M_{2,3}$, the authors of Ref.~\refcite{Araki:2011zg} chose to integrate out $N_{2,3}$ only, although they could actually have applied the seesaw formula according to the keV seesaw practicality theorem, cf.\ Sec.~\ref{sec:keV_general_seesaw}. Restraining from that, according to the seesaw fair play rule, cf.\ Sec.~\ref{sec:keV_general_seesaw}, one light neutrino mass eigenvalue will be precisely zero.

Applying this procedure, one arrives at a $4\times 4$ neutrino mass matrix given by
\begin{equation}
 M_\nu^{4\times 4} =
 \begin{pmatrix}
 M_\nu^{3\times 3} & \Delta \\
 \Delta^T & M_1
 \end{pmatrix},
 \label{eq:Q6_7}
\end{equation}
where $\Delta=(\delta' v, 0, 0)^T$, and the upper left $3\times 3$ block is
\begin{eqnarray}
 M_\nu^{3\times 3} &=& \frac{v^2}{M_a} 
 \begin{pmatrix} 
 2 {\alpha'}^2 & \alpha' (\beta' + \gamma') & \alpha' (\beta' + \gamma') \\
 \alpha' (\beta' + \gamma') & 2 \beta' \gamma' & {\beta'}^2 + {\gamma'}^2 \\
 \alpha' (\beta' + \gamma') & {\beta'}^2 + {\gamma'}^2 & 2 \beta' \gamma'
 \end{pmatrix} \nonumber \\
 && + \epsilon_d \frac{v^2}{M_a}
 \begin{pmatrix}
 0 & \alpha' (2 \beta' - \gamma') & - \alpha' (2 \beta' - \gamma') \\
 \alpha' (2 \beta' - \gamma') & 2 \beta' \gamma' & 0 \\
 - \alpha' (2 \beta' - \gamma') & 0 & - 2 \beta' \gamma'
 \end{pmatrix}
 + \mathcal{O}(\epsilon^2_d),
 \label{eq:Q6_8}
\end{eqnarray}
where $d \equiv \frac{d_1 + d_2}{2}$ and $\epsilon_d \equiv \frac{d_1 - d_2}{2 d}$, such that $\epsilon_d = 0$ in the limit $d_1 = d_2$. Furthermore, the abbreviations used are $\alpha' \equiv \alpha \frac{d}{\Lambda}$, $\beta' \equiv \beta \frac{d}{\Lambda}$, $\gamma' \equiv \gamma\frac{s_z}{\Lambda}$, and $\delta' \equiv \delta \left(\frac{s_x}{\Lambda}\right)^3$. Note that the latter definition suggests that the parameter $\delta'$ is actually quite small.

In Ref.~\refcite{Araki:2011zg}, the $4\times 4$ neutrino mass matrix is approximately diagonalized in the limit of small $\delta'$, which is done by a matrix $\overline{U}_{4 \times 4}$ given by
\begin{equation}
 \overline{U}_{4 \times 4} \simeq
 \begin{pmatrix}
 i (\mathbf{1} - \frac{1}{2} R \otimes R^\dagger) & R\\
 -i R^\dagger & \mathbf{1} - \frac{1}{2} R^\dagger R
 \end{pmatrix}
 \begin{pmatrix}
 V_A & 0\\
 0 & 1
 \end{pmatrix},\  {\rm where}\ V_A \simeq
 \begin{pmatrix}
  \cos \theta e^{i\rho} & \sin \theta e^{i\rho} & 0\\
 -\frac{\sin \theta}{\sqrt{2}} e^{i\sigma} & \frac{\cos \theta}{\sqrt{2}} e^{i\sigma} & -\frac{1}{\sqrt{2}}\\
 -\frac{\sin \theta}{\sqrt{2}} e^{i\sigma} & \frac{\cos \theta}{\sqrt{2}} e^{i\sigma} & \frac{1}{\sqrt{2}}
 \end{pmatrix}
 \label{eq:Q6_9}
\end{equation}
and $R = (\delta' v / M_1, 0, 0)$. Furthermore, the abbreviations $\rho = {\rm arg} (\alpha')$, $\sigma = {\rm arg}(\beta' + \gamma')$, and $\sin \theta = \frac{ \sqrt{2} |\alpha'|}{\sqrt{|\beta' + \gamma' |^2 + 2 |\alpha'|^2}}$ have been used. Note that $V_A$ is exactly unitary and $\overline{U}_{4 \times 4}$ is approximately unitary for small $|R|$, which indeed requires a relatively tiny $\delta'$ [in Ref.~\refcite{Araki:2011zg} it is taken to be of $\mathcal{O}(10^{-13.5})$]. However, the expression given for $V_A$ only diagonalizes the leading part of $M_\nu^{3\times 3}$, \emph{i.e.}, the one for which $\epsilon_d = 0$. Furthermore, it is exactly this suppression of the parameter $\delta'$ which naturally leads to a suppression of the active-sterile mixing, which is in this model given by $\theta_1^2 \simeq {\delta'}^2 v^2/M_1^2 \lesssim 10^{-11}$, while corrections from $N_{2,3}$ are tiny. This suppression prevents the model from too much trouble with the X-ray bound.

The mass eigenvalues of the light neutrinos resulting from the diagonalization of $M_\nu^{3\times 3}$ from Eq.~\eqref{eq:Q6_8} are given by
\begin{equation}
 m_1 = 0\ \ \ ,\ \ \ m_2 = \frac{|\beta' + \gamma'|^2 + 2 |\alpha'|^2}{M_a}v^2\ \ \ ,\ \ \ m_3 = -\frac{(\beta' - \gamma')^2}{M_a} v^2\ \ \ ,
 \label{eq:Q6_10}
\end{equation}
where one light neutrino is massless, as anticipated. In turn, this model predicts not only normal mass ordering, but even a hierarchy among the light neutrinos.

The authors also calculate the neutrino mixing angles resulting from Eq.~\eqref{eq:Q6_9}. They argue that the phases $\rho$ and $\sigma$ are unphysical, which is correct, but still they appear in their expression for $\theta_{23}$. This is not entirely correct, as one could have seen when following the detailed diagonalization procedure of the leptonic mass matrices as outlined in the Appendices~B of Refs.~\refcite{Antusch:2005gp,Antusch:2003kp}. However, one can easily cure the results for the mixing angles as obtained in Ref.~\refcite{Araki:2011zg} by simply setting $\rho = \sigma = 0$ in their epxressions. Then, since $U_e = {\bf 1}$ due to the charged lepton mass matrix being diagonal, the physical mixing angles are identical to the mixing angles from the neutrino sector alone, $\theta_{ij} = \theta_{ij}^\nu$. Their lowest order expressions are given by
\begin{equation}
 \sin \theta_{12} \simeq \sin \theta\ \ \ ,\ \ \ \tan \theta_{23} \simeq \left|1 - \frac{4 \beta' \gamma' }{(\beta' - \gamma')^2} \epsilon_d \right|\ \ \ ,\ \ \ \sin \theta_{13} \simeq \left|\frac{\sqrt{2} \alpha' (2 \beta' - \gamma')}{(\beta' - \gamma')^2}\epsilon_d\right|,
 \label{eq:Q6_11}
\end{equation}
where corrections of $\mathcal{O}(\epsilon_d^2)$ have been neglected. Note that the angle $\theta_{13}$ as well as the deviation of $\theta_{23}$ from $\frac{\pi}{2}$ arise only as higher order corrections. Nevertheless, we now know that the parameter of $\epsilon_d$ must be large enough to yield an angle $\theta_{13}$ as large as measured, cf.\ Eqs.~\eqref{eq:angles_NH} and~\eqref{eq:angles_IH}. However, the mixing angle $\theta_{12}$ is not predicted and must be matched to the data.

Let us end this section by an interesting feature of the model, which is illustrative for how flavor models could be tested. Comparing the equations for $\theta_{13}$ and $\theta_{23}$ in Eq.~\eqref{eq:Q6_11}, one can derive a \emph{correlation} between these two angles:
\begin{equation}
 \tan \theta_{23} \simeq 1 - \frac{2 \sqrt{2} \beta' \gamma' }{|\alpha' (2 \beta' - \gamma')|} \sin \theta_{13}.
 \label{eq:Q6_12}
\end{equation}
This equation implies that $\theta_{23} = \frac{\pi}{2}$ if $\theta_{13} = 0$, but also that $\theta_{23} \neq \frac{\pi}{2}$ if $\theta_{13} \neq 0$. More precisely, the sign of $\beta' \gamma'$ determines whether $\theta_{23}$ is smaller or larger than $\frac{\pi}{2}$. This is particularly interesting when looking at the newest global fits, cf.\ Refs.~\refcite{Tortola:2012te,Fogli:2012ua,GonzalezGarcia:2012sz}. This correlation has been numerically verified and matched to the actual data in Ref.~\refcite{Araki:2011zg}. Such correlations are the keys to identify whether a certain model finally turns out to be close to reality, or not.

\begin{figure}[pb]
\centerline{
\psfig{file=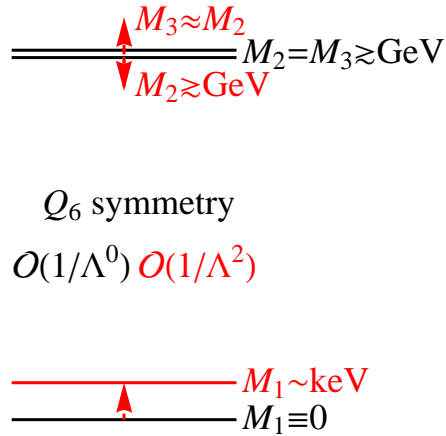,width=6.0cm}
}
\vspace*{8pt}
\caption{\label{fig:Q6-scheme}Mass shifting scheme for the $Q_6$ symmetry. While at zeroth order in the cutoff scale $N_1$ is massless and $N_{2,3}$ are degenerate, both these problems are solved if second order corrections are taken into account. This mechanism clearly resembles the bottom-up scheme, cf.\ left panel of Fig.~\ref{fig:keV-schemes}.}
\end{figure}

After having discussed several models based on flavor symmetries, we will now turn to other ways to justify the keV scale for sterile neutrinos. A completely different mechanism to generate very small mass scales, based on the existence of extra spatial dimensions, will be presented in the next section.

\subsection{\label{sec:keV_split}Models based on the split seesaw mechanism}	

A somewhat orthogonal approach to explain the keV scale is by mechanisms that are not based on some flavor symmetry. One possibility is to have a non-trivial space-time geometry, and the most generic approach is making use of extra spatial dimensions. The corresponding mechanism which can explain the existence of a keV scale is the so-called \emph{split seesaw mechanism}~\cite{Kusenko:2010ik}. Unfortunately, we will here not be able to review all the theory behind extra dimensions. Instead, we guide the reader to excellent introductions to and reviews of extra dimensions, such as Refs.~\refcite{Rizzo:2004kr,AgasheED,CzakiED}. However, the amount of knowledge on extra dimensions that is necessary to understand the key concepts of the split seesaw mechanism is relatively limited, and we will in addition give brief explanations of the required concepts in passing.

\subsubsection{\label{sec:keV_split_mechamism}The split seesaw mechanism}	

The \emph{split seesaw mechanism} has been proposed in Ref.~\refcite{Kusenko:2010ik}. It is based on a 5-dimensional (5D) model compactified on $S^1/Z_2$: the extra dimension is rolled to a circle $S^1$, which is parametrized by an angle $\phi \in [0, 2 \pi R]$. In addition, the pairs of coordinates $(\phi, 2 \pi R - \phi)$ on the circle are identified with each other, which is called \emph{orbifolding}. Hence, the 5D coordinate can be taken to be $y \in [0, l]$, where $l=\pi R$. The model contains two \emph{branes} (= 4D subspaces which are located at specific points of the extra dimension), one at $y=0$ to which all SM fields are confined and another \emph{hidden} brane at $y=l$. Such a setting leads to a reduced 4D Planck scale $M_P$, which is related to the fundamental Planck scale $M_0$ by $M_P^2 = M_0^3 l$. Note that this means $M_P < M_0$, due to $l^{-1} < M_0$.

The next step is to write down the 5D action of the model for a 5D Dirac spinor $\Psi = (\Psi_L, \Psi_R)^T$ with a \emph{bulk mass} $m$ in five dimensions:
\begin{equation}
 S = \int d^4 x \int\limits_0^l dy\ M_0\ (i \overline{\Psi} \Gamma^A \partial_A \Psi - m \overline{\Psi} \Psi).
 \label{eq:split_1}
\end{equation}
This action looks pretty similar to a generic 4D action, apart from one more integration over the extra dimension $y$, an additional insertion of $M_0$ in order to make the full action dimensionless, and and index $A$ instead of $\mu$. Furthermore, the field $\Psi$ depends on all space-time coordinates, $\Psi = \Psi (x^\mu, y)$. Note that the index $A$ runs over $0$, $1$, $2$, $3$, $5$, where $5$ labels the extra dimensional coordinate $y$, and the corresponding gamma matrices are defined as
\begin{equation}
 \Gamma^\mu = \gamma^\mu\ \ , \ \ \ \Gamma^5 = -i
 \begin{pmatrix}
 \mathbf{1} & 0 \\
 0 & - \mathbf{1}
 \end{pmatrix}.
 \label{eq:split_2}
\end{equation}
The important point is that the LH and RH projections are defined by the new $\Gamma^5$ matrix, $\Psi_{L,R} \equiv \frac{1}{2} (1\mp \Gamma^5) \Psi$ (see Ref.~\refcite{AgasheED} for a discussion of the issue of defining chirality in theories with extra dimensions).

The next step is to integrate out the extra dimension, in order to obtain an effective model in 4D. A typical feature of such settings is the existence of a \emph{Kaluza-Klein (KK) tower}, which is a stack of 4D fields with identical quantum numbers (up to chirality) but increasing masses. This existence of this tower is easy to understand: as well-known from non-relativistic quantum mechanics, there is an infinite number of states with increasing energy inside a potential well. Essentially the same happens within the extra dimension.

A suitable ansatz for a decomposition in KK modes in our case is given by~\cite{Grossman:1999ra}:
\begin{equation}
 \Psi_{L,R} (x^\mu, y) = \sum_n \psi_{L,R}^{(n)} (x^\mu)\ f_{L,R}^{(n)}(y),
 \label{eq:split_3}
\end{equation}
where the index $n$ labels the different modes in the KK tower. Note the separation of the coordinates on the right-hand side of Eq.~\eqref{eq:split_3}. This ansatz has to be inserted into Eq.~\eqref{eq:split_1}, where we should be careful to note that $\overline{\Psi} \Gamma^\mu \partial_\mu \Psi = \overline{\Psi_L} \overline{\sigma}^\mu \partial_\mu \Psi_L + \overline{\Psi_R} \sigma^\mu \partial_\mu \Psi_R$ while $\overline{\Psi} \Gamma^5 \partial_5 \Psi = \overline{\Psi_L} (-i \mathbf{1}) \partial_5 \Psi_R + \overline{\Psi_R} (+i \mathbf{1}) \partial_5 \Psi_L$, since $\Gamma^5$ anti-commutes with all $\Gamma^\mu$ but of course commutes with itself. Inserting the KK-decomposition, Eq.~\eqref{eq:split_3}, one obtains
\begin{eqnarray}
 S &=& \int d^4 x\ \sum_{n,k} \left[ \overline{\psi_L^{(n)}} i \overline{\sigma}^\mu \partial_\mu \psi_L^{(k)} \right] \left( \int dy\ M_0\ {f_L^{(n)}}^*(y) f_L^{(k)}(y) \right) \nonumber\\
 && + \int d^4 x\ \sum_{n,k} \left[ \overline{\psi_R^{(n)}} i \sigma^\mu \partial_\mu \psi_R^{(k)} \right] \left( \int dy\ M_0\ {f_R^{(n)}}^*(y) f_R^{(k)}(y) \right)\nonumber \\
 && + \int d^4 x\ \sum_{n,k}\overline{\psi_L^{(n)}} \psi_R^{(k)}  \left( \int dy\ M_0\ {f_L^{(n)}}^*(y) (-\partial_y - m)  f_R^{(k)}(y) \right) \nonumber \\
 && + \int d^4 x\ \sum_{n,k}\overline{\psi_R^{(n)}} \psi_L^{(k)}  \left( \int dy\ M_0\ {f_R^{(n)}}^*(y) (+\partial_y - m)  f_L^{(k)}(y) \right).
 \label{eq:split_4}
\end{eqnarray}
In order to yield the correct 4D interpretation, the kinetic terms should have their canonical normalization. For this to be the case, the action in Eq.~\eqref{eq:split_4} needs to be equivalent to
\begin{equation}
S = \int d^4 x\ \sum_n \left(  \overline{\psi_L^{(n)}} i \overline{\sigma}^\mu \partial_\mu \psi_L^{(n)}  - m_n \overline{\psi_L^{(n)}} \psi_R^{(n)}  + ( L \leftrightarrow R, \overline{\sigma} \to \sigma ) \right),
 \label{eq:split_5}
\end{equation}
which will be fulfilled if the following two relations hold,
\begin{equation}
 \int dy\ M_0\ {f_{L,R}^{(n)}}^*(y) f_{L,R}^{(k)}(y)= \delta_{nk}\ \ \ {\rm and}\ \ \ (\pm \partial_y - m)  f_{L,R}^{(n)}(y) = m_n f_{L,R}^{(n)}(y).
 \label{eq:split_6}
\end{equation}
The solutions of these equations encode the extra-dimensional part of the wave function, which is called the \emph{bulk profile}. For the massless zero mode, defined by $m_0 \equiv 0$, this bulk profile is obtained by solving the second Eq.~\eqref{eq:split_6}, and it is given by a constant $C$ times a function of $y$:
\begin{equation}
 f_{L,R}^{(0)}(y) = C e^{\mp m y}\ , \ \ {\rm where}\ \ \ C = \sqrt{\frac{2 m}{e^{2 m l}-1}} \frac{1}{\sqrt{M_0}}.
 \label{eq:split_7}
\end{equation}
The correct normalization, cf.\ first Eq.~\eqref{eq:split_6}, is fixed to yield a canonically normalized 4D kinetic term, which enforces the given form of $C$. The solution given in Eq.~\eqref{eq:split_7} makes it right away evident that, for the RH solutions, the function value at the hidden brane ($y=l$) is much larger than the one at the SM-brane ($y=0$), since $\frac{f_R^{(0)}(y=l)}{f_R^{(0)}(y=0)} = e^{m l} \gg 1$, \emph{i.e.}, the geometry of the extra dimension influences the shape of the 4D theories on both branes, and in particular on the SM-brane where we reside. This observation will be crucial in the what follows.

In order to obtain the split seesaw mechanism, one applies the above procedure to a 5D RH-neutrino field whose zero mode is purely right-handed and which shares at the SM brane both, a Yukawa interaction $\tilde \lambda_{i\alpha} = \mathcal{O}(1)$ and a baryon number minus lepton number violating Majorana mass $\kappa_i v_{B-L}$,
\begin{eqnarray}
 S &=& \int d^4 x \int dy\ \Bigg[  M_0\ \left( \overline{\Psi^{(0)}_{iR}} i \Gamma^A \partial_A \Psi^{(0)}_{iR} - m_i \overline{\Psi^{(0)}_{iR}} \Psi^{(0)}_{iR} \right) \nonumber\\
 && - \delta(y) \left( \frac{\kappa_i}{2} v_{B-L} \overline{(\Psi^{(0)}_{iR})^c} \Psi^{(0)}_{iR} + \tilde \lambda_{i\alpha} \overline{\Psi^{(0)}_{iR}} L_\alpha H \right) \Bigg],
 \label{eq:split_8}
\end{eqnarray}
where $i=1,2,3$ is a generation index. Applying Eq.~\eqref{eq:split_7}, one obtains the effective 4D sterile neutrino masses and Yukawa couplings,
\begin{equation}
 M_i = \kappa_i \frac{v_{B-L}}{M_0} \frac{2 m_i}{e^{2 m_i l}-1}\ \ \ , \ \ \ \lambda_{i\alpha} = \frac{\tilde \lambda_{i\alpha}}{\sqrt{M_0}} \sqrt{\frac{2 m_i}{e^{2 m_i l}-1}} \equiv \tilde \lambda_{i\alpha} \sqrt{\frac{M_i}{\kappa_i v_{B-L}}}.
 \label{eq:split_9}
\end{equation}
Now we can see what happens: the sterile neutrino masses $M_i$ receive two suppressions in Eq.~\eqref{eq:split_9}, one from the ratio between the Majorana ($B-L$) scale and the fundamental Planck scale, $\frac{v_{B-L}}{M_0} < 1$, and another even more significant one from the term with the exponential, which suppresses the natural mass scale $m_i$ as $e^{-2 m_i l}$ for $l m_i \gg 1$. It is this latter exponential suppression which makes the split seesaw mechanism so strong: a relatively mild mass splitting $m_3 < m_2 < m_1$ is enhanced by the exponential factor to yield physical masses $M_1 \ll M_2 \ll M_3$. The corresponding mass shifting scheme is depicted in Fig.~\ref{fig:RS-scheme}. Clearly, this type of suppression follows the top-down scheme, cf.\ right panel of Fig.~\ref{fig:keV-schemes}, just as FN type models do.

\begin{figure}[pb]
\centerline{
\psfig{file=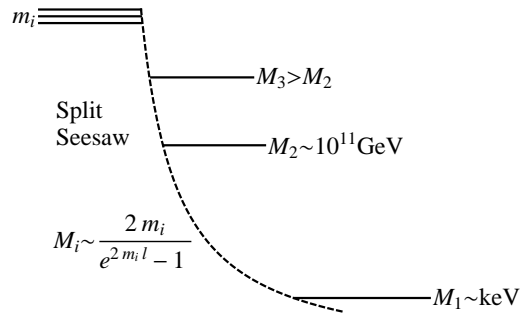,width=7.0cm}
}
\vspace*{8pt}
\caption{\label{fig:RS-scheme}Mass shifting scheme for split seesaw. The wave functions in the extra dimension suppress the physical sterile neutrino masses, which clearly resembles the top-down scheme, cf.\ right panel of Fig.~\ref{fig:keV-schemes}. (Figure similar to Fig.~1.b in Ref.~\cite{Merle:2011yv}.)}
\end{figure}

An interesting feature of the split seesaw mechanism is that the seesaw is guaranteed to work, similar to what happens in FN inspired models, cf.\ Sec.~\ref{sec:keV_FN}. Indeed, calculating the seesaw formula like in Eq.~\eqref{eq:seesawI_2} with $v = \langle H^0 \rangle$, one obtains
\begin{equation}
 (m_\nu)_{\alpha \beta} = \sum_i \lambda_{\alpha i} \lambda_{i \beta} \frac{v^2}{M_i} = \sum_i \frac{\tilde \lambda_{\alpha i} \tilde \lambda_{i \beta}}{\kappa_i} \frac{v^2}{v_{B-L}}.
 \label{eq:split_9p5}
\end{equation}
Hence, apart from a very minor enhancement proportional to $\kappa_i^{-1}$, the Majorana scale contribution $\propto v_{B-L}$ suppresses the electroweak scale contribution $\propto v$. Note in particular that Eq.~\eqref{eq:split_9p5} allows for any light neutrino mass spectrum, may it be hierarchical or quasi-degenerate: the strong hierarchy in the sterile neutrino masses does \emph{not} translate into the active neutrino sector.

Note further that there is less of an enhancement of the active-sterile mixing in the split seesaw mechanism than in FN inspired settings (cf.\ Sec.~\ref{sec:keV_FN_mixed}):~\cite{Kusenko:2010ik}
\begin{equation}
 \theta_1^2 \simeq 3\cdot 10^{-9} \left( \frac{\kappa_1^{-1} \sum_\alpha |\tilde \lambda_{1 \alpha}|^2}{10^{-4}} \right) \left( \frac{10^{15}~{\rm GeV}}{v_{B-L}} \right) \left( \frac{1~{\rm keV}}{M_1} \right).
 \label{eq:split_10}
\end{equation}
Still, a moderate enhancement is present: naturally one would expect $\theta_1 \sim \frac{m_D}{M_R} \propto M_1^{-1}$, while from Eq.~\eqref{eq:split_10} we have the slightly milder relation $\theta_1 \propto M_1^{-1/2}$.

The main advantage of the split seesaw mechanism is, as already mentioned, that a small hierarchy in $m_i$ translates into a large hierarchy in $M_i$: for example, taking $m_1 \simeq 24 l^{-1}$ and $m_2 \simeq 2.3 l^{-1}$ leads to $M_1/M_2 \simeq 10^{-17}$, or $M_1 = \mathcal{O}(1~{\rm keV})$ for $M_2 = \mathcal{O}(10^{11}~{\rm GeV})$. However, one still has no explanation for where the hierarchy within the $m_i$ originates from. Furthermore, one could argue that a whole new hidden sector -- featuring a lot of new particles from the KK tower that can hardly be detected -- is introduced to explain nothing more than the mass splitting of the sterile neutrinos. While the latter issue is a general problem in model building and there is hardly any way around it, the former problem can be circumvented by extending the split seesaw mechanism by a flavor symmetry. This is exactly what we will do next.

\subsubsection{\label{sec:keV_split_A4}The $A_4$ extended split seesaw}	

Until now the only extension of the split seesaw mechanism by a flavor symmetry was presented in Ref.~\refcite{Adulpravitchai:2011rq}, and it makes again use of our ``good old friend'' $A_4$. The authors of this reference made use of the systematic classification of $A_4$ assignments presented by Barry and Rodejohann in Ref.~\refcite{Barry:2010zk}. In this classification, ten different assignment schemes, called A, B, ..., J have been identified. Out of these, some do not even contain RH neutrinos and are hence not very useful for the case at hand. Others assign the three RH neutrinos to a triplet $\mathbf{3}$ representation of $A_4$, and also these schemes are not useful since they would lead to a degenerate mass pattern, $M_R \overline{(N_R)^c} N_R = M_R [\overline{(N_{1R})^c} N_{1R} + \overline{(N_{2R})^c} N_{2R} + \overline{(N_{3R})^c} N_{3R}]$. Another assignment is ruled out by the experimental values of the mass square differences, and out of the remaining two possibilities, the simpler one has been chosen for the discussion~\cite{Adulpravitchai:2011rq}.

The resulting model uses the same assignment (scheme~E) as the original model from Ref.~\refcite{Ma:2005mw}, but in addition to the $Z_2$ parity there is an additional $Z_5$ auxiliary symmetry imposed, which is important to constrain the flavon potential. The resulting representation and charge assignments are detailed in Tab.~\ref{tab:A4TA}, where in addition to the lepton sector and to the SM-like Higgs, one triplet flavon $\varphi_{\nu, t}$ and three singlet flavons $\varphi_{i,e}$ ($i=1,2,3$) have been introduced.

\begin{table}[h]
\tbl{Particle content and representation/charge assignments (in our notation) of the $SU(2)_L\times A_4 \times Z_2 \times Z_5$ extension of the split seesaw mechanism~\cite{Adulpravitchai:2011rq}. $SU(2)_L$ is the usual SM gauge symmetry, and $Z_{2,5}$ are additional auxiliary symmetries. For $SU(2)_L$ and for $A_4$ the representations are given, while for $Z_{2,5}$ the charge is indicated. Note that $\omega = e^{2\pi i /5}$.}
{\begin{tabular}{@{}lcccccccccc@{}} \toprule
Field & $L_{1,2,3}$ & $\overline{e_{1,2,3R}}$ & $\overline{N_1}$ & $\overline{N_2}$ & $\overline{N_3}$ & $H$ & $\varphi_{\nu, t}$ & $\varphi_{1,e}$ & $\varphi_{2,e}$ & $\varphi_{3,e}$\\
\colrule
$SU(2)_L$ & $\mathbf{2}$ & $\mathbf{1}$ & $\mathbf{1}$ & $\mathbf{1}$ & $\mathbf{1}$ & $\mathbf{2}$ & $\mathbf{1}$ & $\mathbf{1}$ & $\mathbf{1}$ & $\mathbf{1}$ \\
$A_4$ & $\mathbf{3}$ & $\mathbf{3}$ & $\mathbf{1}$ & $\mathbf{1'}$ & $\mathbf{1''}$ & $\mathbf{1}$ & $\mathbf{3}$ & $\mathbf{1}$ & $\mathbf{1''}$ & $\mathbf{1'}$ \\
$Z_2$ & $+$ & $-$ & $+$ & $+$ & $+$ & $+$ & $+$ & $-$ & $-$ & $-$ \\
$Z_5$ & $\omega^3$ & $\omega^4$ & $1$ & $1$ & $1$ & $1$ & $\omega^2$ & $\omega^3$ & $\omega^3$ & $\omega^3$ \\
\botrule
\end{tabular}
\label{tab:A4TA}}
\end{table}

These assignments yield a non-diagonal RH neutrino mass matrix, so that the couplings $\kappa_i$ in Eq.~\eqref{eq:split_8} have to be promoted to matrix components $\kappa_{ij}$. Then, the Yukawa coupling matrices for the Majorana masses of the RH neutrinos and their bulk masses are given by
\begin{equation}
 \kappa=
 \begin{pmatrix} 
 a & 0 & 0 \\ 
 0 & 0 & b \\ 
 0 & b & 0 
 \end{pmatrix}\ \ \ ,\ \ \   
 m = {\rm diag}(m_1, m_2, m_3). 
 \label{eq:TA_1}
\end{equation}
With a general triplet flavon VEV $\langle \varphi_{\nu,t} \rangle=(u_1,u_2,u_3)^T$, one obtains the couplings $\tilde \lambda_{i \alpha}$ at higher order,
\begin{equation}
 \frac{1}{\Lambda} \left[y_1^{\nu} \overline{N^{(0)}_{1R}} (L \varphi_{\nu})_{\mathcal{1}}  +y_2^{\nu} \overline{N^{(0)}_{2R}} (L \varphi_{\nu})_{\mathcal{1''}} +y_3^{\nu} \overline{N^{(0)}_{3R}} (L \varphi_{\nu})_{\mathcal{1'}}  \right] \phi \rightarrow \tilde{\lambda}_{i\alpha}\bar{\Psi}^{(0)}_{iR}L_\alpha\phi,
 \label{eq:TA_2}
\end{equation}
which leads to
\begin{equation}
 \tilde{\lambda}=\frac{1}{\Lambda}
 \begin{pmatrix} 
 y_1^{\nu} & 0       & 0 \\ 
 0       & y_2^{\nu} & 0 \\ 
 0       & 0       & y_3^{\nu} 
 \end{pmatrix} 
 \begin{pmatrix} 
 1 & 1        & 1 \\ 
 1 & \omega   & \omega^2 \\ 
 1 & \omega^2 & \omega 
 \end{pmatrix} 
 \begin{pmatrix}
 u_1 & 0  & 0  \\ 
 0  & u_2 & 0  \\ 
 0  & 0  & u_3 
 \end{pmatrix}.
 \label{eq:TA_3}
\end{equation}
Inserting the SM-like Higgs VEV, one obtains the Dirac mass matrix,
\begin{equation}
 m_D^T = \frac{1}{\Lambda} {\rm diag}(C_1, C_2, C_3)
 \begin{pmatrix} 
 y_1^{\nu} & 0       & 0 \\ 
 0       & y_2^{\nu} & 0 \\ 
 0       & 0       & y_3^{\nu} 
 \end{pmatrix} 
 \begin{pmatrix}
 1 & 1        & 1  \\ 
 1 & \omega   & \omega^2 \\ 
 1 & \omega^2 & \omega 
 \end{pmatrix} 
 \begin{pmatrix} 
 u_1 & 0  & 0 \\ 
 0  & u_2 & 0 \\ 
 0  & 0  & u_3 
 \end{pmatrix} v,
 \label{eq:TA_4}
\end{equation}
where $C_i \equiv \sqrt{\frac{2 m_i}{e^{2 m_i l}-1}} \frac{1}{\sqrt{M_0}}$. The RH mass matrix, in turn, is given by
\begin{equation}
 M_R = {\rm diag}(C_1, C_2, C_3)\ \kappa\ {\rm diag}(C_1, C_2, C_3) = 
 \begin{pmatrix} 
 a C_1^2 & 0        & 0 \\ 
 0      & 0        & b C_2 C_3\\ 
 0      & b C_2 C_3 & 0 
 \end{pmatrix}v_{B-L},
 \label{eq:TA_5}
\end{equation}
which yields the eigenvalues
\begin{equation}
 M_1 = a C_1^2 v_{B-L}\ \ \ ,\ \ \ M_2 = M_3 = b C_2 C_3 v_{B-L}.
 \label{eq:TA_6}
\end{equation}
Note the degeneracy of $M_{2,3}$ in Eq.~\eqref{eq:TA_6}, which may be crucial for baryogenesis~\cite{Canetti:2012kh}. The light neutrino mass matrix can be immediately calculated using Eq.~\eqref{eq:seesawI_2},
\begin{equation}
 m_\nu \equiv m_\nu^I =  - \lambda^T \kappa^{-1} \lambda \frac{v^2}{v_{B-L}} = 
 \begin{pmatrix} 
 \frac{u_1}{\Lambda} & 0                  & 0 \\ 
 0                  & \frac{u_2}{\Lambda} & 0 \\ 
 0                  & 0                  & \frac{u_3}{\Lambda} 
 \end{pmatrix} 
 \begin{pmatrix} 
 x & y & y \\ 
 y & x & y \\ 
 y & y & x 
 \end{pmatrix} 
 \begin{pmatrix} 
 \frac{u_1}{\Lambda} & 0                  & 0 \\ 
 0                  & \frac{u_2}{\Lambda} & 0 \\ 
 0                  & 0                  & \frac{u_3}{\Lambda} 
 \end{pmatrix}\frac{v^2}{v_{B-L}},
 \label{eq:TA_7}
\end{equation}
where $x \equiv \frac{(y_1^\nu)^2}{a}+\frac{2y_2^{\nu}y_3^{\nu}}{b}$ and $y \equiv \frac{(y_1^\nu)^2}{a}-\frac{y_2^{\nu}y_3^{\nu}}{b}$. As to be expected, the small factors $C_i$ cancel out completely from the light neutrino mass matrix, Eq.~\eqref{eq:TA_7}. Assuming $u_1 = r u$ and $u_2 = u_3 = u$ in $\langle \varphi_{\nu,t} \rangle$, one obtains suitable eigenvalues for the light neutrino masses, which can be brought into agreement with the measured values of the mass square differences~\cite{Adulpravitchai:2011rq}. However, the problem is that this case leads to $\theta_{13} = 0$ (and also to $\theta_{23} = \pi /4$), which is experimentally excluded, cf.\ Eqs.~\eqref{eq:angles_NH} and~\eqref{eq:angles_IH}. Even worse, the active-sterile mixing turns out to be too large by about three orders of magnitude~\cite{Adulpravitchai:2011rq}, which completely destroys the validity of the scenario.

The way out of this dilemma is to extend the model further, by also including a seesaw type~II contribution, cf.\ Eq.~\eqref{eq:LH_mass}. That is achieved by adding a Higgs triplet field $T$, cf.\ Eq.~\eqref{eq:triplet_Higgs}, as well as two more flavon fields as illustrated in Tab.~\ref{tab:A4TA_add}. This leads to an additional additive contribution $S_\nu^{II}$ to the action, given by
\begin{equation}
 S_\nu^{II} = \int d^4 x dy \delta(y) \left( \tilde{y}^{\nu}_1 \frac{\varphi_{\nu,t} }{\Lambda} + \tilde{y}^{\nu}_2 \frac{\varphi_{\nu,s}}{\Lambda} \right) \left[ - \frac{1}{2} \overline{L} i\sigma_2 T^\dagger L^c + h.c. \right],
 \label{eq:TA_8}
\end{equation}
and the Higgs triplet VEV $v_T$ leads to an additional LH neutrino mass matrix
 \begin{equation}
 m_L =
 \begin{pmatrix} 
 s & t & t \\ 
 t & s & 0 \\ 
 t & 0& s 
 \end{pmatrix} 
 m_\nu' \ \ \ ,
 \label{eq:TA_9}
\end{equation}
where $s \equiv \tilde y^\nu_2 \frac{\langle\varphi_{\nu,s} \rangle / \Lambda}{u / \Lambda}$, $t \equiv \tilde y^\nu_1$, and $m_\nu' \equiv v_T \frac{u}{\Lambda}$. The final neutrino mass matrix is given by
\begin{equation}
 m_\nu=m_L + m_\nu^I \simeq
 \begin{pmatrix}  
 s m_\nu' & t  m_\nu'  & t  m_\nu'  \\ 
 t m_\nu' & s m_\nu' - 2 y m_\nu & y  m_\nu  \\ 
 t m_\nu' & y  m_\nu  & s m_\nu' - 2y m_\nu
 \end{pmatrix},
 \label{eq:TA_10}
\end{equation}
which still leads to $\theta_{13} = 0$ and $\theta_{23} = \pi /4$. The authors have recognized this problem, and they have argued that next-to-leading order terms could be able to cure it. However, no explicit discussion had been presented.

\begin{table}[h]
\tbl{Particle content of the type~II extension in addition to Tab.~\ref{tab:A4TA}.}
{\begin{tabular}{@{}lccc@{}} \toprule
Field & $T$        & $\varphi_{\nu,s}$     & $\tilde{\varphi}_{\nu,t}$ \\
\colrule
$SU(2)_L$ & $\mathbf{3}$ & $\mathbf{1}$ & $\mathbf{1}$ \\
$A_4$ & $\mathbf{1}$ & $\mathbf{1}$ & $\mathbf{3}$ \\
$Z_2$ & $+$ & $+$ & $+$ \\
$Z_5$ & $\omega^2$ & $\omega^2$ & $1$ \\
\botrule
\end{tabular}
\label{tab:A4TA_add}}
\end{table}

Putting this problem aside, the seesaw type~II extension can provide realistic expressions for the mass square differences and for $\theta_{12}$, which in the limit $m_\nu' \simeq m_\nu$ look like:
\begin{eqnarray}
 \Delta m_{21}^2    &\simeq& (2 s-y) \sqrt{8 t^2+y^2} m_\nu^2, \\
 \Delta m_{31}^2    &\simeq& \left[(s-3y)^2-\frac{\left(2s-y-\sqrt{8t^2+y^2} \right)^2}{4}\right]m_\nu^2, \\
 \tan^2 \theta_{12} &\simeq& \frac{\left(y+\sqrt{8 t^2+y^2}\right)^2}{8t^2}.
 \label{eq:TA_11}
\end{eqnarray}
Note that, assuming the parameters to be complex in general, normal and inverted ordering are both possible in this model. Furthermore, this extension overcomes the problem with the X-ray bound: the LH mass matrix, Eq.~\eqref{eq:TA_9}, does \emph{not} contribute to active-sterile mixing, cf.\ \ref{sec:seesaw_diag}. Hence, one can choose the problematic elements in $m_D$, Eq.~\eqref{eq:TA_4}, to be small enough while still keeping a realistic light neutrino spectrum alive by the type~II contributions.

We end this section by pointing out that the VEV alignment used in this model can actually be derived from orbifolding, as described in Ref.~\refcite{Adulpravitchai:2011rq}. Furthermore, the reference discussed also various possibilities for leptogenesis within the model. These discussions, however, are beyond the scope of this review.

\subsubsection{\label{sec:keV_split_separate}Separate seesaw}	

We finally want to comment on a variant of the split seesaw mechanism which was discussed by Takahashi in Ref.~\refcite{Takahashi:2013eva} and called \emph{separate seesaw}. As we had already discussed in Sec.~\ref{sec:keV_split_mechamism}, the basic idea of split seesaw is to localize both the Yukawa couplings as well as the RH Majorana neutrino masses on the SM brane, cf.\ Eq.~\eqref{eq:split_8}, while the RH neutrino wave functions exponentially localize at the hidden brane where their values are largest. The idea introduced in Ref.~\refcite{Takahashi:2013eva} is to alter these assignments, and to try to use alternative combinations of the localizations, \emph{e.g.}, having both the Majorana mass matrix and the wave function localized at the hidden brane.

Going through all the possibilities, it unfortunately turns out that no other assignment than the one chosen for the split seesaw can do the job of simultaneously suppressing the lightest RH neutrino mass and the Yukawa couplings. In other words, either we do not even have a motivation for a keV neutrino or we run into problems with the X-ray bound, cf.\ Sec.~\ref{sec:keV_general_X}, if not both at the same time. The way out is to apply different localizations at the same time. In the \emph{separate seesaw}, the one already introduced in the split seesaw is used for the first generation sterile neutrino, while alternative localizations are used for the second and third generations.

The main benefit of the separate seesaw is that one gains some more flexibility inside the light and heavy neutrino mass matrices, at least for the second and third generations. In particular, a hierarchical pattern as the one likely to be observed in the light neutrino sector might be a bit easier to realize. However, the price to pay is having to treat different generations of neutrinos in qualitatively different ways. This contradicts the spirit of unification, but it might be necessary in order to reproduce certain more elaborate mass spectra.

\subsection{\label{sec:keV_ext}Models based on the extended seesaw mechanism}	

The next mechanism which we want to discuss here is the so-called \emph{extended seesaw mechanism}. This mechanism was first proposed in Ref.~\refcite{Chun:1995js} in order to solve the solar neutrino problem, and it was later on revived in Refs.~\refcite{Barry:2011wb,Zhang:2011vh} for the purpose of explaining keV sterile neutrinos. While this mechanism goes beyond the ``minimal'' picture of having only three RH neutrinos, it nevertheless yields a relatively natural explanation for the existence of keV sterile neutrinos, at least when applied to concrete models.

\subsubsection{\label{sec:keV_ext_mechamism}The extended seesaw mechanism}	

Let us start by discussing the mechanism itself. The basic idea is to introduce, in addition to the three RH neutrinos $N_{Ri}$, another chiral singlet field $S_R$ which is also right-handed.\footnote{Note that we depart from the notation used in Refs.~\refcite{Barry:2011wb,Zhang:2011vh}, since in these references it is not so transparent that the singlet field called $S$ must actually be right-handed. If it was not, its left-handed counterpart $S_L$ would exist, which could then be used to form an undesirable Dirac mass term for the singlet, $\mathcal{L}_{S,\rm Dirac} = - \overline{S_L} m_{D,S} S_R + h.c.$ To make clear that such a term does \emph{not} appear, we decided to change the notation, $S \to S_R$.} The next step is to \emph{postulate} the following Lagrangian,
\begin{equation}
 \mathcal{L}_{\rm ES} = -\overline{\nu_L} m_D N_R  - \overline{(S_R)^c} M_S^T N_R - \frac{1}{2} \overline{(N_R)^c} M_R N_R +  h.c.
 \label{eq:ES_1}
\end{equation}
Up to this point, the Lagrangian presented in Eq.~\eqref{eq:ES_1} is only a postulate. The reason for this statement is that the Lagrangian is \emph{incomplete}: gauge invariance would not forbid a Majorana mass term for $S_R$ alone, $\mathcal{L}_{S,\rm Majorana} = - \frac{1}{2} \overline{(S_R)^c} M_{M,S} S_R + h.c.$, but this term does not appear in the above Lagrangian. However, the problem can be cured easily, by imagining that the term $\overline{(S_R)^c} M_S^T N_R$ could arise from the VEV $f = \langle \phi_S \rangle$ of some singlet scalar field $\phi_S$, $y_{SN} \overline{(S_R)^c} \phi_S N_R \to \overline{(S_R)^c} ( y_{SN} f) N_R$. If one used a discrete symmetry, \emph{e.g.} $Z_4$, one could give a charge of $\pm i$ ($\mp i$) to $S_R$ ($\phi_S$), in which case the mixed mass term $\overline{(S_R)^c} M_S N_R$ in Eq.~\eqref{eq:ES_1} would be allowed, while a Majorana mass term for $S_R$ alone would be forbidden, at least at tree-level. Of course, other solutions of this problem are possible.

Let us press on and take Eq.~\eqref{eq:ES_1} for granted. Then, $M_S$ is a $1\times 3$ matrix, since there is only one $S_R$ versus three $N_{Ri}$. Hence, one can rewrite Eq.~\eqref{eq:ES_1} in terms of a $7\times 7$ mass matrix $M_\nu^{7 \times 7}$, similar to rewriting the seesaw matrix in \ref{sec:seesaw_matrix}:
\begin{equation}
 \mathcal{L}_{\rm ES} = - \frac{1}{2} (\overline{\nu_L}, \overline{(S_R)^c}, \overline{(N_R)^c})
 \begin{pmatrix}
 0 & m_D & 0 \\
 m_D^T & M_R & M_S^T \\
 0 & M_S & 0
 \end{pmatrix}
 \begin{pmatrix}
 (\nu_L)^c\\
 S_R\\
 N_R
 \end{pmatrix}
 +  h.c.
 \label{eq:ES_2}
\end{equation}
The important point of this Lagrangian is that there is essentially no constraint on the size of $M_S$. In particular, there is no need for $M_S$ to be as large as the size one would naturally anticipate for $M_R$, since the different masses do not have the same origin (as explained above, $M_R$ is a bare Majorana mass while $M_S$ is part of a mixing term). On the other hand, there is also no reason to have $M_S \lesssim m_D = \mathcal{O}(v)$, which would typically lead to problems with neutrino oscillations~\cite{Zhang:2011vh}. Hence, the most attractive possibility is to assume the following mass hierarchy:
\begin{equation}
 m_D \ll M_S \ll M_R. 
 \label{eq:ES_3}
\end{equation}
In this case, we will see that the inverse seesaw can indeed motivate a keV scale.

The first step is to \emph{integrate out} the RH neutrinos from the Lagrangian in Eq.~\eqref{eq:ES_1}. In theory, this means that the RH neutrino fields are assumed to be very heavy, and hence not dynamical at low energies. Then, one can neglect their kinetic terms and solve their Euler-Lagrange equations explicitly, just as done for the heavy fermions in the FN mechanism, cf.\ Sec.~\ref{sec:neutrino_modeling_FN}. In practice, this amounts to a seesaw-like formula, and the remaining ``low'' energy $4\times 4$ matrix is given by
\begin{equation}
 M_\nu^{4\times 4} = 
 \begin{pmatrix}  
 m_D  M_R^{-1}  m_D^T & {\ } & m_D  M_R^{-1}  M_S^T  \\ 
 M_S  \left(M_R^{-1}\right)^T  m_D^T & {\ } & M_S M_R^{-1} M_S^T
\end{pmatrix} .
 \label{eq:ES_4}
\end{equation}
The determinant of this matrix can be easily computed by the formula ${\rm det} \begin{pmatrix} A & B\\ C & D \end{pmatrix} = {\rm det}(A) {\rm det}\left( D - C A^{-1} B \right)$, which holds for block matrices with an invertible submatrix $A$. With $\left(M_R^{-1}\right)^T = M_R^{-1}$, one obtains:
\begin{eqnarray}
 && {\rm det} \left( M_\nu^{4\times 4} \right) = {\rm det} \left( m_D  M_R^{-1}  m_D^T \right) \nonumber\\
 && \times {\rm det} \left( M_S M_R^{-1} M_S^T - M_S  M_R^{-1} m_D^T \left( m_D^T \right)^{-1} M_R m_D^{-1} m_D  M_R^{-1}  M_S^T \right) = 0.
 \label{eq:ES_5}
\end{eqnarray}
The determinant of a matrix is the product of its eigenvalues, so the minimal extended seesaw mechanism predicts one light neutrino to be exactly massless.

In the case of $M_S \gg m_D$, one can now apply the seesaw type~II formula, cf.\ Eq.~\eqref{eq:seesaw_general}, to Eq.~\eqref{eq:ES_4}:
\begin{equation}
 m_\nu \equiv M_\nu^{3\times 3} = m_D  M_R^{-1}  M_S^T \left(M_S  M_R^{-1} M_S^T \right)^{-1} M_S  M_R^{-1}  m_D^T - m_D M_R^{-1} m_D^T.
 \label{eq:ES_6}
\end{equation}
Note that the right-hand side of Eq.~\eqref{eq:ES_6} does \emph{not} vanish~\cite{Zhang:2011vh}. The reason is that $m_s^{-1} \equiv \left(M_S  M_R^{-1} M_S^T \right)^{-1}$ is a number instead of a matrix, due to the structure of $M_S$. Hence, one cannot rewrite it as $\left(M_S^T \right)^{-1} M_R M_S^{-1}$. In short, $M_S$ as a $1\times 3$ matrix is not invertible.

The above block-diagonalization of  $M_\nu^{4\times 4}$ yields a singlet mass eigenvalue,
\begin{equation}
 m_s = M_S  M_R^{-1} M_S^T.
 \label{eq:ES_7}
\end{equation}
This equation reveals the whole trick of the extended seesaw mechanism: although the matrix product in Eq.~\eqref{eq:ES_7} looks like a seesaw type~I formula, \emph{i.e.}, $m_s \ll |M_{S,i}|$, we had chosen $M_S \gg m_D$ which enforces $m_s \gg |(m_\nu)_{ij}|$. Due to these hierarchies, we have obtained a motivation for $m_s = \mathcal{O}({\rm keV})$ while $m_\nu \lesssim \mathcal{O}({\rm eV})$. Due to the suppression of the new scale $M_S$ in Eq.~\eqref{eq:ES_7}, the extended seesaw mechanism could be viewed as a variant of the top-down mass shifting scheme, cf.\ right panel of Fig.~\ref{fig:keV-schemes}. However, the difference here is that a completely new mass scale is suppressed, so that one actually has four sterile neutrinos ($N_{1,2,3}$ and $S$) instead of three, of which the natural mass scale of one ($S$) is suppressed.

Although the extended seesaw mechanism in principle yields an explanation for the appearance of a scale that could be potentially of $\mathcal{O}({\rm keV})$, we have already remarked after Eq.~\eqref{eq:ES_1} that the required Lagrangian is a bit unmotivated in general. A suitable way out is to enforce the structure given in Eq.~\eqref{eq:ES_1} by a flavor symmetry and, not too surprisingly, the simplest example that is known in the literature is based on an $A_4$ symmetry~\cite{Zhang:2011vh}.

\subsubsection{\label{sec:keV_ext_A4}The $A_4$ extended extended seesaw}	

The $A_4$ extension of the extended seesaw mechanism~\cite{Zhang:2011vh} can resolve some of the problems mentioned in the previous section. Again the task is to find suitable charge assignments, which are summarized in Tab.~\ref{tab:A4HZ}. Note that this table also contains the flavon fields $\varphi$, $\varphi'$, $\varphi''$, $\xi$, $\xi'$, and $\chi$, while $S_R$ is the additional singlet fermion. Using these alignments, the leading order Lagrangian of the leptonic sector is given by,\footnote{Note that certain dangerous terms are actually neglected in Ref.~\refcite{Zhang:2011vh}. It is argued that these terms could be forbidden by an additional global $U(1)$ symmetry, under which $\chi$ and $S$ are suitably charged. However, since $\chi$ will obtain a VEV later on, this will break the symmetry which will unavoidably lead to a massless Goldstone boson~\cite{Goldstone:1961eq,Goldstone:1962es}. While this problem is not solved in Ref.~\refcite{Zhang:2011vh}, a generalization of the model (which is shortly discussed at the end of this section) was later on given by the same author in Ref.~\refcite{Heeck:2012bz}, where the symmetry is gauged and the Goldstone boson is eaten. Alternatively, one could give a non-perturbative mass to the otherwise massless state, making it a \emph{pseudo Nambu-Goldstone boson}.}
\begin{eqnarray}
 \mathcal{L} &=& \frac{y_e}{\Lambda} \overline{e_R} \left(H L \varphi^* \right)_{\mathbf{1}} + \frac{y_\mu}{\Lambda} \overline{\mu_R} \left(H L \varphi^* \right)_{\mathbf{1'}} + \frac{y_\tau}{\Lambda} \overline{\tau_R} \left(H L \varphi^* \right)_{\mathbf{1}''} + \frac{1}{2} \rho \chi \overline{(S_R)^c} N_{R1}  \nonumber \\
 && + \frac{y_1}{\Lambda} \overline{N_{1R}} \left( \tilde H L \varphi^* \right)_{\mathcal{1}} + \frac{y_2}{\Lambda} \overline{N_{2R}} \left( \tilde H L \varphi' \right)_{\mathbf{1''}} + \frac{y_3}{\Lambda} \overline{N_{3R}} \left( \tilde H L \varphi'' \right)_{\mathbf{1}} \nonumber \\ 
 && + \frac{1}{2} \lambda_1 \xi \overline{(N_{R1})^c} N_{R1} + \frac{1}{2} \lambda_2 \xi' \overline{(N_{R2})^c} N_{R2} + \frac{1}{2} \lambda_3 \xi \overline{(N_{R3})^c} N_{R3} + h.c.,
 \label{eq:ES_A4_1}
\end{eqnarray}
where $\Lambda$ denotes the cutoff scale. Note that we have again reported the Hermitean conjugate part of the Lagrangian presented in Ref.~\refcite{Zhang:2011vh}, in order to conform with our notation. However, except for $\lambda_i$ and $\rho$, all couplings and VEVs can be taken to be real by absorbing the phases into the fermion fields, and some phases will be forced to be equal by the VEV alignment. By radiative breaking, one can obtain the following alignment~\cite{King:2006np}:
\begin{equation}
 \langle \varphi \rangle =(v_\varphi, 0, 0)\,\ \ \langle \varphi' \rangle = (v_\varphi, v_\varphi, v_\varphi), \ \ \langle \varphi'' \rangle = (0, -v_\varphi, v_\varphi), \ \ \langle \xi \rangle = \langle \xi' \rangle = v_\varphi, \ \  \langle \chi \rangle = u.
 \label{eq:ES_A4_2}
\end{equation}
This leads immediately to a diagonal charged lepton mass matrix, $M_e = \frac{v v_\varphi}{\Lambda} {\rm diag} (y_e, y_\mu, y_\tau)$, which means that $v_\varphi y_i$ must be real for $i = e, \mu, \tau$. Note that the hierarchy in the charged lepton sector is \emph{not} explained by this model.

\begin{table}[ph]
\tbl{Particle content and representation/charge assignments (in our notation) of the $SU(2)_L\times A_4 \times Z_4$ extension of the extended seesaw mechanism~\cite{Zhang:2011vh}. $SU(2)_L$ is the usual SM gauge symmetry, and $Z_4$ is an additional auxiliary symmetry. For $SU(2)_L$ and for $A_4$ the representations are given, while for $Z_4$ the charge is indicated.}
{\begin{tabular}{@{}lcccccccccccccccc@{}} \toprule
Field & $L_{1,2,3}$ & $\overline{e_R}$ & $\overline{\mu_R}$ & $\overline{\tau_R}$ & $\overline{N_1}$ & $\overline{N_2}$ & $\overline{N_3}$ & $S_R$ & $H$ & $\varphi$ & $\varphi'$ & $\varphi''$ & $\xi$ & $\xi'$ & $\chi$ \\
\colrule
$SU(2)_L$ & $\mathbf{2}$ & $\mathbf{1}$ & $\mathbf{1}$ & $\mathbf{1}$ & $\mathbf{1}$ & $\mathbf{1}$ & $\mathbf{1}$ & $\mathbf{1}$ & $\mathbf{2}$ & $\mathbf{1}$ & $\mathbf{1}$ & $\mathbf{1}$ & $\mathbf{1}$ & $\mathbf{1}$ & $\mathbf{1}$ \\
$A_4$ & $\mathbf{3}$ & $\mathbf{1}$ & $\mathbf{1''}$ & $\mathbf{1'}$ & $\mathbf{1}$ & $\mathbf{1'}$ & $\mathbf{1}$ & $\mathbf{1}$ & $\mathbf{1}$ & $\mathbf{3}$ & $\mathbf{3}$ & $\mathbf{3}$ & $\mathbf{1}$ & $\mathbf{1'}$ & $\mathbf{1}$ \\
$Z_4$ & $1$ & $1$ & $1$ & $1$ & $1$ & $i$ & $-1$ & $i$ & $1$ & $1$ & $i$ & $-1$ & $1$ & $-1$ & $-i$ \\
\botrule
\end{tabular}
\label{tab:A4HZ}}
\end{table}

Going to neutrinos, the mass matrices can be read off from Eq.~\eqref{eq:ES_A4_1} as
\begin{equation}
 m_D = \frac{v v_\varphi}{\Lambda}
 \begin{pmatrix}
 y_1 & y_2 & 0 \\
 0 & y_2 & y_3 \\
 0 & y_2 & -y_3
 \end{pmatrix}, \ \ M_R = {\rm diag} (\lambda_1 v_\varphi, \lambda_2 v_\varphi, \lambda_3 v_\varphi), \ \ M_S = (\rho u, 0, 0),
 \label{eq:ES_A4_3a}
\end{equation}
which yields $m_s = |\rho| u$ as the mass of the keV sterile neutrino. The light neutrino mass matrix is given by
\begin{equation}
 m_\nu = - \frac{v^2 v_\varphi^2}{\Lambda}
 \begin{pmatrix}
 \frac{y^2_2}{\lambda_2} & \frac{y^2_2}{\lambda_2} & \frac{y^2_2}{\lambda_2} \\
 \frac{y^2_2}{\lambda_2} & \frac{y^2_2 \lambda_3 + y^2_3 \lambda_2}{\lambda_2 \lambda_3} & \frac{y^2_2 \lambda_3 -y^2_3 \lambda_2}{\lambda_2 \lambda_3} \\
 \frac{y^2_2}{\lambda_2} & \frac{y^2_2 \lambda_3 - y^2_3 \lambda_2}{\lambda_2 \lambda_3} & \frac{y^2_2 \lambda_3 + y^2_3 \lambda_2}{\lambda_2 \lambda_3}
 \end{pmatrix},
 \label{eq:ES_A4_3b}
\end{equation}
which leads to the light neutrino masses $(m_1, m_2, m_3) = \left( 0, \frac{3 y_2^2 v^2 v_\varphi}{|\lambda_2| \Lambda^2}, \frac{3 y_3^2 v^2 v_\varphi}{|\lambda_3| \Lambda^2} \right)$. These eigenvalues clearly obey normal ordering. However, the matrix diagonalizing $m_\nu$ from Eq.~\eqref{eq:ES_A4_3b} is tri-bimaximal, cf.\ Eq.~\eqref{eq:TBM}, which must equal the PMNS matrix due to the charged lepton mass matrix being diagonal. Hence, the model presented is excluded by the new data on $\theta_{13}$, cf.\ Eqs.~\eqref{eq:angles_NH}. It is claimed in Ref.~\refcite{Zhang:2011vh} that this problem could be cured by next-to-leading order terms, but up to now there exists no explicit proof for this claim.

Nevertheless, other models based on the extended seesaw mechanism could potentially be found, since there is no principle argument why the extended seesaw should be inconsistent with current data. For example, a very recent proposal~\cite{Heeck:2012bz} uses anomaly free settings of a $U(1)'$ symmetry instead of $A_4$. This has the advantage to be able to accommodate for the new data, but for the cost of introducing some more singlet fields like $S_R$ in order to ensure the cancellation of all anomalies. On the other hand, the model has the interesting feature of intrinsically predicting a mixture of light and heavy DM. Furthermore, the additional singlets lead to a modification of the light neutrino mass matrix such that there exists no massless state anymore. Hence, if it turned out that even the lightest neutrino mass is significantly different from zero, the $U(1)'$-extension of the extended seesaw mechanism would still have no problems, while its minimal version would be under pressure.

\subsection{\label{sec:keV_comp}More complicated seesaw and suppression scenarios}	

Finally, we want to shortly comment on more complicated models which, to some extent, rely on a seesaw or similar type of suppression mechanism. These models are, from a purely scientific point of view, not inferior to the ones already presented. However, for one reason or the other they are slightly more complicated and/or their details go beyond the scope of this (hopefully) pedagogical review. Nevertheless, these models offer interesting alternatives to motivate a hierarchy in the sterile neutrino sector. The reader interested in the finer details is referred to the references given in the following.

\subsubsection{\label{sec:keV_comp_331}Type~II seesaw mechanism in a 331-model}	

This model has been presented in Ref.~\refcite{Cogollo:2009yi}, based on earlier considerations in Ref.~\refcite{Dias:2005yh}. It is a slight extension of the ordinary 331-model~\cite{Pisano:1991ee,Frampton:1992wt}, which extends the SM gauge group to a more complicated group $SU(3)_C \times SU(3)_L \times U(1)_N$. This symmetry is broken at a high scale to the SM gauge group. According to elementary group theory~\cite{Ramond}, an $SU(N)$ group has $N^2-1$ generators, while a $U(1)$ group has only one. Hence, in addition to the eight gluons from $SU(3)_C$ which are unchanged in this model, we expect $3^2 - 1 + 1 = 9$ gauge bosons in total, four out of which are the ordinary electroweak gauge bosons while the other five are new: $V^{\pm}$, $U^0$, ${U^0}^\dagger$, and $Z'$.

The SM leptons come in $SU(3)_L$ triplet and singlet representations,
\begin{equation}
 f_{i L} =
 \begin{pmatrix}
 \nu_{i L} \\
 e_{i L} \\
 (N_{i R})^c
 \end{pmatrix} \sim \left( \mathbf{1}, \mathbf{3}, -\frac{1}{3} \right)\ \ \ , \ \ \ e_{i R} \sim \left( \mathbf{1}, \mathbf{1}, -1 \right)\ \ \ ,
 \label{eq:331_1}
\end{equation}
which means that the fields $(N_{i R})^c$, which we will later on call \emph{sterile} neutrinos, are indeed only sterile under the SM gauge group but not under the full gauge symmetry.

Along with the leptons, we have three triplet scalars and one sextet scalar,
\begin{eqnarray}
 && \chi =
 \begin{pmatrix}
 \chi^0 \\
 \chi^- \\
 {\chi'}^0
 \end{pmatrix} \sim \left( \mathbf{1}, \mathbf{3}, -\frac{1}{3} \right)\ \ \ , \ \ \ \rho =
 \begin{pmatrix}
 \rho^+ \\
 \rho^0 \\
 {\rho'}^+
 \end{pmatrix} \sim \left( \mathbf{1}, \mathbf{3}, -\frac{2}{3} \right)\ \ \ , \label{eq:331_2} \\
 && \eta =
 \begin{pmatrix}
 \eta^0 \\
 \eta^- \\
 {\eta'}^0
 \end{pmatrix} \sim \left( \mathbf{1}, \mathbf{3}, -\frac{1}{3} \right)\ \ \ , \ \ \ S = \frac{1}{\sqrt{2}}
 \begin{pmatrix}
 \Delta^0 & \Delta^- & \Phi^0 \\
 \Delta^- & \Delta^{--} & \Phi^- \\
 \Phi^0 & \Phi^- & \sigma^0 
 \end{pmatrix} \sim \left( \mathbf{1}, \mathbf{6}, -\frac{2}{3} \right)\ \ \ , \nonumber
\end{eqnarray}
where the $\Phi^i$ contain SM-like Higgs components, but they would give a Dirac mass to neutrinos. Other $SU(2)_L$ doublets, not necessarily listed here, are still available to give masses to the other SM fermions.

Since the authors of Ref.~\refcite{Cogollo:2009yi} are after a LH seesaw type~II mass, cf.\ Eq.~\eqref{eq:LH_mass}, they forbid a VEV for $\Phi^0$ by imposing suitable discrete symmetries. They furthermore tune their scalar potential such that the components $( {\chi'}^0, \rho^0, \eta^0, \Delta^0, \sigma^0)$ obtain VEVs $(v_{\chi'}, v_\rho, v_\eta, v_\Delta, v_\sigma)$. The key point is that some of the scalar components do carry lepton number, $L({\eta'}^0, \sigma^0, {\rho'}^+) = -2$ and $L(\chi^0, \chi^-, \Delta^0, \Delta^-, \Delta^{--}) = +2$. Hence, the VEV $v_\Delta$ ($v_\sigma$) unavoidably breaks lepton number and leads to a mass term for LH (RH) neutrinos. This fact can be exploited by adding a term to the scalar potential which explicitly breaks lepton number,
\begin{equation}
 V_{\rm LNV} = - M_1 \eta^T S^\dagger \eta - M_2 \chi^T S^\dagger \chi + {\rm (other\ terms)},
 \label{eq:331_3}
\end{equation}
and by furthermore assuming that the corresponding mass scale $M \approx M_{1,2}$ is larger than all VEVs, $v_{\chi', \rho, \eta, \Delta, \sigma} \ll M$. This might sound a little surprising since some of the VEVs do break lepton number, but in reality one cannot even speak of lepton number anymore after it is broken by the explicit violation terms in Eq.~\eqref{eq:331_3}.

Moreover, a strong hierarchy between certain VEVs is assumed, $v_{\Delta, \sigma} \ll v_{\chi', \rho, \eta}$. This is only an assumption, but it could be motivated by observing that $(\Delta, \sigma)$ are components of the sextet scalar field, while $(\chi', \rho, \eta)$ are contained in triplets under $SU(3)_L$, cf.\ Eq.~\eqref{eq:331_2}. This is enforced by the minimization conditions that yield, under the above assumptions,
\begin{equation}
 m_L = Y \frac{v_\eta^2}{M}\ \ \ ,\ \ \ M_R = Y \frac{v_{\chi'}^2}{M},
 \label{eq:331_4}
\end{equation}
which motivates $m_L \ll M_R$. Note that the Yukawa matrices $Y$ are the same for the LH and RH fields.

This is the seesaw-like mechanism we had mentioned before: both mass matrices in Eq.~\eqref{eq:331_4} are suppressed, and in particular we have a motivation for relatively small $M_R$ eigenvalues which could well be at the keV scale. However, a drawback of this model is that it does not explain any \emph{hierarchy} in the sterile neutrino sector. If the goal of model building is to produce settings which resemble the $\nu$MSM~\cite{Asaka:2005an}, then this model unfortunately fails to do that job unless even more symmetries are imposed.

Finally note that a variant of this model was presented in Ref.~\refcite{Dias:2010vt}, where two new $SU(3)$ symmetries, $SU(3)_L$ and $SU(3)_R$, appear. In this ``3331''-model, each of the two $SU(3)$ yields a lepton triplet, which leads to a left-right symmetric version of Eq.~\eqref{eq:331_1}. Without going through the details of this extended model, we nevertheless want to mention that this model yields \emph{three} keV neutrinos $N_L$ which could play the role of DM and have the remarkable property that they do not mix with the ordinary neutrinos, thereby invalidating the X-ray constraint, cf.\ Sec.~\ref{sec:keV_general_X}.

\subsubsection{\label{sec:keV_comp_Composite}Composite Dirac neutrinos}	

The other seesaw-type mechanism we want to briefly discuss arises in the context of \emph{composite neutrinos}. Such settings are based on so-called \emph{preons}~\cite{Pati:1974yy}, which are assumed to be fundamental fermions $q$ charged under some new hidden gauge group which is often called \emph{$\nu$-color}. Preons condensate in a QCD-like manner to bound states of a generic mass scale $\Lambda$. The key point is to assume the existence of an even more fundamental scale $M \gg \Lambda$, below which everything we talk about is an \emph{effective field theory}. The idea of such a scenario to explain the possible existence of keV sterile neutrinos was put forward in Refs.~\refcite{Grossman:2010iq,Robinson:2012wu}.

All Yukawa couplings of preons to SM-fields are effective vertices, \emph{e.g.},
\begin{equation}
 \mathcal{L}_{\rm Yukawa} = -\frac{\lambda}{M^{3(n-1)/2}} \overline{L} \tilde H q^n,
 \label{eq:composite_1}
\end{equation}
where $\lambda = \mathcal{O}(1)$ and $n$ is the number of preons $q$. Certain fields could also be charged under some hidden flavor symmetry, which could give a certain structure to terms like the one in Eq.~\eqref{eq:composite_1}. Note that, since the $n$ preons obviously have to combine to a fermion, and we also need $n > 1$ to have a bound state in the first place, the smallest value possible is $n \geq 3$. If the preons condense at some scale $\Lambda$ to a spin-$\frac{1}{2}$ bound state, $q^n \to N_R \Lambda^{3(n-1)/2}$, then the corresponding Dirac mass terms are suppressed by powers of the small quantity $\epsilon = \frac{\Lambda}{M}$,
\begin{equation}
 \mathcal{L}_{\rm Yukawa} \to -\lambda \epsilon^{3(n-1)/2}  \overline{L} \tilde H N_R + h.c.,
 \label{eq:composite_2}
\end{equation}
where $\frac{3 (n-1)}{2} \geq 3$. Hence, there will be at least an $\epsilon^3$ mass suppression from such terms, compared to their natural scale $v$ from the Higgs VEV. However, if also left-handed (supposed to be heavy) neutrinos $N_L$ exist, which are also \emph{singlets} under $SU(2)$ -- \emph{i.e.}, the fermions $N$ are \emph{vector-like} -- then other terms of the form
\begin{equation}
 \mathcal{L}_{LH} = -\Lambda_N \overline{N_L} N_R + h.c.
 \label{eq:composite_3}
\end{equation}
are possible, too. Here, $\Lambda_N = \mathcal{O}(\Lambda)$, which would lead to heavier masses~\cite{Grossman:2010iq,Robinson:2012wu}. As an intermediate type of mass term, however, the existence of $N_L$ would also lead to masses of the form
\begin{equation}
 \mathcal{L}_{LH'} = -\Lambda'_N \epsilon^{3(n'-1)/2} \overline{N_L} N_R + h.c.,
 \label{eq:composite_4}
\end{equation}
where $\Lambda'_N = \mathcal{O}(\Lambda)$. This would motivate different types of models for the \emph{sterile} neutrinos $N$, and in particular it could introduce a hierarchy such that some of the sterile states have masses of $\mathcal{O}({\rm keV})$.

The remarkable unique point about this model is that the sterile neutrinos are \emph{Dirac particles}, contrary to all other models presented here: due to the existence of the $SU(2)$ singlet fields $N_L$, one can write down mass terms as in Eqs.~\eqref{eq:composite_3} and~\eqref{eq:composite_4} for sterile neutrinos, but they will nevertheless be sterile under any SM interaction. Furthermore, the mass term in Eq.~\eqref{eq:composite_2} simultaneously offers an explanation for the active neutrino masses to be tiny, and the hidden flavor symmetry could even be used to obtain a suitable leptonic mixing structure.

While this model has several nice features, one could object that the new scales introduced are a bit arbitrary and not fundamental, since they essentially rely on some non-perturbative QCD-like mechanism. Furthermore, the whole mechanism is relatively complicated, which might be another negative point. The lesson to learn is that, as always in model building, whenever we want to explain something we have to pay a price.

\subsection{\label{sec:keV_oth}Scenarios and models based on other ideas}	

Apart from the models discussed in Secs.~\ref{sec:keV_FN} to~\ref{sec:keV_comp}, there are several other settings which can accommodate for keV sterile neutrinos. However, as we had argued at the beginning of Sec.~\ref{sec:keV}, the mere fact that a certain setting can involve keV sterile neutrinos, \emph{i.e.}, it does not disagree with their existence, does not \emph{a priori} yield any explanation for them. Hence, in the terminology used in this review, such settings would be called \emph{scenarios} rather than \emph{models}.

Let us nevertheless briefly mention a couple of examples from the literature which can accommodate for keV sterile neutrinos:

\begin{itemize}

\item keV sterile neutrinos in the \emph{scotogenic model}~\cite{Sierra:2008wj,Gelmini:2009xd,Ma:2012if}:

The scotogenic model is a very simple extension of the SM by three RH neutrinos $N_i$, one additional Higgs doublet $\eta$, and a $Z_2$ symmetry under which only the new fields are charged non-trivially, \emph{i.e.}, they receive a factor of $(-1)$ under this parity transformation~\cite{Ma:2006km}. In this setting, there is a very natural explanation for the scale of the \emph{light} (active) neutrinos, since their masses only arise at 1-loop level. This leads to a natural suppression of their mass. However, this mechanism to generate active neutrino masses intrinsically requires a tree-level Majorana mass for the RH neutrinos $N_i$, since the model does not include any mechanism to spontaneously generate any violation of lepton number which must be present in order for it to translate into the light neutrino sector.

On the other hand, the mass spectrum of the RH neutrinos is not fixed. While they could lead to very interesting phenomenology at LHC if their masses were of $\mathcal{O}({\rm TeV})$~\cite{Sierra:2008wj,Cao:2007rm,Atwood:2007zza,Haba:2011nb}, it is not excluded that one of the ``heavy'' neutrinos could be as light as $\mathcal{O}({\rm keV})$, which is exactly what is considered by Refs.~\refcite{Sierra:2008wj,Gelmini:2009xd,Ma:2012if}. While these references discuss very interesting aspects of the phenomenology of such settings, they nevertheless lack an explanation for the appearance of the keV scale.

\item keV sterile neutrinos with \emph{Left-Right symmetry}~\cite{Bezrukov:2009th,Nemevsek:2012cd}:

A very interesting mechanism in order to correct a too large abundance is the release of a sufficient amount of entropy by late decay of a non-relativistic species that temporarily dominates the energy density of the Universe~\cite{Scherrer:1984fd}. This mechanism has also been discussed in the context of the $\nu$MSM, where it is able to open up new windows in the parameter space~\cite{Asaka:2006ek,Asaka:2006nq}.

If a keV ``sterile'' neutrino is only sterile with respect to SM interactions but takes part in gauge interactions beyond the SM, it will naturally be produced thermally~\cite{KolbTurner}. However, just as if active neutrinos had a mass of $\mathcal{O}({\rm keV})$, such a production will \emph{overclose} the Universe, \emph{i.e.}, the amount of DM produced will be too big. This can be corrected by a sufficient entropy production from the decays of the two heavier sterile neutrinos, as discussed in Sec.~\ref{sec:astro_prod_Therm}.

As before, even though this setting can contain keV sterile neutrinos, it does not by itself lead to an explanation of their mass. Nevertheless it is interesting, in particular since it is not consistent with every way to explain the keV scale, cf.\ Sec.~\ref{sec:keV_FN_pure}.

\end{itemize}

There are also further examples in the literature which would be called \emph{models} according to the terminology used here, but which for one reason or the other do not comprise a setting that is suitable for an illustrative example. These include:

\begin{itemize}

\item $U(1)$ symmetries broken close to the Planck scale~\cite{Allison:2012qn}:

In Ref.~\refcite{Allison:2012qn}, certain structures that can also lead to keV sterile neutrinos are discussed. Concretely, a model is presented which is based on a $U(1) \times U(1)' \times Z_4 \times Z_2$ symmetry. It connects both the production of the keV sterile neutrinos and the generation of their mass to a relatively light inflaton field $\phi$. In particular, the continuous symmetries are broken at a scale so high that it could even be close to the Planck scale.

However, one might object that, in fact, the symmetry used in Ref.~\refcite{Allison:2012qn} is actually not too much different from a generic Froggatt-Nielsen $U(1)$. Furthermore, the breaking of the global continuous symmetry leads to the existence of a massless state $\chi$, which is on the other hand argued to contribute to Dark Radiation~\cite{Hamann:2010bk,Dunkley:2010ge,Keisler:2011aw} -- see Ref.~\refcite{Abazajian:2012ys} for an extensive collection of observations. This could even be desirable from a cosmological point of view. On the other hand, certain potentially dangerous terms are neglected in the scalar potential, and the above mentioned $Z_2$ symmetry is a bit artificial as it has no effect other than stabilizing the lightest sterile neutrino $N_1$ and by this avoiding the X-ray bound, cf.\ Sec.~\ref{sec:keV_general_X}, in order to make the model consistent with current (non-) observations.

\item light sterile neutrinos from global symmetries~\cite{Sayre:2005yh}:

Another approach is to generate the sterile neutrino mass only by an effective operator, which itself arises from an unknown (or partially unknown) high energy sector. This possibility has been extensively discussed in Ref.~\refcite{Sayre:2005yh}, where exactly one sterile neutrino called $\nu_S$ is introduced, which receives mass from a Weinberg-like operator. The authors mention the possibility of an operator $\frac{1}{\Lambda} \overline{(\nu_S)^c} \nu_S H_1 H_2$ with two Higgs doublets $H_{1,2}$, but then discard this possibility because they are after $O(1)$ active-sterile mixing with a sterile neutrino mass of about $1$~eV, in order to explain the by then famous LSND anomaly~\cite{Aguilar:2001ty}. Instead, they use an operator $\frac{1}{\Lambda} \overline{(\nu_S)^c} \nu_S S S$ with a two singlet scalars $S$. While a bare (and hence potentially large) mass term $m_S \overline{(\nu_S)^c} \nu_S$ is forbidden by a global symmetry, the VEV $\langle S \rangle = v_S$ reintroduces the possibility of a higher order mass $\frac{v_S^2}{\Lambda}$ for the sterile neutrino. In particular, the authors discuss how to choose the required global symmetry appropriately without being in danger to end up with anomalies, cf.\ Ref.~\refcite{Sayre:2005yh} for details.

While the general approach of Ref.~\refcite{Sayre:2005yh} is certainly valuable, the details of the high energy sector are not specified in great detail, even though the results derived are not very sensitive to that. However, a more serious drawback (if we want to apply this model to the case of keV sterile neutrinos) is that there is no actual spectrum present in the sterile sector, so that the resulting neutrino mass is only a function of two unknown scales, $v_S$ and $\Lambda$. Nevertheless, the authors' approach certainly fulfills the criterion to completely disentangle the mass generation of the sterile neutrino from the electroweak scale, which allows for this freedom and which could be very useful for future model studies.

\item keV sterile neutrinos from Dark GUTs~\cite{Babu:2004mj}:

The possibility of a seesaw mechanism could not only be used in an active sector, but also in a \emph{dark} sector of a theory which hardly talks to the SM part. Such an approach has been taken in Ref.~\refcite{Babu:2004mj}, where different groups $G'$ are studied to be used in a setting similar to a Grand Unified Theory (GUT), which however does not include the SM itself as a gauge group. The field which plays the role of the sterile neutrino is charged non-trivially under this new symmetry $G'$, and in this sector it receives a suppressed mass in pretty much the same way as an active neutrino receives its mass in the ordinary seesaw type~I mechanism, cf.\ \ref{sec:seesaw}. This field is nevertheless a singlet under the SM, and it can hence play the role of a sterile neutrino.

While a couple of choices for the group $G'$ are discussed in Ref.~\refcite{Babu:2004mj}, the approach is nevertheless somewhat sketchy. Furthermore, although a general argument for the smallness of the sterile neutrino mass scale is given, there is no reason for enforcing a certain structure in the sterile sector, which could lead to phenomenological problems when trying to apply it to the case of keV sterile neutrino DM.

\item loop-masses for keV sterile neutrinos~\cite{Ma:2009gu}:

A very natural mechanism to explain hierarchies in the sterile neutrino sector has been presented by Ma in Ref.~\refcite{Ma:2009gu}, by using loop suppressions of the masses of certain sterile neutrinos. In the model suggested, SM singlet fermions $S_{1,2,3}$ are introduced in addition to the ``ordinary'' RH neutrinos $N_{1,2,3}$, along with several additional scalar fields. The arrangement of the fields and charges is such that a beautiful thing happens in the sterile sector: the field $S_1$ is the only additional singlet which has a tree-level mass $M_{S1}$. $S_2$ in turn, only receives a mass at 1-loop order by a diagram which involves the massive fermion $S_1$ in the loop, leading to a suppressed mass $M_{S2}$. Finally, the field $S_3$ also receives a mass by a 1-loop diagram that contains the field $S_2$ as internal fermion, which by itself only receives a 1-loop mass. This means that the corresponding mass $M_{S3}$ is in fact suppressed by a \emph{2-loop mechanism}. Such a setting naturally yields an explanation for $M_{S3} \ll M_{S2} \ll M_{S1}$, and by this it is a good motivation for $M_{S3} = \mathcal{O}({\rm keV})$.

However, what had been presented above is the beauty, and the beast comes along in the actual charge assignments: even though the above mechanism is plain and simple, the charge assignments necessary to generate it are, unfortunately, very complicated. Not only does the model require many more or less unmotivated scalar fields, it also requires an additional $U(1)$ symmetry with either a very large charge of one scalar or with relatively tricky charge assignments for the singlet fields $S_{1,2,3}$. While this model could certainly yield interesting phenomenology, its setting to some extent looks artificial and engineered, which is the price one has to pay for the elegant mechanism illustrated in the previous paragraph.

\item light sterile neutrinos in minimal radiative inverse seesaw~\cite{Dev:2012bd}:

The setting of Ref.~\refcite{Dev:2012bd} is based on the so-called \emph{inverse seesaw}~\cite{Mohapatra:1986aw,Mohapatra:1986bd} and it is, in effect, a combination of the seesaw mechanism with a loop-suppression. The idea is to extend the seesaw type~I setting by further LH singlet fermions $S_L$, which themselves obtain a very large Majorana-type mass $\mu_S$ and which also form a Dirac mass term with the ordinary right-handed neutrinos $N_R$. Integrating out the heaviest singlets $S_L$ then leads, at tree-level, to a mass matrix that looks very similar to an ordinary type~II mass matrix, and in particular it has a seesaw structure in all entries. However, it turns out that the rank of this matrix is nevertheless too small to yield light sterile neutrino masses (in other words, a cancellation forces them to be zero). On the other hand, at 1-loop level this mass matrix receives a correction which induces small sterile neutrino masses that are naturally loop-suppressed in addition to the seesaw structure that was already present. One could say that the loop-correction destroys the exact cancellation just mentioned.

This mechanism is very strong, so that sterile neutrino masses of $\mathcal{O}({\rm keV})$, or even lighter, are no problem to achieve. Furthermore, this mechanism is also quite minimal, in the sense that the seesaw framework is not extended by very much, while the suppressions follow naturally from the general structure. However, one might ask why the Majorana neutrino mass for the left-handed singlets is so much bigger than the (equally unconstrained) one for the right-handed neutrinos. In addition, a further potential issue is that this mechanism again leads to a very strong suppression of certain sterile neutrino masses, but there is no argument for a particular mass pattern.

\item anomalous sterile neutrino masses from gravitational torsions~\cite{Mavromatos:2012cc}:

A more exotic way of generating a light sterile neutrino mass is presented in Ref.~\refcite{Mavromatos:2012cc}. The mechanism relies on so-called \emph{quantum torsions}, which may arise in string theory or in more general quantum gravity settings. The general idea is to introduce two axion-like pseudoscalar fields $a$ and $b$, one providing the chirality change in the coupling of two RH neutrino fields $N_R$ and one being non-propagating and providing the coupling to the quantum torsion. Such a scenario allows to draw a two-loop diagram giving mass to $N_R$, which involves the pseudoscalar $a$, linearized graviton fields, and a coupling to the torsion.

The setting yields two possible suppressions: with only one axion $a$, the diagram is strongly UV-divergent and the mass $M_R$ is proportional to the sixth power of the cut-off scale $\Lambda$. Hence, for a low string mass/cut-off close to the GUT scale, $M_R$ can be of $\mathcal{O}(10~{\rm keV})$ for a Yukawa coupling $y_a = \mathcal{O} (10^{-3})$ in the term proportional to $a \overline{(N_R)^c} N_R$. Alternatively, if there are at least three axions $a_n$ ($n\geq 3$), the mass $M_R$ can be made finite (at two-loop level), with its size being controlled by the ratio of axion mass-splittings to the axion mass scale, $\left( \frac{\delta M_a^2}{M_a^2} \right)^n$.

Although this mechanism yields a general suppression for geometries with torsion in a quantum gravity setting, one could criticize that it relies on a relatively complicated framework, and hardly any scales and couplings in the game are fixed or at least tied to some other scale. Such ties could be naturally given in a more concrete model which, however, does not exist up to know. Nevertheless, the mechanism is powerful and reveals an interesting starting point for further studies.

\end{itemize}

Let us end this chapter by commenting on a generalization of keV sterile neutrinos, and other quasi-sterile fermions. Extremely weakly interacting particles with a mass of around a few keV acting as (typically warm) Dark Matter can also be found in frameworks other than keV sterile neutrinos, for example axinos~\cite{Jedamzik:2005sx}, gravitinos~\cite{Gorbunov:2008ui,Jedamzik:2005sx,Baltz:2001rq}, KK-gravitons~\cite{Jedamzik:2005sx}, majorons~\cite{Lattanzi:2007ux,Frigerio:2011in}, mixed axion-axino DM~\cite{Baer:2008yd}, modulinos~\cite{Dvali:1998qy,Benakli:1997iu,Benakli:1997hb}, or singlinos~\cite{McDonald:2008ua}.

This observation has led the authors of Ref.~\refcite{King:2012wg} to a generalization of the \emph{fermionic} WDM candidates to a new class of DM called \emph{keVins}, which stands for \underline{keV} \underline{in}ert fermion\underline{s}. Just as WIMPs (\underline{W}eakly \underline{I}nteracting \underline{M}assive \underline{P}article\underline{s}), keVins comprise a relatively general class of DM particles, of which keV sterile neutrinos are only one special case. However, also for keVins there is in general no reason to have a mass of a few keV, just as there is no reason for a general WIMP to be weakly interacting. Such a reason has to be imposed in a concrete model.

Note that in Ref.~\refcite{King:2012wg}, the mechanism to produce the correct DM abundance is, just as in the LR-symmetric versions of keV sterile neutrinos~\cite{Bezrukov:2009th,Nemevsek:2012cd}, thermal overproduction with subsequent dilution by entropy production~\cite{Scherrer:1984fd}. For keVins, the particle that decays and produces the required entropy is a heavier version of the DM particle with a mass of a few GeV: the \emph{GeVin}. However, this method is only one possibility to achieve the correct DM abundance, and even the production mechanism mentioned here has not yet been studied in very great detail. In general, the class of keVins offers more possibilities for relatively general studies, which could then be used to discover new successful regions in the parameter space which might not be visible in all concrete models.\\

In this central chapter of the review, we have seen that the key point in keV sterile neutrino model building is to find a motivation for the keV scale in the first place. This by itself is far from trivial, and most of the known mechanisms achieving an explanation have been presented in Secs.~\ref{sec:keV_FN} to~\ref{sec:keV_comp}. Interesting further possibilities arise when trying to generate keV neutrino masses, \emph{e.g.} at loop-level, but although several such settings are available for light neutrinos, there is still no \emph{simple} model to generate a sterile neutrino mass in that way. Of course, further mechanisms and arguments could be found, and by this we could not only establish links between two very fundamental sectors in particle physics, neutrinos and Dark Matter, but we could also gain an improved understanding of our model building techniques by investigating in how far such an extremely hierarchical mass spectrum as needed for keV sterile neutrinos could be realized.

\section{\label{sec:conc}Conclusions}	

We have presented a review of the current ideas present in the literature which can explain or at least motivate the existence of a keV sterile neutrino. Such a particle would be well suited to act as (typically warm) Dark Matter and it can have very interesting connections to the low energy neutrino sector. While keV sterile neutrinos are very well investigated from the astrophysics side by studies of their production mechanisms and their behavior in a cosmological context, the model building side has only become a rising field within the last few years. Nevertheless, a review at exactly this stage is useful in order to motivate the scientists from the different disciplines to collaborate on this exiting topic. It has become clear for some model builders that there are many interesting relations of keV sterile neutrinos to the light neutrino sector, which could lead to both, new types of mass-creating mechanisms as well as testability of the models by the light neutrino sector. These interrelations are under investigation by several groups by now, but nevertheless the corresponding scientific works have sometimes been hard to understand for people from neighboring fields. This gap has hopefully been bridged at least to some extent by the present review.

Ideally, we should soon enter an era where model studies are combined with detailed studies of the cosmological implications of sterile neutrinos. Many exciting new experiments and observations are on the way, and they can be expected to have a considerable impact on the field within the coming years. We should take the opportunity we have now to be prepared in order to make optimal use of the data to be gained. New discoveries are the advances of science, and the growing field of keV sterile neutrino Dark Matter will hopefully benefit from several such discoveries in the near future.

\section*{Acknowledgments}

AM is grateful to his collaborators in his projects on keV neutrinos and on related topics: P.~Di Bari, M.~D\"urr, S.~F.~King, M.~Lindner, W.~Rodejohann, D.~Schmidt, and in particular V.~Niro. AM furthermore wants to thank V.~Niro and A.~Stuart for carefully reading the first version of this manuscript and giving valuable comments, and he is particularly grateful to A.~Stuart for useful discussions on vacuum alignment. Useful comments on the preprint version of this manuscript from J.~Barry, H.~de Vega, A.~Gomes Dias, J.~Heeck, A.~Kusenko, N.~Mavromatos, D.~Meloni, and N.~Sanchez are acknowledged, too. This review could never have been completed without the continuous motivation boosts provided by M.~Abb. Finally, AM acknowledges financial support by a Marie Curie Intra-European Fellowship within the 7th European Community Framework Programme FP7-PEOPLE-2011-IEF, contract PIEF-GA-2011-297557, and partial support from the European Union FP7  ITN-INVISIBLES (Marie Curie Actions, PITN-GA-2011-289442).

\appendix

\section{\label{sec:seesaw}The seesaw mechanism in detail}

In this appendix, we perform the seesaw calculation in great detail. Although known in principle, in many references (sometimes even in the supposed-to-be-pedagogical ones) this calculation is not presented in sufficient detail, and certain logical steps are skipped over. Due to its importance, we will go through it without skipping anything. We will furthermore explain two related technicalities which are of interest here, namely the so-called \emph{Casas-Ibarra parametrization} and the notion of \emph{form dominance}.

\subsection{\label{sec:seesaw_matrix}The seesaw mass matrix}

The Dirac mass term for neutrinos is given by
\begin{equation}
 \mathcal{L}_D = -\overline{\nu_L} m_D N_R + h.c.,
 \label{eq:Dirac_app}
\end{equation}
while the LH and RH Majorana mass terms are given by
\begin{equation}
 \mathcal{L}_M = -\frac{1}{2} \overline{\nu_L} m_L (\nu_L)^c -\frac{1}{2} \overline{(N_R)^c} M_R N_R + h.c.
 \label{eq:Majorana_app}
\end{equation}
The combination of the two terms, $\mathcal{L}_\nu = \mathcal{L}_D + \mathcal{L}_M$, leads to the well-known general seesaw type~I~\cite{Minkowski:1977sc,Yanagida:1979as,GellMann:1980vs,Glashow:1979nm,Mohapatra:1979ia} and type~II situations~\cite{Magg:1980ut,Lazarides:1980nt}, while the limiting case where $m_L \equiv 0$ is just the type~I seesaw and the limiting case where $m_D = M_R \equiv 0$, where $m_L$ arises due to a Higgs triplet field, is just type~II. Note also the structural difference between the LH and RH Majorana mass terms in Eq.~\eqref{eq:Majorana_app}: one could have defined them in an analogous matter, which would simply amount to a redefinition of either $m_L$ or $M_R$, but the choice made here will turn out convenient later on.

Note that Majorana mass matrices are symmetric, $m_L^T \equiv m_L$ and $M_R^T \equiv M_R$: for any Majorana spinor $\Psi$, we have $\Psi^c \equiv e^{i\phi} \Psi$, with a Majorana phase $\phi$, from which it follows that
\begin{eqnarray}
 && \overline{\Psi^c_\alpha} M_{\alpha \beta} \Psi_\beta = e^{-i\phi_\alpha} \overline{\Psi_\alpha} M_{\alpha \beta} \Psi_\beta = e^{-i\phi_\alpha} (\overline{\Psi_\alpha} M_{\alpha \beta} \Psi_\beta)^T = - e^{-i\phi_\alpha} \Psi_\beta^T M^T_{\beta \alpha} \underbrace{(-C^2)}_{={\bf 1}} \overline{\Psi_\alpha}^T \nonumber\\
 && = e^{-i\phi_\alpha} \Psi_\beta^T C M^T_{\beta \alpha} C \overline{\Psi_\alpha}^T = e^{-i\phi_\alpha} \overline{\Psi_\beta^c} M^T_{\beta \alpha} \Psi_\alpha^c = \overline{\Psi_\beta^c} M^T_{\beta \alpha} \Psi_\alpha,
 \label{eq:Majorana_sym}
\end{eqnarray}
where the first minus sign arises from the anti-commutation of fermion fields. We have also used the ordinary charge conjugation formula $\Psi^c = C \overline{\Psi}^T$, which holds for any spinor $\Psi$, Majorana or not, with the charge conjugation matrix $C = i \gamma^2 \gamma^0$. We have further exploited that
\begin{equation}
 \Psi^c = C \overline{\Psi}^T = C (\Psi^\dagger \gamma^0)^T = \underbrace{C \gamma^0}_{-\gamma^0 C} \Psi^* \Rightarrow \gamma^0 \Psi^c = \underbrace{-C}_{C^\dagger} \Psi^* \Rightarrow \overline{\Psi^c} = \Psi^T C. 
 \label{eq:CC_trick}
\end{equation}
Renaming the indices on the right-hand side of Eq.~\eqref{eq:Majorana_sym} finally leads to 
\begin{equation}
 \overline{\Psi^c_\alpha} M_{\alpha \beta} \Psi_\beta = \overline{\Psi_\alpha^c} M^T_{\alpha \beta} \Psi_\beta,
 \label{eq:Majorana_sym_2}
\end{equation}
from which we can immediately read off that $M = M^T$.

The final trick to arrive at the correct form of the seesaw mass matrix is to split up the Dirac mass term, Eq.~\eqref{eq:Dirac_app}, into two pieces:
\begin{equation}
 \overline{\nu_L} m_D N_R = \frac{1}{2}\overline{\nu_L} m_D N_R + \frac{1}{2}(\overline{\nu_L} m_D N_R)^T.
 \label{eq:Dirac_split_1}
\end{equation}
Making use of the anti-commutation of fermion fields as well as of Eq.~\eqref{eq:CC_trick} one more time, we can rewrite the transposed term:
\begin{equation}
 (\overline{\nu_L} m_D N_R)^T = - (N_R)^T m_D^T \underbrace{(-C^2)}_{={\bf 1}} \overline{\nu_L}^T = \underbrace{(N_R)^T C}_{= \overline{(N_R)^c}} m_D^T \underbrace{C \overline{\nu_L}^T}_{= (\nu_L)^c}.
 \label{eq:Dirac_split_2}
\end{equation}
Putting all together, we finally obtain:
\begin{eqnarray}
 \mathcal{L}_\nu &=& \mathcal{L}_D + \mathcal{L}_M = -\frac{1}{2} \overline{\nu_L} m_D N_R - \frac{1}{2} \overline{(N_R)^c} m_D^T (\nu_L)^c -\frac{1}{2} \overline{\nu_L} m_L (\nu_L)^c -\frac{1}{2} \overline{(N_R)^c} M_R N_R + h.c.\nonumber \\
 & = & -\frac{1}{2} (\overline{\nu_L} , \overline{(N_R)^c})
 \begin{pmatrix}
 m_L & m_D\\
 m_D^T & M_R
 \end{pmatrix}
 \begin{pmatrix}
 (\nu_L)^c\\
 N_R
 \end{pmatrix} + h.c.
 \label{eq:total_seesaw}
\end{eqnarray}
This equation involves the full (typically $6\times 6$) neutrino mass matrix.

\subsection{\label{sec:seesaw_diag}The seesaw diagonalization}

The next step is to (approximately) diagonalize the full neutrino mass matrix:
\begin{equation}
 M_\nu = \begin{pmatrix}
 m_L & m_D\\
 m_D^T & M_R
 \end{pmatrix},
 \label{eq:full_mass}
\end{equation}
where one assumes that somehow $m_L \ll m_D \ll M_R$. Of course, this is an intrinsically ill-defined condition when dealing with matrices, in particular since even the ``large'' matrix $M_R$ could have texture zeros. Still, we will follow the practical assumption that the hierarchy between the mass matrices is justified, but we should keep in mind that in some cases it may be useful to numerically check if the low energy mass matrix actually does conserve the structure of the full mass matrix.

Similar diagonalization procedures as the one here are outlined in, \emph{e.g.}, Refs.~\refcite{Barry:2011fp,Kiers:2005vx}, although we slightly depart from both references. The first step is to approximately block-diagonalize the matrix by a nearly unitary matrix $\tilde U$ given by
\begin{equation}
 \tilde U = \begin{pmatrix}
 {\bf 1} - \frac{1}{2} b b^\dagger & b\\
 - b^\dagger & {\bf 1} - \frac{1}{2} b^\dagger b
 \end{pmatrix},
 \label{eq:U_tilde}
\end{equation}
where the parameter $b$, to be defined later, is assumed to be of $\mathcal{O}(m_D M_R^{-1})$, and it can hence be regarded as small. In particular, since $m_L$ is the smallest mass scale in the game and since we expect $\mathcal{O}(m_L) = \mathcal{O}(m_D^2 M_R^{-1})$ for type~II seesaw, we will only need to keep $m_L$ with a coefficient of $\mathcal{O}(1)$, $m_D$ up to a coefficient of $\mathcal{O}(b)$, and $M_R$ up to a coefficient of $\mathcal{O}(b^2)$. Then, we can immediately see that
\begin{equation}
 \tilde U \tilde U^\dagger = \tilde U^\dagger \tilde U = \begin{pmatrix}
 {\bf 1} + \frac{1}{4} b b^\dagger b b^\dagger & 0\\
 0 & {\bf 1} + \frac{1}{4} b^\dagger b b^\dagger b
 \end{pmatrix} = {\bf 1} + \mathcal{O}(b^4) \simeq {\bf 1}.
 \label{eq:U_tilde_unit}
\end{equation}
The next step is to calculate $\tilde D \equiv \tilde U^T M_\nu \tilde U$, keeping only relevant (\emph{i.e.}, large enough) terms. For example, we have
\begin{eqnarray}
 \tilde U^T M_\nu && = \begin{pmatrix}
 m_L - \frac{1}{2} b^* b^\dagger m_L - b^* m_D^T & \  & m_D - \frac{1}{2} b^* b^T m_D - b^* M_R\\
 b^T m_L + m_D^T - \frac{1}{2} b^T b^* m_D^T & \ & b^T m_D + M_R - \frac{1}{2} b^T b^* M_R
 \end{pmatrix} \nonumber\\
 && \simeq \begin{pmatrix}
 m_L - b^* m_D^T & \  & m_D - b^* M_R\\
 m_D^T & \  & b^T m_D + M_R - \frac{1}{2} b^T b^* M_R
 \end{pmatrix}.
 \label{eq:diagM_1}
\end{eqnarray}
All that remains is to calculate the product,
\begin{equation}
 \tilde D \simeq \begin{pmatrix}
 m_L - b^* m_D^T & \  & m_D - b^* M_R\\
 m_D^T & \  & b^T m_D + M_R - \frac{1}{2} b^T b^* M_R
 \end{pmatrix} \begin{pmatrix}
 {\bf 1} - \frac{1}{2} b b^\dagger & b\\
 - b^\dagger & {\bf 1} - \frac{1}{2} b^\dagger b
 \end{pmatrix}. 
 \label{eq:diagM_2}
\end{equation}
The off-diagonal elements of this expression should vanish, which determines the matrix $b$ unambiguously. For the 12-element, we have
\begin{equation}
 \tilde D_{12} \simeq m_L b - b^* m_D^T b + m_D - b^* M_R \simeq m_D - b^* M_R \stackrel{!}{=} 0\ \ \Rightarrow \ \ b^* = m_D M_R^{-1},
 \label{eq:D_12}
\end{equation}
which leads to $b = m_D^* {M_R^{-1}}^*$, $b^T = {M_R^{-1}}^* m_D^\dagger$, and $b^\dagger = {b^T}^* = M_R^{-1} m_D^T$. This also implies that $\tilde D_{21} \simeq 0$. Similarly, we can calculate the diagonal blocks:
\begin{eqnarray}
 \tilde D_{11} &\simeq& m_L - b^* m_D^T - m_D b^\dagger + b^* M_R b^\dagger = m_L - m_D M_R^{-1} m_D^T, \nonumber\\
 \tilde D_{22} &\simeq& M_R + \frac{1}{2} \left[ m_D^T m_D^* {M_R^{-1}}^* + {M_R^{-1}}^* m_D^\dagger m_D \right] \simeq M_R.
 \label{eq:D_1122}
\end{eqnarray}
By this we have recovered the general seesaw formula,
\begin{equation}
 \tilde D \equiv \tilde U^T M_\nu \tilde U \simeq \begin{pmatrix}
 m_L - m_D M_R^{-1} m_D^T & 0\\
 0 & M_R
 \end{pmatrix} \equiv \begin{pmatrix}
 m_\nu & 0\\
 0 & M_R
 \end{pmatrix}.
 \label{eq:seesaw_general}
\end{equation}
In the basis where the charged lepton masses are diagonal, $M_e = {\rm diag}(m_e, m_\mu, m_\tau)$, the remaining $3\times 3$ blocks can be diagonalized by $U_{\rm PMNS}^T m_\nu U_{\rm PMNS} = {\rm diag}(m_1, m_2, m_3) \equiv D_\nu$ and $V_R^T M_R V_R = {\rm diag}(M_1, M_2, M_3) \equiv D_N$, which finally leads to
\begin{eqnarray}
 && D \equiv \overline{U}^T M_\nu \overline{U} = \begin{pmatrix}
 D_\nu & 0\\
 0 & D_N
 \end{pmatrix},\ \ {\rm where} \label{eq:seesaw_general_diag} \\
 && \overline{U}\equiv \begin{pmatrix}
 {\bf 1} - \frac{1}{2} m_D^* {M_R^{-1}}^* M_R^{-1} m_D^T & m_D^* {M_R^{-1}}^*\\
 - M_R^{-1} m_D^T & {\bf 1} - \frac{1}{2} M_R^{-1} m_D^T m_D^* {M_R^{-1}}^*
 \end{pmatrix} \begin{pmatrix}
 U_{\rm PMNS} & 0\\
 0 & V_R
 \end{pmatrix},
 \nonumber
\end{eqnarray}
which is a central result of this appendix.

In certain situations it might be convenient to work in a basis where the charged lepton mass matrix $M_e$ is not diagonal, but is instead diagonalized by a unitary matrix $U_e$ according to $U_e^\dagger M_e^\dagger M_e U_e = {\rm diag}(m_e^2, m_\mu^2, m_\tau^2)$. Then, the PMNS matrix is given by the mismatch of the two bases, $U_{\rm PMNS} = U_e^\dagger U_\nu$, where the unitary matrix $U_\nu$ diagonalizes the neutrino mass matrix according to $U_\nu^T m_\nu U_\nu$ in the Majorana case. We will, in such situations, distinguish the charged lepton and neutrino-related quantities by super- or subscripts $l$ and $\nu$, respectively. For example, in Sec.~\ref{sec:neutrino_modeling_flavor} such a case appears, where the neutrino-related mixing angles are denoted by $\theta^\nu_{ij}$. However, the physical leptonic mixing angles are always denoted without such a super- or subscript.

\subsection{\label{sec:seesaw_CI}The Casas-Ibarra parametrization}

In the context of the seesaw type~I mechanism, a smart parametrization has been proposed by Casas and Ibarra in Ref.~\refcite{Casas:2001sr}. The trick of this parametrization is to make maximal use of the observable parameters in a seesaw type~I setting, while all unknowns are cast in an orthogonal matrix $R$. We will shortly review the derivation of the parametrization in the conventions used in this review.

The starting point is the type~I seesaw formula for light neutrinos, cf.\ Eq.~\eqref{eq:seesaw_general},
\begin{equation}
 m_\nu = - m_D M_R^{-1} m_D^T,
 \label{eq:CI_1}
\end{equation}
in a basis where the charged lepton and RH neutrino mass matrices are diagonal. Then, exactly as in Eq.~\eqref{eq:seesaw_general_diag}, $m_\nu$ can be diagonalized by the PMNS matrix $U \equiv U_{\rm PMNS}$, $D_\nu = {\rm diag}(m_1, m_2, m_3) = U^T m_\nu U$, where $U$ has exactly the form given in Eq.~\eqref{eq:PMNS_expl}. With the following definitions,
\begin{equation}
 D_{\sqrt{\nu}} \equiv {\rm diag}(\sqrt{m_1}, \sqrt{m_2}, \sqrt{m_3})\ \ \ {\rm and}\ \ \ D^{-1}_{\sqrt{N}} \equiv {\rm diag}(1/\sqrt{M_1}, 1/\sqrt{M_2}, 1/\sqrt{M_3}),
 \label{eq:CI_2}
\end{equation}
one can rewrite the diagonalization of Eq.~\eqref{eq:CI_1},
\begin{equation}
 D_{\sqrt{\nu}} D_{\sqrt{\nu}} = D_\nu = U^T m_\nu U = - U^T m_D M_R^{-1} m_D^T U = U^T i m_D D^{-1}_{\sqrt{N}} D^{-1}_{\sqrt{N}} i m_D^T U.
 \label{eq:CI_3}
\end{equation}
Multiplying from left and right by $D^{-1}_{\sqrt{\nu}}$, one obtains
\begin{equation}
 {\bf 1} = \left[ D^{-1}_{\sqrt{N}} i m_D^T U D^{-1}_{\sqrt{\nu}} \right]^T \left[ D^{-1}_{\sqrt{N}} i m_D^T U D^{-1}_{\sqrt{\nu}} \right],
 \label{eq:CI_4}
\end{equation}
which suggests the definition of the complex orthogonal matrix
\begin{equation}
 R \equiv i D^{-1}_{\sqrt{N}} m_D^T U D^{-1}_{\sqrt{\nu}}.
 \label{eq:CI_5}
\end{equation}

Note that the above parametrization can also be extended to seesaw type~II scenarios~\cite{Akhmedov:2008tb}: the trick is simply to make the replacement
\begin{equation}
 m_\nu \to X_\nu \equiv m_\nu - m_L
 \label{eq:CI_6}
\end{equation}
in Eq.~\eqref{eq:CI_1}, where $m_L$ is the LH neutrino mass matrix, cf.\ Eq.~\eqref{eq:total_seesaw}. Of course, this also means that in the subsequent equations one has to replace
\begin{equation}
 D_\nu \to D_X = {\rm diag}(X_1, X_2, X_3)\ \ \ {\rm and}\ \ \ U \to U_X,
 \label{eq:CI_7}
\end{equation}
where $X_i$ are the eigenvalues of the matrix $X_\nu$ and $U_X$ is the matrix that diagonalizes $X_\nu$ according to $D_X = {\rm diag}(X_1, X_2, X_3) = U_X^T X_\nu U_X$. While by this one loses the direct connection to the PMNS matrix, it is nevertheless useful since one can still exploit the general parametrization of the complex orthogonal matrix $R$,
\begin{equation}
 R = R_{12} R_{13} R_{23},
 \label{eq:CI_8}
\end{equation}
where $R_{ij}$ are the rotations by complex angles $\omega_{ij} = \rho_{ij} + i \sigma_{ij}$ in the $ij$-plane,
\begin{equation}
 R_{12} = \begin{pmatrix}
 c_{\omega_{12}} & s_{\omega_{12}} & 0\\
 -s_{\omega_{12}} & c_{\omega_{12}} & 0\\
 0 & 0 & 1
 \end{pmatrix}, \ \ \ R_{13} = \begin{pmatrix}
 c_{\omega_{13}} & 0 & s_{\omega_{13}}\\
 0 & 1 & 0\\
 -s_{\omega_{13}} & 0 & c_{\omega_{13}}
 \end{pmatrix},\ \ \ R_{23} = \begin{pmatrix}
 1 & 0 & 0\\
 0 & c_{\omega_{23}} & s_{\omega_{23}}\\
 0 & -s_{\omega_{23}} & c_{\omega_{23}} \end{pmatrix},
 \label{eq:CI_9}
\end{equation}
where $c_{\omega_{ij}} \equiv \cos \omega_{ij}$ and $s_{\omega_{ij}} \equiv \sin \omega_{ij}$.

\subsection{\label{sec:seesaw_FD}Form dominance}

The notion of \emph{form dominance} was defined in Ref.~\refcite{Chen:2009um}. In a seesaw type~I situation, we first denote the columns of the Dirac mass matrix by 3-vectors in a basis where the right-handed neutrino mass matrix is diagonal,\footnote{Only here, we change the notation of the right-handed masses to $M_{A,B,C}$, since it will make it easier to follow the calculation.}
\begin{equation}
 m_D = (\vec{A}, \vec{B}, \vec{C}),\ \ \ M_R = {\rm diag} (M_A, M_B, M_C).
 \label{eq:FD_1}
\end{equation}
Using this notation, the seesaw formula can be easily shown to have the form,
\begin{equation}
 m_\nu = - m_D M_R^{-1} m_D^T = - \left( \frac{\vec{A} \otimes \vec{A}^T}{M_A} + \frac{\vec{B} \otimes \vec{B}^T}{M_B} + \frac{\vec{C} \otimes \vec{C}^T}{M_C} \right),
 \label{eq:FD_2}
\end{equation}
where the tensor product is defined in the usual way, \emph{e.g.},
\begin{equation}
 \vec{A} \otimes \vec{A}^T = \begin{pmatrix}
 A_1 A_1 & A_1 A_2 & A_1 A_3\\
 A_2 A_1 & A_2 A_2 & A_2 A_3\\
 A_3 A_1 & A_3 A_2 & A_3 A_3
 \end{pmatrix},\ \ \ {\rm where}\ \ \ \vec{A} = \begin{pmatrix}
 A_1\\
 A_2\\
 A_3
 \end{pmatrix}.
 \label{eq:FD_3}
\end{equation}
Let us now write the PMNS matrix $U$ in terms of its column vectors,
\begin{equation}
 U = (\vec{U}_1, \vec{U}_2, \vec{U}_3).
 \label{eq:FD_4}
\end{equation}
Then, in the basis where the charged lepton mass matrix is diagonal, form dominance is defined as the situation where each column of $m_D$ is proportional to the corresponding column of the conjugate PMNS matrix $U^*$,\footnote{Note that in Ref.~\refcite{Chen:2009um} $U$ was taken to be real, which is why the complex conjugation does not appear in the corresponding formula.}
\begin{equation}
 \text{\em form dominance} : \Leftrightarrow (\vec{A}, \vec{B}, \vec{C}) = (a \vec{U}_1^*, b \vec{U}_2^*, c \vec{U}_3^*),
 \label{eq:FD_5}
\end{equation}
where $a,b,c$ are arbitrary complex numbers. Note that, once a proportionality of the columns is given, we can always write this as in Eq.~\eqref{eq:FD_5}, by relabeling the columns $\vec{A}$, $\vec{B}$, and $\vec{C}$. Explicitly, one can write:
\begin{equation}
 \vec{A} = a \begin{pmatrix}
 U_{e 1}\\
 U_{\mu 1}\\
 U_{\tau 1}
 \end{pmatrix}, \vec{B} = b \begin{pmatrix}
 U_{e 2}\\
 U_{\mu 2}\\
 U_{\tau 2}
 \end{pmatrix}, \vec{C} = c \begin{pmatrix}
 U_{e 3}\\
 U_{\mu 3}\\
 U_{\tau 3}
 \end{pmatrix}.
 \label{eq:FD_6}
\end{equation}
Since the columns of a unitary matrix are orthonormal, $\vec{U}_i^\dagger \vec{U}_j = \delta_{ij}$ and $\vec{U}_i^T \vec{U}_j^* = \delta_{ij}$, one can easily show that:
\begin{equation}
 U^T \vec{A} \otimes \vec{A}^T U = a^2 \begin{pmatrix}
 1 & 0 & 0\\
 0 & 0 & 0\\
 0 & 0 & 0
 \end{pmatrix},
 U^T \vec{B} \otimes \vec{B}^T U = b^2 \begin{pmatrix}
 0 & 0 & 0\\
 0 & 1 & 0\\
 0 & 0 & 0
 \end{pmatrix},
 U^T \vec{C} \otimes \vec{C}^T U = c^2 \begin{pmatrix}
 0 & 0 & 0\\
 0 & 0 & 0\\
 0 & 0 & 1
 \end{pmatrix}.
 \label{eq:FD_7}
\end{equation}
This allows to express the light neutrino masses in terms of the heavy neutrino masses and only three free parameters,
\begin{equation}
 U^T m_\nu U =  - U^T \left( \frac{\vec{A} \otimes \vec{A}^T}{M_A} + \frac{\vec{B} \otimes \vec{B}^T}{M_B} + \frac{\vec{C} \otimes \vec{C}^T}{M_C} \right) U = {\rm diag} \left( - \frac{a^2}{M_A}, - \frac{b^2}{M_B}, - \frac{c^2}{M_C} \right).
 \label{eq:FD_8}
\end{equation}
Finally, we note that form dominance is equivalent to setting $R = {\bf 1}$ in the Casas-Ibarra parametrization, cf.\ \ref{sec:seesaw_CI}.\cite{Choubey:2010vs}. Indeed, using Eq.~\eqref{eq:CI_5}, one obtains
\begin{equation}
 m_D =U^* D_{\sqrt{\nu}} (-i R^T) D_{\sqrt{N}} = U^* D_{\sqrt{\nu N}} = -i (\sqrt{m_1 M_1} \vec{U}_1^*, \sqrt{m_2 M_2} \vec{U}_2^*, \sqrt{m_3 M_3} \vec{U}_3^*),
 \label{eq:FD_9}
\end{equation}
which just coincides with Eq.~\eqref{eq:FD_5}.

\section{\label{sec:vacuum}Vacuum alignment}

This appendix is supposed to give some more information on the procedures necessary to get the desired \emph{vacuum alignment} in models based on discrete symmetries. We have already noted that this part of model building is often a bit technical, yet it is necessary in order to arrive at a consistent model. To give a flavor of how this is done, we discuss how to obtain the vacuum alignment for the model from Ref.~\refcite{Chen:2009um}, which had been discussed in Sec.~\ref{sec:neutrino_modeling_flavor}. Note that the vacuum alignment was (purposely) not given in Ref.~\refcite{Chen:2009um}, but it is actually not too difficult to construct.

As we had seen, the $A_4$ model presented contains one triplet flavon $\phi_S\sim \mathbf{3}$ and one singlet flavon $u\sim \mathbf{1}$, which should obtain VEVs as specified in Eq.~\eqref{eq:ex_VEVs} which, using the abbreviations $\tilde \alpha_i \equiv \alpha_i \Lambda_F$, look like
\begin{equation}
 \langle \phi_S \rangle = \begin{pmatrix} 1\\ 1\\ 1\end{pmatrix} \tilde \alpha_S,\ \ \ \langle u \rangle = \tilde \alpha_0.
 \label{eq:ex_VEVs_abbr}
\end{equation}
These VEVs can be obtained in the same way as the VEV of the SM-Higgs, namely by minimizing the full scalar potential of the model. In order to do this, however, we first have to construct the most general potential.

This is done with the $A_4$ multiplication rules we had discussed in Sec.~\ref{sec:neutrino_modeling_flavor}, using the information we have about the transformation properties of the flavons $\phi_S$ and $u$. However, we have to be careful not to forget that the model also contains the SM-Higgs field $H$, which is a doublet under $SU(2)_L$ but a singlet under $A_4$. Thus the combination $H^\dagger H$ is a total singlet, under the gauge group as well as under the discrete symmetry. While it would in principle be correct to also include couplings between the Higgs and the flavon sector, they are often neglected in practice, since one can argue that the flavor breaking happens at a scale much higher than electroweak symmetry breaking, such that the corrections induced by the Higgs are small~\cite{Ding:2013hpa}. Furthermore, the Higgs will not even develop a VEV at the high flavor breaking scale, and and in order to compute the back-reactions of the flavon sector onto the Higgs, one would need to solve the full set of renormalization group equations.

Putting that aside, we are left with the task of constructing all possible singlet combinations of scalar fields which appear in the potential. Since the potential is part of the Lagrangian, only trivial singlet combinations are allowed, and if we want to keep the potential renormalizable we are only allowed to keep terms of a mass dimension smaller than or equal to four.\footnote{Note that this statement actually introduces a slight inconsistency, since in the model under consideration the non-renormalizable Weinberg operator had been used, cf.\ Eq.~\eqref{eq:Wein_op}, so one could ask why renormalizability is required for the scalar potential. Indeed, adding higher order terms to the potential could be a potential solution for a problem we will encounter in the computation of the vacuum alignment. However, since the calculation here is only an example, we will not discuss this subtlety further.} Then we can make use of Eqs.~\eqref{eq:dir_prod}, \eqref{eq:dir_prod_sing}, \eqref{eq:dir_prod_trip}, \eqref{eq:dir_prod_tripsing}, and~\eqref{eq:further_prod} in order to construct all the singlet combinations. For example, a direct product of two fields $\phi_S$ contains exactly one trivial singlet combination, which according to Eqs.~\eqref{eq:dir_prod} and \eqref{eq:dir_prod_sing} is for $\phi_S = (\varphi_1, \varphi_2, \varphi_3)^T$ given by
\begin{equation}
 (\phi_S \otimes \phi_S)_{\mathbf{1}} =(\varphi_1^2 + 2 \varphi_2 \varphi_3)\sim \mathbf{1}.
 \label{eq:sing_phiS}
\end{equation}
Hence, this combination of fields could also be squared or combined with $u$ or $u^2$ (as well as $H^\dagger H$, in principle) and all these terms would be allowed to appear in the scalar potential. Noting furthermore that the asymmetric triplet contraction $\mathbf{3}_A$ contained in Eq.~\eqref{eq:dir_prod} will vanish if two identical fields $\phi_S$ are combined, cf.\ second Eq.~\eqref{eq:dir_prod_trip}, all the possible field combinations transforming as total singlets under $A_4$ can be constructed:
\begin{eqnarray}
 u_{\mathbf{1}} & \to & u, \nonumber \\
 (u \otimes u)_{\mathbf{1}} & \to & u^2, \nonumber \\
 (u \otimes u \otimes u)_{\mathbf{1}} & \to & u^3, \nonumber \\
 (u \otimes u \otimes u \otimes u)_{\mathbf{1}} & \to & u^4, \nonumber \\
 (\phi_S \otimes \phi_S)_{\mathbf{1}} & \to & \varphi_1^2 + 2 \varphi_2 \varphi_3, \nonumber \\
 \left[(\phi_S \otimes \phi_S)_{\mathbf{3}_S} \otimes \phi_S  \right]_{\mathbf{1}} & \to &  \varphi_1^3 + \varphi_2^3 + \varphi_3^3  - 3 \varphi_1 \varphi_2 \varphi_3, \nonumber \\
 \left[(\phi_S \otimes \phi_S)_{\mathbf{1}} \otimes (\phi_S \otimes \phi_S)_{\mathbf{1}} \right]_{\mathbf{1}} & \to & (\varphi_1^2 + 2 \varphi_2 \varphi_3)^2, \nonumber \\
 \left[(\phi_S \otimes \phi_S)_{\mathbf{1'}} \otimes (\phi_S \otimes \phi_S)_{\mathbf{1''}} \right]_{\mathbf{1}} & \to & (\varphi_2^2 + 2 \varphi_1 \varphi_3) (\varphi_3^2 + 2 \varphi_1 \varphi_2), \nonumber \\
 \left[(\phi_S \otimes \phi_S)_{\mathbf{3}_S} \otimes (\phi_S \otimes \phi_S)_{\mathbf{3}_S} \right]_{\mathbf{1}} & \to & \varphi_1^4 + 3 \varphi_2^2 \varphi_3^2 - 2 \varphi_1 (\varphi_2^3 + \varphi_3^3), \nonumber \\
 \left[ (\phi_S \otimes \phi_S)_{\mathbf{1}} \otimes u \right]_{\mathbf{1}} & \to & (\varphi_1^2 + 2 \varphi_2 \varphi_3) u, \nonumber \\
 \left[ (\phi_S \otimes \phi_S)_{\mathbf{1}} \otimes u \otimes u \right]_{\mathbf{1}} & \to & (\varphi_1^2 + 2 \varphi_2 \varphi_3) u^2, \nonumber \\
 \left( \left[(\phi_S \otimes \phi_S)_{\mathbf{3}_S} \otimes \phi_S  \right]_{\mathbf{1}} \otimes u \right)_{\mathbf{1}} & \to & (\varphi_1^3 + \varphi_2^3 + \varphi_3^3  - 3 \varphi_1 \varphi_2 \varphi_3) u .
 \label{eq:scal_combs}
\end{eqnarray}
Each such combination then receives an independent coefficient in the potential:
\begin{eqnarray}
 V &=& a_1 u + a_2 u^2 + b_2 (\varphi_1^2 + 2 \varphi_2 \varphi_3) + a_3 u^3 + b_3 (\varphi_1^2 + 2 \varphi_2 \varphi_3) u \nonumber\\
 && + c_3 (\varphi_1^3 + \varphi_2^3 + \varphi_3^3  - 3 \varphi_1 \varphi_2 \varphi_3) + a_4 u^4 + b_4 (\varphi_1^2 + 2 \varphi_2 \varphi_3) u^2 \nonumber \\
 && + c_4 (\varphi_1^3 + \varphi_2^3 + \varphi_3^3  - 3 \varphi_1 \varphi_2 \varphi_3) u + d_4 (\varphi_1^2 + 2 \varphi_2 \varphi_3)^2 + e_4 (\varphi_2^2 + 2 \varphi_1 \varphi_3) (\varphi_3^2 + 2 \varphi_1 \varphi_2)\nonumber \\
 && + f_4 [ \varphi_1^4 + 3 \varphi_2^2 \varphi_3^2 - 2 \varphi_1 (\varphi_2^3 + \varphi_3^3) ].
 \label{eq:scal_pot}
\end{eqnarray}
Obviously, scalar potentials have the inconvenient tendency to become very lengthy, which could be difficult to handle from a technical point of view. Furthermore, the unknown coefficients $a_1$, $a_2$, $b_2$, $a_3$, $b_3$, $c_3$, $a_4$, $b_4$, $c_4$, $d_4$, $e_4$, $f_4$ are hardly constrained, apart from the necessary requirement to lead to flavor symmetry breaking. However, this is exactly the property we can make use of here: by varying these coefficients, we can try to enforce the VEVs given in Eq.~\eqref{eq:ex_VEVs_abbr}, in order to use them in the model. It is important to understand that this is a necessary requirement for the model to be consistent, and it is for complicated potentials not always easy to achieve. On the other hand, it can be seen as a strength of models based on discrete symmetries to yield a framework for such a consistency check, which can lead to a useful constraint despite the arbitrariness present in the coefficients.

In order to check if the vacuum alignment as required is possible, we need to compute the partial derivatives of the potential in Eq.~\eqref{eq:scal_pot} with respect to the components $u$ and $\varphi_{1,2,3}$. Then, one has to insert the VEVs and check if they can really be simultaneous solutions of the resulting equations (even if either VEV is a valid solution, it may be that all VEVs together do not work consistently). Note that, in principle, one would also need to show that the VEVs do not only lead to an extremal value of the potential, but actually to a minimum (and ideally to the global minimum). However, a detailed discussion of these technicalities lies beyond the scope of this review. 

Computing the derivatives and inserting the VEVs leads to a system of four equations, two of which are identical:
\begin{eqnarray}
 \partial V / \partial u = 0 &\Rightarrow& a_1 + 2 a_2 \tilde \alpha_0 + 3 a_3 \tilde \alpha_0^2 + 3 b_3 \tilde \alpha_S^2 + 4 a_4 \tilde \alpha_0^3 + 6 b_4 \tilde \alpha_S^2 \tilde \alpha_0 = 0, \nonumber \\
 \partial V / \partial \varphi_1 = 0 &\Rightarrow& \tilde \alpha_S [b_2 \tilde \alpha_0 + b_4 \tilde \alpha_0^2 + 6 (d_4 + e_4) \tilde \alpha_S^2] = 0, \nonumber \\
  \partial V / \partial \varphi_{2,3} = 0 &\Rightarrow& \tilde \alpha_S [b_2 \tilde \alpha_0 + b_4 \tilde \alpha_0^2 + 2 (3 d_4 + 2 e_4) \tilde \alpha_S^2] = 0.
 \label{eq:scal_mins}
\end{eqnarray}
The question to answer is if these equations can be simultaneously solved with non-zero VEVs $\tilde \alpha_S$ and $\tilde \alpha_0$. From the last two equations, one can immediately see that this can only be the case if $e_4 = 0$. This may seem like a requirement that is easy to fulfill, but in fact it is a serious problem, since the corresponding term in the potential is a total singlet and there is no argument to simply set it to zero. One could try to modify the potential either by introducing more scalar fields (typically called \emph{waterfall fields} or \emph{driving fields} in case their only purpose is to make the desired vacuum alignment possible) or by adding higher order correction terms~\cite{Altarelli:2005yx}, both of which choices would make the potential much more complicated. Alternative strategies are to forbid certain terms in the potential by imposing supersymmetry~\cite{Altarelli:2005yx,Ma:2005qf,Altarelli:2008bg} or separating conflicting alignments by localizing them on different branes in an extra-dimensional setting~\cite{Altarelli:2005yp}.

Assuming there is a solution for this problem, setting $e_4 = 0$ implies that only the first two Eqs.~\eqref{eq:scal_mins} are actually independent, leaving us with two equations to determine two unknown VEVs. One could, for example, solve the first equation for $\tilde \alpha_S$ and insert the result in the second equation, leaving us with a polynomial of fourth degree in $\tilde \alpha_0$. In any case, the system of equations admits non-zero solutions for the VEVs whose values can be tuned by varying the coefficients in the potential. This shows that a suitable alignment can, at least in principle, be obtained in the model under consideration.

Similar procedures like the one explained here are often used in real models, partially supported by numerical calculations to get a feeling about whether the desired minimum of the potential is stable and can be obtained for reasonable choices of the parameters involved. As we had already pointed out on several occasions, the vacuum alignment problem can be technically very demanding, and sometimes it is close to impossible to solve. Nevertheless, good models should at least include an attempt so solve this problem, in order for the reader to judge about their validity.

A more in-depth discussion of the subjects can be found in several good reviews, see \emph{e.g.} Ref.~\refcite{King:2013eh} for the most recent one. The principles used in the literature for the determination of the vacuum alignment will hopefully be a bit more transparent with the help of this appendix.

\bibliographystyle{ws-ijmpd}
\bibliography{keV_Merle}

\end{document}